\newif\ifIEEE
\newif\ifPAGELIMIT
\IEEEtrue
\PAGELIMITfalse
\ifPAGELIMIT
\IEEEtrue
\fi
\ifIEEE
    \ifPAGELIMIT
        \documentclass[conference,final,letterpaper]{IEEEtran}
    \else   
        \documentclass[journal,final,letterpaper]{IEEEtran}
    \fi
    \newcommand{\bibauthor}[1]{#1}
    \newcommand{\bibpaper}[1]{``#1''}
    \newcommand{\Footnotetext}[2]
    {
        \begin{figure}[!b]
        \footnotesize\vspace{-3ex}\hrulefill\hfill
        \makebox[0em]{}\hfill\makebox[0em]{}%
                                          \par${}^{#1}$ #2\vspace{-.6ex}
        \end{figure}
        \addtocounter{figure}{0}
     }
\else
    \documentclass[12pt]{article}
    \newcommand{\bibauthor}[1]{\textsc{#1}}
    \newcommand{\bibpaper}[1]{\textsl{#1}}
    \newenvironment{IEEEkeywords}{\begin{small}%
                                  \textbf{Index Terms} ---}{\end{small}}
\fi
\usepackage{amsfonts,bm}
\usepackage{amsthm,amsmath}
\usepackage{pict2e}
\usepackage{color}
\newcommand{\bibbook}[1]{\textit{#1}}

\newcommand{\bibperiodical}[1]{\textit{#1}}

\ifPAGELIMIT
\fi
\newtheorem{theorem}{\indent Theorem}
\newtheorem{proposition}[theorem]{\indent Proposition}
\newtheorem{lemma}[theorem]{\indent Lemma}

\theoremstyle{remark}
\newtheorem{remark}{\indent Remark}
\theoremstyle{definition}
\newtheorem{example}{\indent Example}
\renewcommand{\mathbf}[1]{{\bm{#1}}}     % Use italics-bold in formulas
\newlength{\figunit}
\definecolor{gray}{rgb}{0.5,0.5,0.5}
\definecolor{brown}{rgb}{0.7,0.4,0.4}

\newcommand{\bldone}{{\mathbf{1}}}
\newcommand{\bldc}{{\mathbf{c}}}

\newcommand{\bldy}{{\mathbf{y}}}

\newcommand{\bldz}{{\mathbf{z}}}
\newcommand{\bldd}{{\mathbf{d}}}
\newcommand{\bldw}{{\mathbf{w}}}
\newcommand{\bldpi}{{\mathbf{\pi}}}
\newcommand{\bldvartheta}{{\mathbf{\vartheta}}}
\newcommand{\bldeta}{{\mathbf{\eta}}}
\newcommand{\bldzeta}{{\mathbf{\zeta}}}
\newcommand{\code}{{\mathcal{C}}}
\newcommand{\varcode}{{\mathsf{C}}}
\newcommand{\Code}{{\mathbb{C}}}
\newcommand{\entropy}{{\mathsf{H}}}
\newcommand{\Local}{{\mathcal{R}}}
\newcommand{\Ambient}{{\mathsf{A}}}
\newcommand{\Set}{{\mathcal{X}}}
\newcommand{\Subset}{{\mathsf{B}}}
\newcommand{\AltSet}{{\mathcal{S}}}
\newcommand{\Typical}{{\mathcal{T}}}
\newcommand{\YY}{{\mathcal{Y}}}
\newcommand{\InfSet}{{\mathcal{L}}}
\newcommand{\SupSet}{{\mathcal{U}}}
\newcommand{\Support}{{\mathsf{Supp}}}
\newcommand{\Integers}{{\mathbb{Z}}}
\newcommand{\Finitefield}{{\mathbb{F}}}
\newcommand{\Realfield}{{\mathbb{R}}}
\newcommand{\Rate}{{\mathsf{\textsl{R}}}}
\newcommand{\x}{s}

\newcommand{\LP}{{\mathrm{LP}}}
\newcommand{\CM}{{\mathrm{CM}}}
\newcommand{\SP}{{\mathrm{SP}}}
\newcommand{\LRC}{{\mathrm{LRC}}}
\newcommand{\opt}{{\mathrm{opt}}}
\newcommand{\Prob}{{\mathsf{Prob}}}
\newcommand{\Expected}{{\mathbb{E}}}
\newcommand{\weight}{{\mathsf{w}}}
\newcommand{\distance}{{\mathsf{d}}}
\newcommand{\bigcupdot}%
                 {{\textstyle{\bigcup\!\!\!\!\!\hspace{0.25ex}\cdot\;}}}
\ifPAGELIMIT
    \newcommand{\Theorem}{Thm.}
    \newcommand{\Theorems}{Thms.}
    \newcommand{\Proposition}{Prop.}
    
\else
    \newcommand{\Theorem}{Theorem}
    \newcommand{\Theorems}{Theorems}
    \newcommand{\Proposition}{Proposition}
    
\fi

\newcommand{\Title}{Asymptotic Bounds on the Rate of
                                               Locally Repairable Codes}
\newcommand{\Namea}{Ron M. Roth}
\newcommand{\Addressa}{Computer Science Department}
\newcommand{\Addressatwo}{Technion, Haifa 3200003, Israel}

\newcommand{\Emaila}{ronny@cs.technion.ac.il}
\newcommand{\Grant}{This work was supported by 
                    Grants~1396/16 and 1713/20 from
                    the Israel Science Foundation.}
\newcommand{\Thnxa}{\Namea\ is with the \Addressa, \Addressatwo.
                    Email: \Emaila}

\begin{document}
\ifIEEE
    \title{\Title}
       \ifPAGELIMIT
           \author{\IEEEauthorblockN{\Namea\vspace{-1ex}}\\
                   \IEEEauthorblockA{\Addressa,
                                     \Addressatwo\\
                                     \Emaila\vspace{-2ex}}
           }
       \else
           \author{\Namea%%\Membershipa
                   \thanks{\Grant}
                   \thanks{\Thnxa}
           }
       \fi
\else
    \title{\textbf{\Title}\thanks{\Grant}}
    \author{\textsc{\Namea}\thanks{\Thnxa}}
    \date{}
\fi
\maketitle
% \pagestyle{empty} alone does not remove the numbering of the title
%  page in the two-column mode
%%\ifIEEE  \thispagestyle{empty}  \fi

\begin{abstract}
New asymptotic upper bounds are presented on the rate
of sequences of locally repairable codes (LRCs)
with a prescribed relative minimum distance and locality
over a finite field $F$.
The bounds apply to LRCs in which
the recovery functions are linear;
in particular, the bounds apply to linear LRCs over $F$.
The new bounds are shown to improve on previously published
results, especially when the repair groups are
\ifPAGELIMIT
    disjoint.
\else
disjoint,
namely, they form a partition of the set of coordinates.
\fi
\end{abstract}

\ifPAGELIMIT
\else
\fi

\ifPAGELIMIT
    \Footnotetext{\quad}{\Grant}
\fi

\section{Introduction}
\label{sec:introduction}

Hereafter, we fix $F$ to be a finite field
\ifPAGELIMIT
    $\Finitefield_q$.
\else
$\Finitefield_q$ (of size $q$).
\fi
For a positive integer $a$, we denote by $[a]$ the integer set
$\{ 1, 2, \ldots, a \}$.
For a word (vector) $\bldy \in F^N$
and a nonempty subset $\Local \subseteq [N]$, we let $(\bldy)_\Local$
denote the sub-word of $\bldy$ that is indexed by $\Local$.

An $(N,M,d)$ code of length over $F$ is
a nonempty subset $\varcode \subseteq F^N$ of size $M$
and minimum (Hamming) distance $d$.
The rate of the code is $R = (\log_q M)/N$
and its relative minimum distance (r.m.d.) is $d/N$.
When $\varcode$ is linear over $F$ we will use the standard notation
$[N,k,d]$ where $k = \log_q M$ is the dimension of $\varcode$.
For a nonempty subset $\Local \subseteq [N]$,
we let $(\varcode)_\Local$ be the punctured code
$\{ (\bldc)_\Local \,:\, \bldc \in \varcode \}$.

We say that an $(N,M,d)$ code $\varcode$ over $F$ is
a \emph{locally repairable code} (in short, LRC) with locality $r$
if for every coordinate $j \in [N]$
there is a subset
\[
\Local_j = \Local_j^* \; \bigcupdot \; \{ j \} \subseteq [N]
\]
of size at most $n = r+1$ (that contains $j$)
such that for every codeword $\bldc = (c_t)_{t \in [N]}$,
the value $c_j$ is uniquely determined from
$(\bldc)_{\Local_j^*}$. In other words,
there is a \emph{recovery function}
$\varphi_j : F^{|\Local_j^*|} \rightarrow F$ such that
\ifPAGELIMIT
    $c_j = \varphi_j \bigl( (\bldc)_{\Local_j^*} \bigr)$.
\else
\begin{equation}
\label{eq:recoveryfunction}
c_j = \varphi_j \left( (\bldc)_{\Local_j^*} \right) .
\end{equation}
\fi
The set $\Local_j$ is called the \emph{repair group}
of $j$ and $(\varcode)_{\Local_j}$ is the respective constituent code.
\ifPAGELIMIT
\else
Clearly, the constituent code $(\varcode)_{\Local_j}$ cannot
contain two codewords that differ only on position $j$.
\fi
%%Moreover, if $\Local_j$ is a minimal repair group for $j$
%%(in the sense that no proper subset of it is a repair group for $j$)
%%then the minimum distance of $(\varcode)_{\Local_j}$
%%is at least $2$ when $q = 2$ or when $\varcode$ is linear.
Repair groups will usually be represented as
a list $(\Local_j)_{j=1}^N$, where each repair group is indexed by
the coordinate $j \in [N]$ that it corresponds to.\footnote{%
\label{footnote:locality}%
Occasionally, however, we will refer to the set of \emph{distinct}
repair groups within this list; we will then use
the notation $\{ \Local_j \}_{j \in [\ell]}$,
where $\ell$ is the number of such repair
\ifPAGELIMIT
    groups.
\else
groups (this notation implicitly assumes that the first $\ell$
coordinates of $\varcode$ are associated with distinct repair groups).
Note that a list of repair groups of an LRC $\varcode$---and, therefore,
the set of distinct repair groups in such a list---may not be
unique. Yet such a list (or set) always covers
all the coordinates of $\varcode$.%
\fi}

When the recovery functions $\varphi_j$ are
\ifPAGELIMIT
    all linear over $F$,
\else
linear over $F$ for all $j \in [N]$,
\fi
we say that the LRC $\varcode$ is \emph{linearly recoverable}.
In this case, each constituent code $(\varcode)_{\Local_j}$ is
a subcode of
a linear $[|\Local_j|,|\Local_j|{-}1]$ code over $F$.
If $\varcode$ is a linear code over $F$ then
it is also linearly recoverable~\cite[Lemma~10]{ABHMT}.
The LRC $\varcode$ is said to be \emph{all-disjoint}
if it has a list of repair groups $(\Local_j)_{j=1}^N$
\ifPAGELIMIT
    such that the set of distinct repair groups forms
    a partition of $[N]$.
\else
that satisfies
$\Local_j \cap \Local_{j'} \in \{ \emptyset, \Local_j, \Local_{j'} \}$
for all $j, j' \in [N]$. By possibly adding dummy elements
to repair groups, this definition is equivalent to requiring
that no two distinct repair groups intersect; in this case,
the set of distinct repair groups $\{ \Local_j \}_j$
forms a partition of $[N]$.
In the all-disjoint case, each repair group $\Local_j$ is also
a repair group for all $j' \in \Local_j$ and, so,
each constituent code has minimum distance${} \ge 2$.
\fi

The general study of LRCs was initiated in~\cite{GHSY} and~\cite{PD},
and has generated an extensive body of literature since,
including constructions of LRCs, bounds on their parameters,
and studies of additional attributes, such as availability
and local minimum distance of the constituent codes; see
\cite{ABHMT},
\cite{BK},
\cite{Blaum},
\cite{CM},
\cite{GFY},
\cite{GFWH},
\cite{GJX},
\cite{HYS},
\cite{HYUS},
\cite{MG},
\cite{PHO},
\cite{PKLK},
\cite{TB},
\cite{WZ},
\cite{WZL}.

In this work, we will be interested in asymptotic upper bounds on
the rate of linearly recoverable LRCs with prescribed r.m.d.\ and
locality, as the code length tends to
\ifPAGELIMIT
    infinity.
\else
infinity.\footnote{%
While our results are stated for linearly recoverable LRCs,
we shall only use the fact
that each recovery function $\varphi_j$ in~(\ref{eq:recoveryfunction})
preserves the addition of the field $F$.
Hence, our results apply more generally to the case where $F$
is a finite Abelian group and each $\varphi_j$
is a homomorphism from $F^{|\Local_j^*|}$ to $F$.}
\fi

Hereafter, by an infinite sequence of codes over $F$
we mean a sequence $(\varcode_i)_{i=1}^\infty$,
where each $\varcode_i$ is an $(N_i,M_i,d_i)$ code over $F$
and the length sequence $(N_i)_{i=1}^\infty$ is strictly increasing.
The rate and the r.m.d.\ of the sequence are defined, respectively,
as $R = \varlimsup_{i \rightarrow \infty} R_i$
and $\delta = \varliminf_{i \rightarrow \infty} d_i/N_i$.
The supremum over all the rates of sequences of codes
with r.m.d.${} \ge \delta$ will be denoted by $\Rate_\opt(\delta)$.
Given $\omega \in [0,1]$, we let $\Rate_\opt(\delta,\omega)$
be such a supremum for sequences $(\varcode_i)_{i=1}^\infty$
of constant-weight codes, with the codewords of each $\varcode_i$
all having the same weight $w_i$, such that
$\lim_{i \rightarrow \infty} w_i/N_i = \omega$.
There are several known upper bounds on
$\Rate_\opt(\delta)$ and $\Rate_\opt(\delta,\omega)$,
and the best known are obtained using
the linear-programming method
\cite{Aaltonen},
\cite{BHL1},
\cite{MRRW}.
In particular, it is known that
$\Rate_\opt(\delta,\omega) = \Rate_\opt(\delta) = 0$
when $\delta \ge (q{-}1)/q$;
hence, we will implicitly assume hereafter that $\delta < (q{-}1)/q$.
\ifPAGELIMIT
\else
It is also easy to see that $\Rate_\opt(\delta,\omega) = 0$
when $\omega < \delta/2$.
\fi
The upper bounds on the rate of LRCs which are considered
in this work will often be expressions that involve
$\Rate_\opt(\delta)$ or $\Rate_\opt(\delta,\omega)$;
concrete upper bounds for LRCs can then be obtained
by replacing these terms with any upper bound on them.

A \emph{$(\delta,n)$-LRC sequence} over $F$
is an infinite sequence of LRCs $(\varcode_i)_{i=1}^\infty$ over $F$
with r.m.d.${} \ge \delta$
where each LRC $\varcode_i$ has locality $r = n-1$.
The sequence is said to be all-disjoint
if each $\varcode_i$ is all-disjoint,
and it is linearly recoverable (respectively, linear)
if so is each $\varcode_i$.
It follows from the results of~\cite{GHSY} and~\cite{PD}
that the rate $R$ of any $(\delta,n)$-LRC sequence
is bounded from above by
\begin{equation}
\label{eq:Singleton}
R \le \frac{n{-}1}{n} \cdot (1 - \delta) .
\end{equation}
This bound, which is oblivious to the field size,
\ifPAGELIMIT
\else
can be seen as the LRC counterpart of the Singleton bound, and it
\fi
applies generally to $(\delta,n)$-LRC sequences
(which are not necessarily all-disjoint or linearly recoverable).
The bound~(\ref{eq:Singleton}) was improved
in~\cite[\Theorem~1]{CM} to
$R \le \Rate_\CM(\delta,n)$, where
\begin{equation}
\label{eq:CM}
\Rate_\CM(\delta,n)
= \min_{\tau \in [0,1{-}\delta]}
\left\{ 
\tau \cdot \frac{n{-}1}{n}
+ (1{-}\tau) \cdot \Rate_\opt \left( \frac{\delta}{1{-}\tau} \right)
\right\} .
\end{equation}
This bound, which depends on the field size,
coincides with~(\ref{eq:Singleton})
when $\delta \mapsto \Rate_\opt(\delta)$ is taken to be
the Singleton bound $\delta \mapsto 1 - \delta$.
Small improvements to~(\ref{eq:CM})
were obtained
\ifPAGELIMIT
    in~\cite{ABHMT}.
\else
in~\cite{ABHMT}; the latter paper considered
the more general setting where the constituent codes
$(\varcode)_{\Local_j}$ have a prescribed minimum distance $\rho \ge 2$ 
(see also \cite{GFWH}).
\fi
For non-asymptotic improvements, see~\cite{WZL}.

For reference, we also mention the Gilbert--Varshamov-type
\emph{lower} bound on the largest attainable rate of LRC sequences:
it was shown in~\cite[\Theorem~2]{CM} and~\cite[\Theorem~B]{TBF}
that there exists
an all-disjoint linear $(\delta,R)$-LRC sequence of rate
\begin{equation}
\label{eq:R0}
R \ge \Rate_0(\delta,n) = \frac{n{-}1}{n}
- \lambda \left( \delta,n \right) ,
\end{equation}
where $\lambda(\omega,n)$ is defined
for every $n \in \Integers^+$ and $\omega \in \Realfield_{\ge 0}$ by
\begin{eqnarray}
\lefteqn{
\lambda(\omega,n)
= \lambda_q(\omega,n)
= \inf_{z \in (0,1]} 
\Bigl\{
-\omega \log_q z - \frac{1}{n}
} \makebox[0ex]{} \nonumber \\
\label{eq:lambda}
&& {} + \frac{1}{n}
\log_q \Bigl( (1 + (q{-}1)z)^n + (q{-}1)(1{-}z)^n \Bigr)
\Bigr\} .
\end{eqnarray}
The function $\omega \mapsto \Rate_0(\omega,n)$ will play a role in
our results as well. The expression $\lambda(\omega,n)$
is the growth rate of the volume of a ball
of radius $\omega N$ in the subspace of $F^N$
formed by the Cartesian product of copies of the $[n,n{-}1,2]$
parity code over $F$, as
\ifPAGELIMIT
    $N \rightarrow \infty$.
\else
$N \rightarrow \infty$.\footnote{%
The proof of the lower bound~(\ref{eq:R0}) only requires
$\lambda(\omega,n)$ to be an upper bound
on this growth rate---a relation which follows from
the Chernoff bound as in Eq.~(\ref{eq:chernoff}) below. 
On the other hand, for our results, we will need
a lower bound on this growth rate; to this end, we will use
a stronger result from large deviations theory (Eq.~(\ref{eq:cramer})).}
The function $\omega \mapsto \lambda_q(\omega,n)$ is drawn
in Figure~\ref{fig:lambda} for $(q,n) = (2,4), (2,5), (4,5)$;
it will follow from our analysis
that it is continuous, increasing, and concave
on $[0^+,\infty]$ and its values range from $0$ at $\omega = 0$
to $(n{-}1)/n$ at $\omega \ge (q{-}1)/q$.
As we show in Lemma~\ref{lem:unique} in
Appendix~\ref{sec:R0properties},
finding the infimum in~(\ref{eq:lambda}) is computationally easy:
it amounts to finding a root of a certain real polynomial
of degree${} \le n$.
\newcommand{\PlotLambda}[1]{%
    \put(000,000){#1}
    \put(-10,000){\vector(1,0){120}}
    \put(114,000){\makebox(0,0)[l]{$\omega$}}
    \put(000,-10){\vector(0,1){120}}
    \put(-05,-04){\makebox(0,0)[t]{$0$}}
    \put(100,000){\line(0,-1){2}}
    \put(100,-04){\makebox(0,0)[t]{$1.0$}}
    \put(000,100){\line(-1,0){2}}
    \put(-04,100){\makebox(0,0)[r]{$1.0$}}
}%
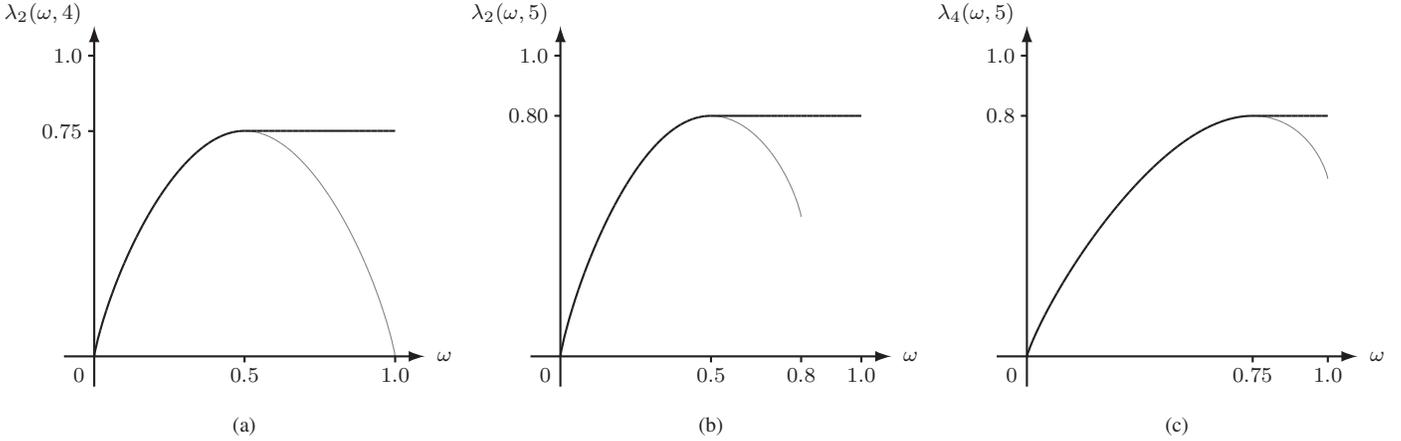
\begin{figure*}[hbt]
\begin{center}
\footnotesize
\thicklines
\ifIEEE
    \setlength{\unitlength}{0.4mm}
\else
    \setlength{\unitlength}{0.35mm}
\fi
\begin{picture}(440,150)(-10,-25)
\put(000,000){\PlotLambda{%
    \put(050,000){\line(0,-1){2}}
    \put(050,-04){\makebox(0,0)[t]{$0.5$}}
    \put(000,075){\line(-1,0){2}}
    \put(-04,075){\makebox(0,0)[r]{$0.75$}}
    \put(-04,114){\makebox(0,0)[r]{$\lambda_2(\omega,4)$}}
    \put(050,-20){\makebox(0,0)[t]{(a)}}
    {\thinlines\color{gray}
    	\qbezier(50.000,75.000)(51.002,75.000)(52.000,74.885)
    	\qbezier(52.000,74.885)(52.998,74.770)(54.000,74.539)
    	\qbezier(54.000,74.539)(54.998,74.310)(56.000,73.966)
    	\qbezier(56.000,73.966)(56.997,73.624)(58.000,73.168)
    	\qbezier(58.000,73.168)(58.997,72.715)(60.000,72.149)
    	\qbezier(60.000,72.149)(60.997,71.586)(62.000,70.912)
    	\qbezier(62.000,70.912)(62.997,70.242)(64.000,69.463)
    	\qbezier(64.000,69.463)(64.998,68.687)(66.000,67.803)
    	\qbezier(66.000,67.803)(66.999,66.923)(68.000,65.937)
    	\qbezier(68.000,65.937)(69.000,64.952)(70.000,63.864)
    	\qbezier(70.000,63.864)(71.001,62.776)(72.000,61.586)
    	\qbezier(72.000,61.586)(73.002,60.393)(74.000,59.102)
    	\qbezier(74.000,59.102)(75.003,57.804)(76.000,56.409)
    	\qbezier(76.000,56.409)(77.005,55.004)(78.000,53.504)
    	\qbezier(78.000,53.504)(79.007,51.987)(80.000,50.380)
    	\qbezier(80.000,50.380)(81.009,48.748)(82.000,47.028)
    	\qbezier(82.000,47.028)(83.011,45.274)(84.000,43.436)
    	\qbezier(84.000,43.436)(85.014,41.551)(86.000,39.588)
    	\qbezier(86.000,39.588)(87.018,37.561)(88.000,35.460)
    	\qbezier(88.000,35.460)(89.023,33.274)(90.000,31.022)
    	\qbezier(90.000,31.022)(91.030,28.650)(92.000,26.228)
    	\qbezier(92.000,26.228)(93.041,23.629)(94.000,21.005)
    	\qbezier(94.000,21.005)(95.061,18.104)(96.000,15.231)
	\qbezier(96.000,15.231)(97.108,11.843)(98.000,8.644)
    	\qbezier(98.000,8.644)(99.678,2.625)(100.000,0.000)
    }
    \qbezier(0.000,0.000)(0.322,2.625)(2.000,8.644)
    \qbezier(2.000,8.644)(2.892,11.843)(4.000,15.231)
    \qbezier(4.000,15.231)(4.939,18.104)(6.000,21.005)
    \qbezier(6.000,21.005)(6.959,23.629)(8.000,26.228)
    \qbezier(8.000,26.228)(8.970,28.650)(10.000,31.022)
    \qbezier(10.000,31.022)(10.977,33.274)(12.000,35.460)
    \qbezier(12.000,35.460)(12.982,37.561)(14.000,39.588)
    \qbezier(14.000,39.588)(14.986,41.551)(16.000,43.436)
    \qbezier(16.000,43.436)(16.989,45.274)(18.000,47.028)
    \qbezier(18.000,47.028)(18.991,48.748)(20.000,50.380)
    \qbezier(20.000,50.380)(20.993,51.987)(22.000,53.504)
    \qbezier(22.000,53.504)(22.995,55.004)(24.000,56.409)
    \qbezier(24.000,56.409)(24.997,57.804)(26.000,59.102)
    \qbezier(26.000,59.102)(26.998,60.393)(28.000,61.586)
    \qbezier(28.000,61.586)(28.999,62.776)(30.000,63.864)
    \qbezier(30.000,63.864)(31.000,64.952)(32.000,65.937)
    \qbezier(32.000,65.937)(33.001,66.923)(34.000,67.803)
    \qbezier(34.000,67.803)(35.002,68.687)(36.000,69.463)
    \qbezier(36.000,69.463)(37.003,70.242)(38.000,70.912)
    \qbezier(38.000,70.912)(39.003,71.586)(40.000,72.149)
    \qbezier(40.000,72.149)(41.003,72.715)(42.000,73.168)
    \qbezier(42.000,73.168)(43.003,73.624)(44.000,73.966)
    \qbezier(44.000,73.966)(45.002,74.310)(46.000,74.539)
    \qbezier(46.000,74.539)(47.002,74.770)(48.000,74.885)
    \qbezier(48.000,74.885)(48.999,75.000)(50.000,75.000)
    \qbezier(50.000,75.000)(50.000,75.000)(52.000,75.000)
    \qbezier(52.000,75.000)(52.000,75.000)(54.000,75.000)
    \qbezier(54.000,75.000)(54.000,75.000)(56.000,75.000)
    \qbezier(56.000,75.000)(56.000,75.000)(58.000,75.000)
    \qbezier(58.000,75.000)(58.000,75.000)(60.000,75.000)
    \qbezier(60.000,75.000)(60.000,75.000)(62.000,75.000)
    \qbezier(62.000,75.000)(62.000,75.000)(64.000,75.000)
    \qbezier(64.000,75.000)(64.000,75.000)(66.000,75.000)
    \qbezier(66.000,75.000)(66.000,75.000)(68.000,75.000)
    \qbezier(68.000,75.000)(68.000,75.000)(70.000,75.000)
    \qbezier(70.000,75.000)(70.000,75.000)(72.000,75.000)
    \qbezier(72.000,75.000)(72.000,75.000)(74.000,75.000)
    \qbezier(74.000,75.000)(74.000,75.000)(76.000,75.000)
    \qbezier(76.000,75.000)(76.000,75.000)(78.000,75.000)
    \qbezier(78.000,75.000)(78.000,75.000)(80.000,75.000)
    \qbezier(80.000,75.000)(80.000,75.000)(82.000,75.000)
    \qbezier(82.000,75.000)(82.000,75.000)(84.000,75.000)
    \qbezier(84.000,75.000)(84.000,75.000)(86.000,75.000)
    \qbezier(86.000,75.000)(86.000,75.000)(88.000,75.000)
    \qbezier(88.000,75.000)(88.000,75.000)(90.000,75.000)
    \qbezier(90.000,75.000)(90.000,75.000)(92.000,75.000)
    \qbezier(92.000,75.000)(92.000,75.000)(94.000,75.000)
    \qbezier(94.000,75.000)(94.000,75.000)(96.000,75.000)
    \qbezier(96.000,75.000)(96.000,75.000)(98.000,75.000)
    \qbezier(98.000,75.000)(100.000,75.000)(100.000,75.000)
}}
\put(155,000){\PlotLambda{%
    \put(050,000){\line(0,-1){2}}
    \put(050,-04){\makebox(0,0)[t]{$0.5$}}
    \put(080,000){\line(0,-1){2}}
    \put(080,-04){\makebox(0,0)[t]{$0.8$}}
    \put(000,080){\line(-1,0){2}}
    \put(-04,080){\makebox(0,0)[r]{$0.80$}}
    \put(-04,114){\makebox(0,0)[r]{$\lambda_2(\omega,5)$}}
    \put(050,-20){\makebox(0,0)[t]{(b)}}
    {\thinlines\color{gray}
    	\qbezier(50.000,80.000)(51.003,80.000)(52.000,79.885)
    	\qbezier(52.000,79.885)(53.001,79.769)(54.000,79.538)
    	\qbezier(54.000,79.538)(55.002,79.306)(56.000,78.958)
    	\qbezier(56.000,78.958)(57.003,78.608)(58.000,78.142)
    	\qbezier(58.000,78.142)(59.005,77.672)(60.000,77.086)
    	\qbezier(60.000,77.086)(61.007,76.492)(62.000,75.780)
    	\qbezier(62.000,75.780)(63.009,75.057)(64.000,74.216)
    	\qbezier(64.000,74.216)(65.012,73.356)(66.000,72.377)
    	\qbezier(66.000,72.377)(67.016,71.369)(68.000,70.242)
    	\qbezier(68.000,70.242)(69.021,69.071)(70.000,67.780)
    	\qbezier(70.000,67.780)(71.028,66.424)(72.000,64.947)
    	\qbezier(72.000,64.947)(73.039,63.368)(74.000,61.673)
    	\qbezier(74.000,61.673)(75.059,59.804)(76.000,57.834)
    	\qbezier(76.000,57.834)(77.107,55.518)(78.000,53.170)
    	\qbezier(78.000,53.170)(79.677,48.762)(80.000,46.439)
    }
    \qbezier(0.000,0.000)(0.323,2.698)(2.000,9.054)
    \qbezier(2.000,9.054)(2.893,12.438)(4.000,16.040)
    \qbezier(4.000,16.040)(4.941,19.101)(6.000,22.200)
    \qbezier(6.000,22.200)(6.961,25.011)(8.000,27.797)
    \qbezier(8.000,27.797)(8.972,30.401)(10.000,32.951)
    \qbezier(10.000,32.951)(10.979,35.379)(12.000,37.735)
    \qbezier(12.000,37.735)(12.984,40.005)(14.000,42.192)
    \qbezier(14.000,42.192)(14.988,44.317)(16.000,46.353)
    \qbezier(16.000,46.353)(16.991,48.344)(18.000,50.239)
    \qbezier(18.000,50.239)(18.993,52.103)(20.000,53.866)
    \qbezier(20.000,53.866)(20.995,55.608)(22.000,57.245)
    \qbezier(22.000,57.245)(22.997,58.868)(24.000,60.383)
    \qbezier(24.000,60.383)(24.998,61.889)(26.000,63.284)
    \qbezier(26.000,63.284)(26.999,64.675)(28.000,65.953)
    \qbezier(28.000,65.953)(29.000,67.229)(30.000,68.390)
    \qbezier(30.000,68.390)(31.000,69.552)(32.000,70.597)
    \qbezier(32.000,70.597)(33.001,71.643)(34.000,72.572)
    \qbezier(34.000,72.572)(35.001,73.502)(36.000,74.316)
    \qbezier(36.000,74.316)(37.000,75.130)(38.000,75.826)
    \qbezier(38.000,75.826)(39.000,76.524)(40.000,77.104)
    \qbezier(40.000,77.104)(41.000,77.685)(42.000,78.148)
    \qbezier(42.000,78.148)(43.000,78.612)(44.000,78.959)
    \qbezier(44.000,78.959)(44.999,79.307)(46.000,79.538)
    \qbezier(46.000,79.538)(46.999,79.769)(48.000,79.885)
    \qbezier(48.000,79.885)(48.999,80.000)(50.000,80.000)
    \qbezier(50.000,80.000)(50.000,80.000)(52.000,80.000)
    \qbezier(52.000,80.000)(52.000,80.000)(54.000,80.000)
    \qbezier(54.000,80.000)(54.000,80.000)(56.000,80.000)
    \qbezier(56.000,80.000)(56.000,80.000)(58.000,80.000)
    \qbezier(58.000,80.000)(58.000,80.000)(60.000,80.000)
    \qbezier(60.000,80.000)(60.000,80.000)(62.000,80.000)
    \qbezier(62.000,80.000)(62.000,80.000)(64.000,80.000)
    \qbezier(64.000,80.000)(64.000,80.000)(66.000,80.000)
    \qbezier(66.000,80.000)(66.000,80.000)(68.000,80.000)
    \qbezier(68.000,80.000)(68.000,80.000)(70.000,80.000)
    \qbezier(70.000,80.000)(70.000,80.000)(72.000,80.000)
    \qbezier(72.000,80.000)(72.000,80.000)(74.000,80.000)
    \qbezier(74.000,80.000)(74.000,80.000)(76.000,80.000)
    \qbezier(76.000,80.000)(76.000,80.000)(78.000,80.000)
    \qbezier(78.000,80.000)(78.000,80.000)(80.000,80.000)
    \qbezier(80.000,80.000)(80.000,80.000)(82.000,80.000)
    \qbezier(82.000,80.000)(82.000,80.000)(84.000,80.000)
    \qbezier(84.000,80.000)(84.000,80.000)(86.000,80.000)
    \qbezier(86.000,80.000)(86.000,80.000)(88.000,80.000)
    \qbezier(88.000,80.000)(88.000,80.000)(90.000,80.000)
    \qbezier(90.000,80.000)(90.000,80.000)(92.000,80.000)
    \qbezier(92.000,80.000)(92.000,80.000)(94.000,80.000)
    \qbezier(94.000,80.000)(94.000,80.000)(96.000,80.000)
    \qbezier(96.000,80.000)(96.000,80.000)(98.000,80.000)
    \qbezier(98.000,80.000)(98.000,80.000)(100.000,80.000)
}}
\put(310,000){\PlotLambda{%
    \put(075,000){\line(0,-1){2}}
    \put(075,-04){\makebox(0,0)[t]{$0.75$}}
    \put(000,080){\line(-1,0){2}}
    \put(-04,080){\makebox(0,0)[r]{$0.8$}}
    \put(-04,114){\makebox(0,0)[r]{$\lambda_4(\omega,5)$}}
    \put(050,-20){\makebox(0,0)[t]{(c)}}
    {\thinlines\color{gray}
    	\qbezier(75.000,80.000)(76.054,80.000)(77.083,79.915)
    	\qbezier(77.083,79.915)(78.137,79.828)(79.167,79.652)
    	\qbezier(79.167,79.652)(80.222,79.472)(81.250,79.200)
    	\qbezier(81.250,79.200)(82.308,78.919)(83.333,78.540)
    	\qbezier(83.333,78.540)(84.394,78.149)(85.417,77.655)
    	\qbezier(85.417,77.655)(86.482,77.139)(87.500,76.515)
    	\qbezier(87.500,76.515)(88.570,75.858)(89.583,75.084)
    	\qbezier(89.583,75.084)(90.660,74.262)(91.667,73.311)
    	\qbezier(91.667,73.311)(92.755,72.284)(93.750,71.117)
    	\qbezier(93.750,71.117)(94.858,69.818)(95.833,68.368)
    	\qbezier(95.833,68.368)(96.990,66.649)(97.917,64.785)
    	\qbezier(97.917,64.785)(99.667,61.266)(100.000,59.069)
    }
    \qbezier(0.000,0.000)(0.316,1.447)(2.000,5.376)
    \qbezier(2.000,5.376)(2.889,7.451)(4.000,9.766)
    \qbezier(4.000,9.766)(4.937,11.720)(6.000,13.775)
    \qbezier(6.000,13.775)(6.958,15.627)(8.000,17.529)
    \qbezier(8.000,17.529)(8.969,19.297)(10.000,21.088)
    \qbezier(10.000,21.088)(10.976,22.783)(12.000,24.485)
    \qbezier(12.000,24.485)(12.981,26.115)(14.000,27.742)
    \qbezier(14.000,27.742)(14.985,29.315)(16.000,30.876)
    \qbezier(16.000,30.876)(16.988,32.394)(18.000,33.896)
    \qbezier(18.000,33.896)(18.990,35.364)(20.000,36.811)
    \qbezier(20.000,36.811)(20.992,38.231)(22.000,39.626)
    \qbezier(22.000,39.626)(22.994,41.002)(24.000,42.348)
    \qbezier(24.000,42.348)(24.995,43.679)(26.000,44.979)
    \qbezier(26.000,44.979)(26.997,46.267)(28.000,47.521)
    \qbezier(28.000,47.521)(28.998,48.767)(30.000,49.977)
    \qbezier(30.000,49.977)(30.999,51.182)(32.000,52.347)
    \qbezier(32.000,52.347)(33.000,53.511)(34.000,54.633)
    \qbezier(34.000,54.633)(35.001,55.755)(36.000,56.834)
    \qbezier(36.000,56.834)(37.002,57.915)(38.000,58.949)
    \qbezier(38.000,58.949)(39.002,59.988)(40.000,60.979)
    \qbezier(40.000,60.979)(41.003,61.976)(42.000,62.922)
    \qbezier(42.000,62.922)(43.004,63.875)(44.000,64.776)
    \qbezier(44.000,64.776)(45.004,65.684)(46.000,66.539)
    \qbezier(46.000,66.539)(47.005,67.401)(48.000,68.209)
    \qbezier(48.000,68.209)(49.005,69.024)(50.000,69.783)
    \qbezier(50.000,69.783)(51.005,70.550)(52.000,71.259)
    \qbezier(52.000,71.259)(53.005,71.975)(54.000,72.632)
    \qbezier(54.000,72.632)(55.005,73.297)(56.000,73.901)
    \qbezier(56.000,73.901)(57.005,74.512)(58.000,75.062)
    \qbezier(58.000,75.062)(59.005,75.617)(60.000,76.110)
    \qbezier(60.000,76.110)(61.005,76.608)(62.000,77.043)
    \qbezier(62.000,77.043)(63.005,77.482)(64.000,77.856)
    \qbezier(64.000,77.856)(65.006,78.234)(66.000,78.546)
    \qbezier(66.000,78.546)(67.006,78.861)(68.000,79.108)
    \qbezier(68.000,79.108)(69.006,79.358)(70.000,79.538)
    \qbezier(70.000,79.538)(71.007,79.721)(72.000,79.831)
    \qbezier(72.000,79.831)(73.008,79.943)(74.000,79.981)
    \qbezier(74.000,79.981)(74.502,80.000)(76.000,80.000)
    \qbezier(76.000,80.000)(76.000,80.000)(78.000,80.000)
    \qbezier(78.000,80.000)(78.000,80.000)(80.000,80.000)
    \qbezier(80.000,80.000)(80.000,80.000)(82.000,80.000)
    \qbezier(82.000,80.000)(82.000,80.000)(84.000,80.000)
    \qbezier(84.000,80.000)(84.000,80.000)(86.000,80.000)
    \qbezier(86.000,80.000)(86.000,80.000)(88.000,80.000)
    \qbezier(88.000,80.000)(88.000,80.000)(90.000,80.000)
    \qbezier(90.000,80.000)(90.000,80.000)(92.000,80.000)
    \qbezier(92.000,80.000)(92.000,80.000)(94.000,80.000)
    \qbezier(94.000,80.000)(94.000,80.000)(96.000,80.000)
    \qbezier(96.000,80.000)(96.000,80.000)(98.000,80.000)
    \qbezier(98.000,80.000)(100.000,80.000)(100.000,80.000)
}}
\end{picture}
\thinlines
\setlength{\unitlength}{1pt}
\end{center}
\caption{Function $\omega \mapsto \lambda_q(\omega,n)$
for (a) $(q,n) = (2,4)$, (b) $(q,n) = (2,5)$,
and (c) $(q,n) = (4,5)$. The lighter curves
depict the function $\omega \mapsto \lambda^*_q(\omega,n)$.}
\label{fig:lambda}
\end{figure*}
\fi

Our first set of results pertain to linearly recoverable LRC sequences
that are all-disjoint.
We prove the next bound-enhancement theorem by using, \emph{inter alia},
the generalization of~\cite{LL} to the shortening method for
improving upper bounds on the rate of code sequences.

\begin{theorem}
\label{thm:bound1}
Let $\delta \mapsto \Rate_\LRC(\delta, n)$ be an upper bound on
the rate of any
all-disjoint linearly recoverable $(\delta,n)$-LRC sequence over $F$.
Then the rate of such a sequence is bounded from above also by
\ifPAGELIMIT
    \begin{eqnarray}
    \Rate_1(\delta,n)
    & = & \inf_{\tau \in (0,1)} \min_{(\theta,\theta')}
    \Bigl\{ 
    \tau \cdot \Rate_0(\theta/2,n)
    \nonumber \\
    \label{eq:bound1}
    &&
    \quad \quad \quad
    {}
    + 
    (1-\tau) \cdot \Rate_\LRC(\theta',n)
    \Bigr\} ,
\end{eqnarray}
\else
\begin{eqnarray}
\Rate_1(\delta,n)
& = & \inf_{\tau \in (0,1)} \min_{(\theta,\theta')}
\Biggl\{ 
\tau \cdot \Rate_0 \left( \frac{\theta}{2}, n \right)
\nonumber \\
\label{eq:bound1}
&&
\quad \quad \quad
{}
+ 
(1-\tau) \cdot \Rate_\LRC(\theta',n)
\Biggr\} ,
\end{eqnarray}
\fi
where $\Rate_0(\cdot,n)$ is defined in~(\ref{eq:R0})--(\ref{eq:lambda})
and the (inner) minimum is taken over
all pairs $(\theta,\theta')$ in $[0,(q{-}1)/q]^2$ such that
\begin{equation}
\label{eq:theta}
\tau \cdot \theta + (1-\tau) \cdot \theta' = \delta .
\end{equation}
\end{theorem}

In particular, \Theorem~\ref{thm:bound1} holds for
\ifPAGELIMIT
    $\Rate_\LRC(\delta,n) = \Rate_\opt(\delta)$.
\else
$\Rate_\LRC(\delta,n) = \Rate_\opt(\delta)$
(i.e., ignoring locality or linear recoverability).
\fi
This, in turn, yields a concrete upper bound
for all-disjoint linearly recoverable LRC sequences.
When we do so and substitute $\theta = 0$
in the objective function in~(\ref{eq:bound1}),
we get the expression $\Rate_\CM(\delta,n)$ in~(\ref{eq:CM}).
Yet, generally, the minimum in~(\ref{eq:bound1}) is obtained
at some nonzero $\theta$, thereby yielding an improvement.
For $\tau \rightarrow 1$ (which forces $\theta = \delta$),
the objective function in~(\ref{eq:bound1}) becomes
the asymptotic version of the sphere-packing bound of~\cite{WZL}
(see \Theorem~\ref{thm:spherepacking} below):
\ifPAGELIMIT
    \[
    R \le \Rate_\SP(\delta,n) = \Rate_0(\delta/2,n) .
    \]
\else
\begin{equation}
\label{eq:spherepacking}
R \le \Rate_\SP(\delta,n) = \Rate_0 \left( \frac{\delta}{2},n \right) .
\end{equation}
Thus, the objective function in~(\ref{eq:bound1}) can also be written
as:
\[
\tau \cdot \Rate_\SP(\theta,n)
+ (1-\tau) \cdot \Rate_\LRC(\theta',n) .
\]
\fi

\ifPAGELIMIT
\else
\begin{remark}
\label{rem:bound1}
\Theorem~\ref{thm:bound1} can be given the following
geometric interpretation (see~\cite[pp.~77]{LL}):
for any distinct $\theta_1, \theta_2 \in [0,(q{-}1)/q]$,
the line in the $(\delta,R)$-plane
through the points $(\theta_1,\Rate_\SP(\theta_1))$
and $(\theta_2,\Rate_\LRC(\theta_2,n))$
is an upper bound on the rate for any
$\delta \in [\min\{\theta_1,\theta_2\},\max\{\theta_1,\theta_2\}]$.
In particular, if 
$\delta \mapsto \Rate_\LRC(\delta,n)$ is convex,
then, from the convexity of
$\delta \mapsto \Rate_\SP(\delta,n)$ we get that
the lower convex envelope of
$\min \left\{ \Rate_\SP(\delta,n), \Rate_\LRC(\delta,n) \right\}$
is also an upper bound.\qed
\end{remark}
\fi

Our second main result is the following bound.

\begin{theorem}
\label{thm:bound2}
The rate of any
all-disjoint linearly recoverable $(\delta,n)$-LRC sequence over $F$
is bounded from above by
\[
\Rate_2(\delta,n) = \min_{\omega \in [\delta/2,(q{-}1)/q]}
\Bigl\{ \Rate_0(\omega,n) + \Rate_\opt(\delta,\omega) \Bigr\} .
\]
\end{theorem}

The bound of \Theorem~\ref{thm:bound2} can be further improved by
substituting $\Rate_\LRC(\delta,n) = \Rate_2(\delta,n)$
in \Theorem~\ref{thm:bound1}.

\ifPAGELIMIT
    The various bounds are summarized in Table~\ref{tab:q=2,n=4}
    for $q = 2$, $n = 4$, and
    $\delta \in \{ 0.07, 0.10, 0.15, 0.30 \}$
    (full plots can be found in~\cite{RothFull}).
    The bound $\Rate_\CM(\delta,n)$ is computed by substituting
    the linear-programming bound $\Rate_\LP(\delta)$ of~\cite{MRRW}
    for $\Rate_\opt(\delta)$ and, similarly,
    $\Rate_1(\delta,n)$ is computed taking
    $\Rate_\LRC(\delta,n) = \Rate_\LP(\delta)$ in~(\ref{eq:bound1}).
    The last column corresponds to the bound $\Rate_1(\delta,n)$,
    now taking $\Rate_\LRC(\delta,n) = \Rate_2(\delta,n)$.
    The bounds of~\cite{ABHMT} slightly improve on $\Rate_\CM(\delta,n)$
    yet not for the parameters shown in the table.
\else
The various bounds are plotted in
Figures~\ref{fig:q=2,n=3} and~\ref{fig:q=2,n=4}
for $q = 2$ and $n = 3, 4$:
\begin{itemize}
\itemsep0ex
\item
curve~(a) presents $\delta \mapsto \Rate_\SP(\delta,n)$;
\item
curve~(b) presents $\delta \mapsto \Rate_\CM(\delta,n)$,
where we have substituted
the linear-programming bound $\Rate_\LP(\delta)$ of~\cite{MRRW}
for $\Rate_\opt(\delta)$;
\item
curve~(c) presents $\delta \mapsto \Rate_1(\delta,n)$ taking
$\Rate_\LRC(\delta,n) = \Rate_\LP(\delta)$ in~(\ref{eq:bound1})
(the same curve is obtained also for
$\Rate_\LRC(\delta,n) = \Rate_\CM(\delta,n)$);
\item
curve~(d) presents $\delta \mapsto \Rate_2(\delta,n)$;
\item
curve~(e) presents $\delta \mapsto \Rate_1(\delta,n)$ taking
$\Rate_\LRC(\delta,n) = \Rate_2(\delta,n)$ in~(\ref{eq:bound1});
\item
and curve~(f) presents the lower bound
$\delta \mapsto \Rate_0(\delta,n)$.
\end{itemize}
In the range where any of the curves~(b)--(d)
is not seen it coincides with curve~(e).
The values of the upper bounds for $q = 2$, $n = 4$, and
$\delta \in \{ 0.07, 0.10, 0.15, 0.30 \}$
are summarized in Table~\ref{tab:q=2,n=4}.
Notice that curve~(d) is not convex and that there is
a (small) range where it is worse than curve~(b),
but curve~(e) yields the best results.
The bounds of~\cite{ABHMT} slightly improve on $\Rate_\CM(\delta,n)$
but are too close to it to be seen at the scale of the figures.
\fi
\begin{table}[hbt]
\caption{Values of the bounds for $q = 2$ and $n = 4$.}
\label{tab:q=2,n=4}
\ifPAGELIMIT
    \vspace{-3ex}
\fi
\[
\renewcommand{\arraystretch}{1.1}
\begin{array}{cccccc}
\hline\hline
\quad\quad \delta \quad\quad &
\ifPAGELIMIT
    \Rate_\SP & \Rate_\CM & \Rate_1 & \Rate_2 & \Rate_{1,2} \\
\else
\mathrm{(a)}&\mathrm{(b)}&\mathrm{(c)}&\mathrm{(d)}&\mathrm{(e)} \\
\fi
\hline
0.07 & 0.6133 & 0.6317 & 0.6131 & 0.6079 & 0.6079 \\
0.10 & 0.5681 & 0.5809 & 0.5643 & 0.5576 & 0.5576 \\
0.15 & 0.5004 & 0.4964 & 0.4830 & 0.4781 & 0.4781 \\
0.30 & 0.3346 & 0.2427 & 0.2391 & 0.2470 & 0.2391 \\
\hline\hline
\end{array}
\]
\end{table}

Our second set of results includes (weaker) counterparts
of \Theorems~\ref{thm:bound1} and~\ref{thm:bound2} that apply
generally to linearly recoverable LRC sequences
(that are not necessarily all-disjoint). 

\ifPAGELIMIT
    Due to space limitations,
    some proofs are omitted from this abstract;
    they can be found in~\cite{RothFull}.
\else
The rest of this work is organized as follows.
In Section~\ref{sec:tools}, we present some basic tools from
large deviations theory which are tailored to our needs.
Additional tools will be presented in Section~\ref{sec:spherepacking},
where we also prove
the asymptotic sphere-packing bound~(\ref{eq:spherepacking}).
Section~\ref{sec:mainresults} is devoted to proving
\Theorems~\ref{thm:bound1} and~\ref{thm:bound2}.
Then, in Section~\ref{sec:q=2,n=3},
we present improved results for the special case $q = 2$ and $n = 3$.
Finally, in Section~\ref{sec:nondisjoint},
we present our bounds for
$(\delta,n)$-LRC sequences that are not necessarily
all-disjoint (but are still linearly recoverable).

\section{Large deviation tools}
\label{sec:tools}

We summarize here several basic notions from
large deviations theory (see~\cite[Sections~2.1.2 and 2.2]{DZ}).
Let $X$ be a random variable
which takes values in a finite subset $\Set \subseteq \Realfield$,
with $\Prob \{ X = x \} = p(x) > 0$ for every $x \in \Set$.
For every $u \in \Realfield$, let the function
$g_u : (0,1] \rightarrow \Realfield$ be defined by
\[
g_u(z) = z^{-u} \cdot \Expected \left\{ z^X \right\}
= \sum_{x \in \Set} p(x) \cdot z^{x - u}
\]
and let
\begin{equation}
\label{eq:gamma}
\gamma(u) = \gamma_X(u)
= \inf_{z \in (0,1]} g_u(z) .
\end{equation}

The following theorem follows from
the Chernoff bound and Cram\'{e}r's theorem.

\begin{theorem}
\label{thm:cramer}
Let $(X_i)_{i=1}^\infty$ be a sequence of i.i.d.\ random variables
which take values in a finite subset $\Set \subseteq \Realfield$.
Then for every real $u \ge x_{\min} = \min_{x \in \Set} x$
and $\ell \in \Integers^+$,
\begin{equation}
\label{eq:chernoff}
\frac{1}{\ell} \log
\Prob \left\{ \frac{1}{\ell} \sum_{i=1}^\ell X_i \le u \right\}
\le \log \gamma(u) .
\end{equation}
Moreover,
\begin{equation}
\label{eq:cramer}
\lim_{\ell \rightarrow \infty}
\frac{1}{\ell} \log
\Prob \left\{ \frac{1}{\ell} \sum_{i=1}^\ell X_i \le u \right\}
= \log \gamma(u) .
\end{equation}
\end{theorem}

Some properties of $u \mapsto \gamma(u)$ are summarized
in the next
\ifPAGELIMIT
    lemma.
\else
lemma, which is proved in Appendix~\ref{sec:skippedproofs}.
\fi

\begin{lemma}
\label{lem:gamma}
The function $u \mapsto \gamma(u)$ is
\begin{itemize}
\item
identically zero when 
$u < x_{\min} = \min_{x \in \Set} x$,
\item
equal to $p(x_{\min})$ when $u = x_{\min}$,
\item
constant $1$ when
$u \ge \Expected \{ X \}$,
\item
strictly increasing when
$x_{\min} \le u < \Expected \{ X \}$,
\item
continuous when $u \in [x_{\min}^+,\infty)$, and---
\item
log-concave (i.e., $u \mapsto \log \gamma(u)$ is concave)
when $u \ge x_{\min}$.
\end{itemize}
\end{lemma}

Let $\code$ be a linear $[n,k]$ code over $F$ and denote by
$W_\code(z) = \sum_{i=0}^n W_i z^i$ the weight enumerator polynomial 
of $\code$, namely,
\[
W_i = \left| \{ \bldc \in \code \,:\, \weight(\bldc) = i \} \right| ,
\]
where $\weight(\cdot)$ denotes Hamming weight.
For $\omega \in \Realfield$, define
\[
\alpha(\omega) = \alpha_\code(\omega)
= \inf_{z \in (0,1]} \left\{ z^{-n \omega} \cdot W_\code(z) \right\} .
\]
We have $\alpha(\omega) = q^k \cdot \gamma_X(n \omega)$,
where $X$ is a random variable which equals the weight of
a codeword selected uniformly from $\code$.
In particular, $\omega \mapsto \alpha(\omega)$ is
continuous on $\Realfield_{\ge 0}$.
Moreover, when $\code$ has no trivial coordinates
(i.e., none of the coordinates is identically zero
across all codewords) then $\Expected \{ X \} = n (q{-}1)/q$;
so, $\omega \mapsto \alpha(\omega)$ is strictly increasing
(and therefore also positive) at any positive $\omega < (q{-}1)/q$.

For a code $\varcode$ of length $N$ over $F$ and
$\omega \in \Realfield_{\ge 0}$,
we denote by $\varcode(\omega)$ the set of all codewords
in $\varcode$ that are contained in a Hamming ball of radius $\omega N$:
\[
\varcode(\omega) =
\left\{ \bldc \in \varcode \,:\, \weight(\bldc) \le \omega N \right\} .
\]

\begin{lemma}
\label{lem:cramer}
Let $\code$ be a linear code $[n,k]$ over $F$ 
and for $\ell \in \Integers^+$,
let $\Code^{(\ell)}$ be the linear $[\ell n, \ell k]$ code over $F$
defined by the Cartesian product
\[
\Code^{(\ell)} = \code^{\times \ell} 
= \underbrace{\code \times \code \cdots \times \code}_%
                                              {\ell \; \mathrm{times}} .
\]
Then the following holds.
\begin{list}{}{\settowidth{\labelwidth}{\textit{(ii)}}}
\item[(i)]
For any $\omega \in \Realfield_{\ge 0}$ and $\ell \in \Integers^+$:
\[
\frac{1}{\ell} \log |\Code^{(\ell)}(\omega)| \le \log \alpha(\omega) .
\]
\item[(ii)]
For any $\omega \in \Realfield_{\ge 0}$:
\[
\lim_{\ell \rightarrow \infty}
\frac{1}{\ell} \log |\Code^{(\ell)}(\omega)| = \log \alpha(\omega) .
\]
\end{list}
\end{lemma}

\begin{proof}
Apply \Theorem~\ref{thm:cramer}
with the random variable $X_i$ taken as
the weight of a codeword selected uniformly from $\code$.
\end{proof}

\ifPAGELIMIT
\else
\begin{remark}
\label{rem:variant}
A variant of Lemma~\ref{lem:cramer} holds also for
$\Code^{(\ell)}(\omega \pm \varepsilon)$,
defined as the set of all codewords
in $\Code^{(\ell)}$ of weight within $(\omega \pm \varepsilon) \ell n$.
Assuming that $\code$ has no trivial coordinates,
it follows from~\cite[\Theorem~2.1.24]{DZ}
that for any $\omega \in (0,(q{-}1)/q]$:
\[
\lim_{\varepsilon \rightarrow 0}
\lim_{\ell \rightarrow \infty}
\frac{1}{\ell} \log |\Code^{(\ell)}(\omega \pm \varepsilon)|
= \log \alpha(\omega) .
\]
For $(q{-}1)/q < \omega < \max_{\bldc \in \code} \weight(\bldc)/n$,
this equality holds if we replace $\alpha(\omega)$ by:
\[
\alpha^*(\omega)
= \inf_{z \in (0,1]} z^{n \omega} \cdot W_\code(z^{-1})
\]
(this can be shown by stating \Theorem~\ref{thm:cramer}
with $X_i$ and $u$ replaced by $-X_i$ and $-u$, respectively).\qed
\end{remark}
\fi

\begin{example}
\label{ex:lambda}
Let $\code$ be the $[n,n{-}1,2]$ parity code over $F$.
The weight enumerator polynomial of the $[n,1,n]$ repetition code,
which is the dual code of $\code$, is given by
\[
W_{\code^\perp}(z) = 1 + (q{-}1) z^n .
\]
Hence, by MacWilliams' identities (see~\cite[\Theorem~4.6]{Roth}):
\[
W_\code(z)
= \frac{1}{q} \Bigl( (1 + (q{-}1)z)^n + (q{-}1)(1{-}z)^n \Bigr) .
\]
Thus, for the code $\code$ we have:
\begin{eqnarray*}
\frac{1}{n} \cdot \log_q \alpha(\omega)
& = &
\frac{1}{n}
\inf_{z \in (0,1]}
\log_q \left( z^{-n \omega} \cdot W_\code(z) \right) \\
& = &
\lambda(\omega,n) ,
\end{eqnarray*}
where $\lambda(\omega,n)$ is defined in~(\ref{eq:lambda}). From this
we can conclude that the mapping $\omega \mapsto \lambda(\omega,n)$
takes the value $0$ at $\omega = 0$ and $(n{-}1)/n$
at $\omega = (q{-}1)/q$, and is strictly increasing in between
for any $n > 1$.
For $\omega \ge (q{-}1)/q$ it remains
\ifPAGELIMIT
a constant $(n{-}1)/n$.\qed
\else
a constant $(n{-}1)/n$.

In Figure~\ref{fig:lambda}, we have also depicted the following
function, which is defined
for $(q{-}1)/q < \omega < \max_{\bldc \in \code} \weight(\bldc)/n$:
\begin{eqnarray*}
\lambda^*(\omega,n) = \lambda^*_q(\omega,n)
& = &
\frac{1}{n} \cdot \log_q \alpha^*(\omega) \\
& = &
\frac{1}{n}
\inf_{z \in (0,1]}
\log_q \left( z^{n \omega} \cdot W_\code(z^{-1}) \right) .
\end{eqnarray*}
Note that $\max_{\bldc \in \code} \weight(\bldc)/n$ equals $1$,
except when $q = 2$ and $n$ is odd, in which case it equals $(n{-}1)/n$.
It is also fairly easy to see that when $q = 2$ and $n$ is even
we have
$\lambda^*_2(\omega,n) = \lambda_2(1{-}\omega,n)$.\qed
\fi
\end{example}
\fi

\section{Asymptotic sphere-packing bound}
\label{sec:spherepacking}

\ifPAGELIMIT
    We present in this section several tools and definitions.
\else
The purpose of this section is to present additional
tools that will be used in this work.
\fi
Along the way, we present
an asymptotic formulation of the sphere-packing bound of~\cite{WZL}
for the all-disjoint linearly recoverable case.
\ifPAGELIMIT
    We start with the next proposition, which will follow from
    Lemma~\ref{lem:mu} below.
\else
The following proposition will be useful for this purpose,
as well as for other proofs in the sequel.
\fi

\begin{proposition}
\label{prop:R0}
Given $n \in \Integers^+$,
let $(\Ambient_i)_{i=1}^\infty$ be
an infinite sequence of codes over $F$
where each $\Ambient_i$ is
a linear code of length $N_i$ over $F$ of the form
\[
\Ambient_i = \code_{i,1} \times \code_{i,2}
\times \cdots \times \code_{i,\ell_i} ,
\]
with each constituent code $\code_{i,j}$ being
a linear $[n_{i,j},n_{i,j}{-}1]$ code of length $n_{i,j} \le n$
over $F$.
Then for any nonnegative real sequence $(\omega_i)_{i=1}^\infty$
such that $\varliminf_{i \rightarrow \infty} \omega_i = \omega$:
\ifPAGELIMIT
    \[
    \varlimsup_{i \rightarrow \infty}
    \frac{1}{N_i}
    \log_q
    \frac{|\Ambient_i|}{|\Ambient_i(\omega_i)|}
    \le \Rate_0(\omega,n) ,
    \]
\else
\begin{equation}
\label{eq:limit}
\varlimsup_{i \rightarrow \infty}
\frac{1}{N_i}
\log_q
\frac{|\Ambient_i|}{|\Ambient_i(\omega_i)|}
\le \Rate_0(\omega,n) ,
\end{equation}
\fi
where $\Rate_0(\omega,n)$ is defined
in~(\ref{eq:R0})--(\ref{eq:lambda}).
\end{proposition}

\ifPAGELIMIT
\else
\subsection{Proof of \Proposition~\protect\ref{prop:R0}}
\label{sec:R0}

We will prove \Proposition~\ref{prop:R0} in two steps.
We first prove it under an additional assumption
on the sequence $(\Ambient_i)_{i=1}^\infty$
(Lemma~\ref{lem:lambda}).
We then prove a more general claim (Lemma~\ref{lem:mu}),
which will also be needed in Section~\ref{sec:nondisjoint}
where we remove the all-disjoint assumption.

\begin{lemma}
\label{lem:lambda}
\Proposition~\ref{prop:R0} holds when
all the constituent codes have the same length $n_{i,j} = n$.
\end{lemma}

\begin{proof}
The claim is immediate when $\omega = 0$, so we assume hereafter
in the proof that $\omega > 0$.

For $s \in [n]$, let $\code_s$
be the linear $[n,n{-}1]$ code over $F$ with the parity-check matrix
\[
\bigl( \underbrace{1 \, 1 \, \ldots \, 1}_{s \; \mathrm{times}}
0 \, 0 \, \ldots \, 0 \bigr) .
\]
Without real loss of generality we can assume that
$\code_{i,j} \in \left\{ \code_s \right\}_{s=1}^n$
for each $j \in [\ell_i]$. Therefore, we can write
\begin{equation}
\label{eq:product}
\Ambient_i = \Code_1^{(\ell_{i,1})} \times \Code_2^{(\ell_{i,2})}
\times \cdots \times \Code_n^{(\ell_{i,n})} ,
\end{equation}
where $\Code_s^{(\ell)} = (\code_s)^{\times \ell}$
and $\ell_{i,1}, \ell_{i,2}, \ldots, \ell_{i,n}$ are nonnegative
integers that sum to $\ell_i$.
(The code $\Code_s^{(0)}$ is taken to have zero length,
which will practically mean that we remove
$\Code_s^{(\ell_{i,s})}$ from the product~(\ref{eq:product})
in case $\ell_{i,s} = 0$.)

For each $s \in [n]$, the weight enumerator polynomial of
the dual code of $\code_s$ is given by
\[
W_{\code_s^\perp}(z) = 1 + (q{-}1) z^s
\]
and, so, by MacWilliams' identities:
\[
W_{\code_s}(z) = \frac{1}{q}
(1 + (q{-}1)z)^{n-s}
\Bigl( (1 + (q{-}1)z)^s + (q{-}1)(1{-}z)^s \Bigr) .
\]
In particular, for every $z \in (0,1]$:
\[
W_{\code_s}(z) \ge W_{\code_n}(z) = \frac{1}{q}
\Bigl( (1 + (q{-}1)z)^n + (q{-}1)(1{-}z)^n \Bigr) .
\]
Hence, for every $i \in \Integers^+$:
\[
\log_q \alpha_{\code_s}(\omega_i) 
\ge \log_q \alpha_{\code_n}(\omega_i) 
= n \cdot \lambda(\omega_i,n) .
\]
Fixing some $\varepsilon \in (0,\omega)$,
by Lemma~\ref{lem:cramer}(ii) we then get that,
whenever $i$ and $\ell_{i,s}$ are sufficiently large,
\begin{eqnarray}
\frac{1}{\ell_{i,s} n} \log_q |\Code_s^{(\ell_{i,s})}(\omega_i)|
& \ge &
\frac{1}{\ell_{i,s} n}
\log_q |\Code_s^{(\ell_{i,s})}(\omega{-}\varepsilon)| \nonumber \\
\label{eq:limit-s}
& \ge & \lambda(\omega{-}\varepsilon,n) - \varepsilon .
\end{eqnarray}

For $i \in \Integers^+$ define
\[
\AltSet(i) =
\left\{ s \in [n] \,:\, \ell_{i,s} \ge \sqrt{\ell_i} \right\} .
\]
For sufficiently large $i$ (and, therefore,
sufficiently large $\ell_i$) we have:
\begin{eqnarray*}
\lefteqn{
\frac{1}{\ell_i n} \log_q |\Ambient_i(\omega_i)|
} \makebox[5ex]{} \\
& \ge &
\frac{1}{\ell_i n}
\sum_{s=1}^n \log_q |\Code^{(\ell_{i,s})}(\omega_i)| \\
& \ge &
\sum_{s \in \AltSet(i)}
\frac{\ell_{i,s}}{\ell_i}
\cdot\frac{1}{\ell_{i,s} n}
\log_q |\Code^{(\ell_{i,s})}(\omega_i)| \\
& \stackrel{\textrm{(\ref{eq:limit-s})}}{\ge} &
(\lambda(\omega{-}\varepsilon,n) - \varepsilon)
\sum_{s \in \AltSet(i)}
\frac{\ell_{i,s}}{\ell_i} .
\end{eqnarray*}
Noting that
\[
1 - \frac{n}{\sqrt{\ell_i}}
< \sum_{s \in \AltSet(i)}
\frac{\ell_{i,s}}{\ell_i} \le 1 ,
\]
we conclude that
\[
\varliminf_{i \rightarrow \infty}
\frac{1}{\ell_i n} \log_q |\Ambient_i(\omega_i)|
\ge \lambda(\omega{-}\varepsilon,n) - \varepsilon.
\]
Finally, we get~(\ref{eq:limit}) by taking
the limit $\varepsilon \rightarrow 0$
and recalling that $N_i = \ell_i n$ and that
\[
\frac{1}{\ell_i n} \log_q |\Ambient_i| = \frac{n{-}1}{n} .
\]
\end{proof}

\begin{remark}
\label{rem:lambda}
It follows from the proof that equality is
attained in~(\ref{eq:limit})
when $\Ambient_i = (\code_n)^{\times \ell_i}$, namely,
the growth rate of $\Ambient_i(\omega)$ is minimized
when each $\code_{i,j}$ is the $[n,n{-}1,2]$ parity code.
This result, however, holds only asymptotically
(namely, when $i \rightarrow \infty$)
and not necessarily for individual values of $i$.
For example, for $q = 2$, $n = 5$, and $\ell_i = \ell_1 = 1$
we have $|\code_5(0.4)| = 11$
yet $|\code_2(0.4)| = 8$ and $|\code_3(0.4)| = 7$.\qed
\end{remark}

The next lemma involves the following definition.
\fi
For $n \in \Integers^+$, $\omega \in \Realfield_{\ge 0}$, and
a real $\mu \in [1,n]$, define:
\begin{equation}
\label{eq:WZL1}
\overline{\Rate}_0(\omega,n,\mu) = \max_\bldpi \; \min_\bldvartheta
\sum_{s \in [n]} \pi_s \cdot \Rate_0(\vartheta_s,s) ,
\end{equation}
where $\Rate_0(\cdot,\cdot)$ is as in~(\ref{eq:R0}),
the maximum is taken over all vectors
$\bldpi = (\pi_s)_{s \in [n]} \in \Realfield_{\ge 0}^n$
that satisfy
\ifPAGELIMIT
    \[
    \sum_{s \in [n]} \pi_s = 1
    \quad \textrm{and} \quad
    \sum_{s \in [n]} \frac{\pi_s}{s} \ge \frac{1}{\mu} ,
    \]
\else
\begin{list}{}{\settowidth{\labelwidth}{\textrm{P3)}}}
\item[P1)]
$\displaystyle \sum_{s \in [n]} \pi_s = 1$ and
\item[P2)]
$\displaystyle \sum_{s \in [n]} \frac{\pi_s}{s} \ge \frac{1}{\mu}$,
\end{list}
\fi
and the minimum is taken over all vectors
$\bldvartheta = (\vartheta_s)_{s \in [n]} \in \Realfield_{\ge 0}^n$
that satisfy
\ifPAGELIMIT
    \[
    \sum_{s \in [n]} \pi_s \cdot \vartheta_s = \omega .
    \]
\else
\begin{list}{}{\settowidth{\labelwidth}{\textrm{P3)}}}
\item[P3)]
$\displaystyle \sum_{s \in [n]} \pi_s \cdot \vartheta_s = \omega$.
\end{list}
\fi

\begin{remark}
\label{rem:computations1}
The expression in~(\ref{eq:WZL1}) is fairly easy to compute,
and we provide a formula for it
\ifPAGELIMIT
    in~\cite{RothFull}.
    In particular, we show there that
    $\overline{\Rate}_0(\omega,n,\mu) = \Rate_0(\omega,\mu)$
    when $\mu$ is an integer.\qed
\else
in \Proposition~\ref{prop:WZL} in Appendix~\ref{sec:R0properties}.
In particular, we show there that the support of
the maximizing vector $\bldpi$ in~(\ref{eq:WZL1}) 
is $\left\{ \lfloor\mu\rfloor, \lceil\mu\rceil \right\}$; thus,
when $\mu$ is an integer then
\[
\overline{\Rate}_0(\omega,n,\mu) = \Rate_0(\omega,\mu).
\]
It will also follow from our analysis that
for given $\omega \in \Realfield_{\ge 0}$ and $\mu \in [1,\infty)$,
the value $\overline{\Rate}_0(\omega,n,\mu)$ is the \emph{same}
for all $n \ge \mu$.\qed
\fi
\end{remark}

\Proposition~\ref{prop:R0} follows by substituting $\mu = n$
in the next lemma.

\begin{lemma}
\label{lem:mu}
Let $(\Ambient_i)_{i=1}^\infty$ be as in
\Proposition~\ref{prop:R0}
and assume in addition that for a prescribed $\mu \in [1,n]$,
the \underline{average length} of the constituent codes satisfies
\ifPAGELIMIT
    \[
    \varlimsup_{i \rightarrow \infty} N_i/\ell_i \le \mu .
    \]
\else
\begin{equation}
\label{eq:averagelength}
\varlimsup_{i \rightarrow \infty} \frac{N_i}{\ell_i} \le \mu .
\end{equation}
\fi
Then for any nonnegative real sequence $(\omega_i)_{i=1}^\infty$
such that $\varliminf_{i \rightarrow \infty} \omega_i = \omega$:
\[
\varlimsup_{i \rightarrow \infty}
\frac{1}{N_i}
\log_q
\frac{|\Ambient_i|}{|\Ambient_i(\omega_i)|} \le
\overline{\Rate}_0(\omega,n,\mu) ,
\]
where $\overline{\Rate}_0(\omega,n,\mu)$
is defined in~(\ref{eq:WZL1}).
\end{lemma}

\ifPAGELIMIT
Refer to~\cite{RothFull} for a proof of the lemma.
\else
\begin{proof}
By possibly permuting the coordinates of $\Ambient_i$,
we can write
\[
\Ambient_i =
\Ambient_{i,1} \times \Ambient_{i,2} \times \cdots \times
\Ambient_{i,n} ,
\]
where each $\Ambient_{i,s}$
is a linear $[s \cdot \ell_{i,s}, (s{-}1) \cdot \ell_{i,s}]$ code
over $F$ of the form
\[
\Ambient_{i,s} = \code'_{i,1} \times \code'_{i,2}
\times \cdots \times \code'_{i,\ell_{i,s}} ,
\]
with each $\code'_{i,j}$ being
a linear $[s,s{-}1]$ code over $F$.

Define $\bldpi_i = (\pi_{i,s})_{s \in [n]}$ by
\begin{equation}
\label{eq:piis}
\pi_{i,s} = \frac{s \cdot \ell_{i,s}}{N_i}
\end{equation}
(i.e., $\pi_{i,s}$ is the fraction of coordinates of $\Ambient_i$
that correspond to constituent codes of length $s$).
We have
\begin{equation}
\label{eq:piis1}
\sum_{s \in [n]} \pi_{i,s}
= \frac{1}{N_i} \sum_{s \in [n]} s \cdot \ell_{i,s}
= \frac{1}{N_i} \sum_{j \in [\ell_i]} n_{i,j} = 1
\end{equation}
and
\begin{equation}
\label{eq:piis2}
\varliminf_{i \rightarrow \infty}
\sum_{s \in [n]} \frac{\pi_{i,s}}{s}
=
\varliminf_{i \rightarrow \infty}
\frac{1}{N_i} \sum_{s \in [n]} \ell_{i,s}
=
\varliminf_{i \rightarrow \infty}
\frac{\ell_i}{N_i}
\stackrel{\textrm{(\ref{eq:averagelength})}}{\ge}
\frac{1}{\mu} .
\end{equation}
By possibly restricting to a subsequence of
$(\Ambient_i)_{i=1}^\infty$,
we can assume that the sequence $(\bldpi_i)_{i=1}^\infty$ converges to
a limit $\bldpi = (\pi_s)_{s \in [n]}$;
by~(\ref{eq:piis1})--(\ref{eq:piis2}),
this limit satisfies conditions~(P1) and~(P2).
 
Let $(\omega_i)_{i=1}^\infty$
be a nonnegative real sequence
such that $\varliminf_{i \rightarrow \infty}\omega_i = \omega$
and let $\bldvartheta = (\vartheta_s)_{s \in [n]}$
be a vector in $\Realfield_{\ge 0}^n$ that satisfies condition~(P3).
For any $i \in \Integers^+$ and $s \in [n]$, define:
\[
\vartheta_{i,s} =
\left\{
\begin{array}{ccl}
\rule[0ex]{0ex}{3ex}
\displaystyle
\frac{\pi_s}{\pi_{i,s}} \cdot
\frac{\omega_i}{\omega} \cdot \vartheta_s
&& \textrm{if $\omega > 0$ and $s \in \Support(\bldpi_i)$} \\
0
&& \textrm{otherwise}
\end{array}
\right.
,
\]
where $\Support(\cdot)$ denotes the support of a vector.
It is easily seen that
\begin{equation}
\label{eq:vartheta}
\sum_{s \in [n]} \pi_{i,s} \cdot \vartheta_{i,s} \le \omega_i
\end{equation}
and that for any $s \in \Support(\bldpi)$:
\[
\lim_{i \rightarrow \infty} \vartheta_{i,s} = \vartheta_s .
\]
We then have:
\[
\frac{1}{N_i}
\log_q |\Ambient_i|
\stackrel{\textrm{(\ref{eq:piis})}}{=}
\sum_{s \in \Support(\bldpi_i)}
\pi_{i,s} \cdot
\frac{1}{s \cdot \ell_{i,s}}
\log_q |\Ambient_{i,s}|
\]
and
\[
\frac{1}{N_i}
\log_q |\Ambient_i(\omega_i)|
\stackrel{\textrm{(\ref{eq:piis})+(\ref{eq:vartheta})}}{\ge}
\!
\sum_{s \in \Support(\bldpi_i)}
\!
\pi_{i,s} \cdot
\frac{1}{s \cdot \ell_{i,s}}
\log_q |\Ambient_{i,s}(\vartheta_{i,s})| .
\]
Hence,
\begin{eqnarray*}
\lefteqn{
\varlimsup_{i \rightarrow \infty}
\frac{1}{N_i}
\log_q \frac{|\Ambient_i|}{|\Ambient_i(\omega_i)|}
} \makebox[0ex]{} \\
& \le &
\varlimsup_{i \rightarrow \infty}
\sum_{s \in \Support(\bldpi_i)}
\pi_{i,s} \cdot
\frac{1}{s \cdot \ell_{i,s}}
\log_q \frac{|\Ambient_{i,s}|}{|\Ambient_{i,s}(\vartheta_{i,s})|} \\
& = &
\sum_{s \in \Support(\bldpi)}
\varlimsup_{i \rightarrow \infty}
\pi_{i,s} \cdot
\frac{1}{s \cdot \ell_{i,s}}
\log_q \frac{|\Ambient_{i,s}|}{|\Ambient_{i,s}(\vartheta_{i,s})|} \\
& \stackrel{\textrm{Lemma~\ref{lem:lambda}}}{\le} &
\sum_{s \in [n]} \pi_s \cdot \Rate_0(\vartheta_s,s) .
\end{eqnarray*}
Finally, selecting $\bldvartheta \in \Realfield_{\ge 0}^n$
that minimizes the inner expression in~(\ref{eq:WZL1})
(subject to condition~(P3)), we get:
\begin{eqnarray*}
\varlimsup_{i \rightarrow \infty}
\frac{1}{N_i}
\log_q \frac{|\Ambient_i|}{|\Ambient_i(\omega_i)|}
& \le &
\min_\bldvartheta
\sum_{s \in [n]} \pi_s \cdot \Rate_0(\vartheta_s,s) \\
& \le &
\overline{\Rate}_0(\omega,n,\mu) ,
\end{eqnarray*}
for any $\bldpi \in \Realfield_{\ge 0}^n$
that satisfies conditions (P1)--(P2).
\end{proof}

\subsection{Ambient spaces of LRCs}
\label{sec:ambient}
\fi

Let $\varcode$ be a linearly recoverable LRC 
(which is not necessarily all-disjoint)
of length $N$ and locality $r = n-1$ over $F$,
and let $\{ \Local_j \}_j$ be
a set of distinct repair groups of $\varcode$.
Recall that for each repair group $\Local_j$,
the constituent code $(\varcode)_{\Local_j}$ is a subcode
of a linear $[|\Local_j|,|\Local_j|{-}1]$ code, $\code_j$, over $F$.
Let $\Ambient$ be a largest subset of $F^N$
which satisfies the following linear constraints:
\begin{equation}
\label{eq:ambient}
(\Ambient)_{\Local_j} \subseteq \code_j ,
\quad j \in [\ell] .
\end{equation}
Clearly, $\Ambient$ is a linear subspace of $F^N$
(which is defined uniquely by~(\ref{eq:ambient}))
and $\varcode \subseteq \Ambient$.
We will refer hereafter to $\Ambient$ as an \emph{ambient space}
of $\varcode$ (this term is similar to the notion of
an $\mathcal{L}$-space defined in~\cite{WZL}).

In the all-disjoint case, we can assume that
$\{ \Local_j \}_j$ forms a partition of $[N]$.
By possibly permuting the coordinates of both
$\varcode$ and $\Ambient$, we can further assume that
$\Ambient$ takes the form
\[
\Ambient = \code_1 \times \code_2 \times \cdots \times \code_\ell .
\]

\ifPAGELIMIT
\else
\subsection{Sphere-packing bound}
\label{sec:spherepackingbound}
\fi

Ambient spaces were used in~\cite{WZL} as an ingredient
in obtaining a sphere-packing bound on LRCs.
The asymptotic formulation of this bound is presented in
the next theorem.

\begin{theorem}
\label{thm:spherepacking}
The rate of any all-disjoint linearly recoverable
$(\delta,n)$-LRC sequence over $F$ is bounded from above by
\[
\Rate_\SP(\delta,n) = 
\ifPAGELIMIT
    \Rate_0(\delta/2,n)
\else
\Rate_0 \left( \frac{\delta}{2} , n \right)
\fi
.
\]
\end{theorem}

\begin{proof}
Let $(\varcode_i)_{i=1}^\infty$ be
an all-disjoint linearly recoverable $(\delta,n)$-LRC sequence,
with $N_i$ and $d_i$ being the length and the minimum distance
of $\varcode_i$.
Letting $\Ambient_i$ be an ambient space of $\varcode_i$,
by a sphere-packing argument we get:
\[
|\varcode_i| \le \frac{|\Ambient_i|}{\Ambient_i(d_i{-}1/(2N_i))} .
\]
The result follows from \Proposition~\ref{prop:R0}
by taking $\omega_i = (d_i{-}1)/(2N_i)$.
\end{proof}

\begin{remark}
\label{rem:averagelength1}
\ifPAGELIMIT
    More generally, we get from Lemma~\ref{lem:mu} that
    $\overline{\Rate}_0(\delta/2,n,\mu)$ is a sphere-packing
    upper bound in a setting where
    the \emph{average size} of the distinct (and disjoint)
    repair groups of each code in
    the LRC sequence is at most $\mu$.\qed
\else
For general $\mu \in [1,n]$, we get from Lemma~\ref{lem:mu} that
\[
\overline{\Rate}_0 \left( \frac{\delta}{2} , n, \mu \right)
\]
is an upper bound on the rate of any all-disjoint
linearly recoverable $(\delta,n)$-LRC sequence
$(\varcode_i)_{i=1}^\infty$ over $F$,
with $\mu$ bounding from above
the supremum (as $i \rightarrow \infty$) over
the \emph{average sizes}, $N_i/\ell_i$,
of the distinct (and disjoint) repair groups of $\varcode_i$.
A similar bound can be stated when the average size
is computed per coordinate, so that each repair group
is counted a number of times equaling its size.
See Remark~\ref{rem:averagelength2}
in Appendix~\ref{sec:R0properties}.\qed
\fi
\end{remark}

\section{Proofs of Theorems~\ref{thm:bound1} and~\ref{thm:bound2}}
\label{sec:mainresults}

In this section, we prove
\Theorems~\ref{thm:bound1} and~\ref{thm:bound2}.

We start with the next simple lemma, usually attributed to
\ifPAGELIMIT
    Bassalygo or Elias (see a short proof in~\cite{RothFull}).
\else
Bassalygo or Elias; it provides an effective tool
in obtaining upper bounds.
\fi

\begin{lemma}
\label{lem:bassalygo}
Let $\Ambient$ be a subspace of $F^N$
and let $\Subset$ and $\varcode$ be subsets of $\Ambient$.
Then there exists $\bldy \in \Ambient$ for which
\[
\left| (\bldy + \varcode) \cap \Subset \right| \ge
\frac{|\Subset|}{|\Ambient|} \cdot |\varcode| .
\]
\end{lemma}

\ifPAGELIMIT
\else
\begin{proof}
Let $\chi_\Subset : \Ambient \rightarrow \{ 0, 1 \}$
be the characteristic function of $\Subset$, namely,
$\chi(\bldy) = 1$ when $\bldy \in \Subset$
and $\chi(\bldy) = 0$ otherwise.
Then,
\begin{eqnarray*}
\frac{1}{|\Ambient|} \cdot |\varcode| \cdot |\Subset|
& = &
\frac{1}{|\Ambient|}
\sum_{\bldc \in \varcode} \sum_{\bldy \in \Ambient}
\chi_\Subset(\bldc + \bldy) \\
& = &
\frac{1}{|\Ambient|}
\sum_{\bldy \in \Ambient} \sum_{\bldc \in \varcode}
\chi_\Subset(\bldy + \bldc) \\
& = &
\frac{1}{|\Ambient|}
\sum_{\bldy \in \Ambient}
\left| (\bldy + \varcode) \cap \Subset \right| .
\end{eqnarray*}
Hence,
\[
\max_{\bldy \in \Ambient}
\left| (\bldy + \varcode) \cap \Subset \right|
\ge \frac{|\Subset|}{|\Ambient|} \cdot |\varcode| .
\]
\end{proof}
\fi

\begin{proof}[Proof of \Theorem~\ref{thm:bound1}]
We tailor the generalization of
the shortening method of~\cite[\Theorems~1 and~2]{LL} to our setting.
Let $(\varcode_i)_{i=1}^\infty$ be an all-disjoint linearly recoverable
$(\delta,n)$-LRC sequence over $F$
and fix $\tau$ in the interval $(0,1)$
and a pair $(\theta,\theta') \in [0,(q{-}1)/q]^2$
that satisfies~(\ref{eq:theta}).
For each $i \in \Integers^+$, let $N_i$ and $d_i$ be
the length and the minimum distance of $\varcode_i$, respectively,
where
\begin{equation}
\label{eq:di}
\varliminf_{i \rightarrow \infty}
\ifPAGELIMIT
d_i/N_i
\else
\frac{d_i}{N_i}
\fi
\ge \delta .
\end{equation}

Let $\{ \Local_{i,j} \}_{j \in [\ell_i]}$
be a set of distinct (and disjoint) repair groups of $\varcode_i$,
and denote
\[
t_i = \left|
{\textstyle\bigcup_{j \in [\ell'_i]}} \Local_{i,j} \right| ,
\]
where $\ell'_i$ is the smallest in $[\ell_i]$ so that
$t_i \ge \tau N_i$; in particular,
\begin{equation}
\label{eq:ti}
\lim_{i \rightarrow \infty}
\ifPAGELIMIT
t_i/N_i
\else
\frac{t_i}{N_i}
\fi
= \tau .
\end{equation}
By possibly permuting the coordinates
of $\varcode_i$, we assume hereafter in the proof that
$\bigcup_{j \in [\ell'_i]} \Local_{i,j}$
indexes the first $t_i$
\ifPAGELIMIT
    coordinates
\else
coordinates\footnote{%
Namely, we deviate here from the notational convention that we set
in Footnote~\ref{footnote:locality}.}
\fi
of $\varcode_i$.
Letting $\Ambient_i$ be an ambient space of $\varcode_i$,
we denote by $\Ambient'_i$ the set of the distinct $t_i$-prefixes of
the vectors in $\Ambient_i$;
the sequence $(\Ambient'_i)_{i=1}^\infty$
satisfies the conditions of \Proposition~\ref{prop:R0}.

Define
\[
w_i =
\min
\left\{
\ifPAGELIMIT
\lfloor \theta \cdot t_i/2 \rfloor ,
\lfloor (d_i{-}1)/2 \rfloor 
\else
\Bigl\lfloor \frac{\theta}{2} \cdot t_i \Bigr\rfloor ,
\Bigl\lfloor \frac{d_i-1}{2} \Bigr\rfloor 
\fi
\right\} .
\]
By~(\ref{eq:di})--(\ref{eq:ti}) (and~(\ref{eq:theta})) we have:
\begin{equation}
\label{eq:mi}
\lim_{i \rightarrow \infty}
\ifPAGELIMIT
\omega_i/t_i = \theta/2
\else
\frac{\omega_i}{t_i} = \frac{\theta}{2}
\fi
.
\end{equation}
Next, we apply Lemma~\ref{lem:bassalygo}
with $\varcode \leftarrow \varcode_i$
and $\Ambient \leftarrow \Ambient_i$, and with
$\Subset$ taken as the set of all vectors in $\Ambient_i$
whose $t_i$-prefixes are in
\ifPAGELIMIT
    $\Ambient'_i(w_i/t_i)$.
\else
$\Ambient'_i(w_i/t_i)$; namely,
\begin{equation}
\label{eq:Subset}
\Subset = \Ambient'_i(w_i/t_i) \times \Ambient''_i ,
\end{equation}
where $\Ambient''_i$ is the set of
$(N_i - t_i)$-suffixes of the vectors of $\Ambient_i$.
\fi
It follows that
there exists $\bldy = (\bldy' \; \bldy'') \in \Ambient_i$
(where $\bldy' \in \Ambient'_i$)
for which the intersection
\[
\varcode^*_i = (\bldy + \varcode_i) \cap \Subset
\]
satisfies
\begin{equation}
\label{eq:bound1-1}
|\varcode_i|
\le
\frac{|\Ambient_i| \cdot |\varcode^*_i|}{|\Subset|}
= \frac{|\Ambient'_i|}{|\Ambient'_i(w_i/t_i)|} \cdot |\varcode^*_i| .
\end{equation}
Now, on the one hand,
the code $-\bldy + \varcode^*_i$ is
an all-disjoint linearly recoverable LRC
with locality $r = n-1$ and minimum distance${} \ge d_i$;
yet, on the other hand, all the $t_i$-prefixes of its codewords
are at distance at most $w_i$ from $-\bldy'$,
which means that any two prefixes are at most $2w_i$ apart.
It follows that all
the $(N_i - t_i)$-suffixes of the codewords of
$-\bldy + \varcode^*_i$
form an all-disjoint linearly recoverable LRC, $\varcode''_i$,
\ifPAGELIMIT
    of minimum distance${} \ge d_i - 2w_i \; (> 0)$.
\else
of minimum distance${} \ge d_i - 2w_i \; (> 0)$
(in particular, these suffixes are all distinct).
\fi
Hence,
\begin{equation}
\label{eq:bound1-2}
|\varcode^*_i| = |\varcode''_i| \le M_\LRC(N_i - t_i, d_i - 2w_i, n) ,
\end{equation}
where $M_\LRC(N,D,n)$ denotes the largest size
of any all-disjoint linearly recoverable LRC over $F$
with length $N$, minimum distance $D$, and locality $n-1$.
Combining~(\ref{eq:bound1-1}) and~(\ref{eq:bound1-2}) leads to
\[
|\varcode_i|
\le
\frac{|\Ambient'_i|}{|\Ambient'_i(w_i/t_i)|}
\cdot M_\LRC(N_i - t_i, d_i - 2w_i, n) ,
\]
and taking logarithms and dividing by $N_i$ yield
the following upper bound on the rate,
$R_i =  (\log_q |\varcode_i|)/N_i$, of $\varcode_i$:
\begin{eqnarray}
R_i
& \le &
\frac{1}{N_i}
\Bigl(
\log_q \frac{|\Ambient'_i|}{|\Ambient'_i(w_i/t_i)|} \nonumber \\
&&
\quad
{} + \log_q M_\LRC(N_i - t_i, d_i - 2w_i, n) \Bigr) \nonumber \\
& = &
\frac{t_i}{N_i} \cdot
\left(
\frac{1}{t_i}
\log_q \frac{|\Ambient'_i|}{|\Ambient'_i(w_i/t_i)|} \right) \nonumber \\
&&
\;\;
{} +
\left( 1 - \frac{t_i}{N_i} \right)
\cdot
\frac{\log_q M_\LRC(N_i - t_i, d_i - 2w_i, n)}{N_i - t_i} . \nonumber \\
\label{eq:Ri1}
&&
\end{eqnarray}

We now take the limit as $i \rightarrow \infty$
of each of the terms in~(\ref{eq:Ri1}).
By~(\ref{eq:mi}) and \Proposition~\ref{prop:R0} we get
\ifPAGELIMIT
    \[
    \varlimsup_{i \rightarrow \infty}
    \frac{1}{t_i}
    \log_q \frac{|\Ambient'_i|}{|\Ambient'_i(w_i/t_i)|}
    \le \Rate_0(\theta/2,n) .
\]
\else
\begin{equation}
\label{eq:ratio}
\varlimsup_{i \rightarrow \infty}
\frac{1}{t_i}
\log_q \frac{|\Ambient'_i|}{|\Ambient'_i(w_i/t_i)|}
\le \Rate_0 \left( \frac{\theta}{2} , n \right) .
\end{equation}
\fi
In addition,
\[
\varliminf_{i \rightarrow \infty}
\frac{d_i - 2w_i}{N_i - t_i}
\stackrel{\textrm{(\ref{eq:di})--(\ref{eq:mi})}}{\ge}
\frac{\delta - \tau \cdot \theta}{1-\tau}
\stackrel{\textrm{(\ref{eq:theta})}}{=}
\theta'
\]
and, so,
\[
\varlimsup_{i \rightarrow \infty}
\frac{\log_q M_\LRC(N_i - t_i, d_i - 2w_i, n)}{N_i - t_i}
\le \Rate_\LRC(\theta',n) .
\]
We conclude that the rate $R$
of the code sequence $(\varcode_i)_{i=1}^\infty$
satisfies
\[
R = \varlimsup_{i \rightarrow \infty} R_i \le \tau \cdot
\ifPAGELIMIT
\Rate_0(\theta/2)
\else
\Rate_0 \left( \frac{\theta}{2} \right)
\fi
+ (1-\tau) \cdot \Rate_\LRC(\theta', n) .
\]
The sought result is reached by
minimizing over $\theta$ and $\tau$.
\end{proof}

We now turn to the proof of \Theorem~\ref{thm:bound2}.
Recall that the expression $\Rate_2(\delta,n)$ therein
involves the value $\Rate_\opt(\delta,\omega)$, being
the supremum over all rates of sequences $(\varcode_i)_{i=1}^\infty$
with r.m.d.${} \ge \delta$
such that the codewords in each $\varcode_i$ all have
the \emph{same} weight${} \approx \omega N_i$.
For our purposes, it would be convenient if this definition
were relaxed so that the codeword weights only need to be
\emph{bounded from above} by
\ifPAGELIMIT
    $\omega N_i + o(N_i)$.
    It turns out that $\Rate_\opt(\delta,\omega)$ is
    the rate supremum also under this relaxed setting, provided that
    $\omega \in [0,(q{-}1)/q]$ (see~\cite{BCCST} and~\cite{RothFull}).
\else
(approximately) $\omega N_i$.
It turns out that $\Rate_\opt(\delta,\omega)$ is
the rate supremum also under this relaxed setting, provided that
$\omega \in [0,(q{-}1)/q]$.
We make this statement precise in Lemma~\ref{lem:CW} below.
\fi

For $N, d, w \in \Integers^+$,
let $M_\opt(N,d,w)$ denote the largest size of
any code in $F^N$ with minimum distance $d$ and codeword
weights that are all bounded from above by $w$.
\ifPAGELIMIT
\else
Define
\[
\Rate_\opt(\delta,{\le}\omega)
= \sup_{
\genfrac{}{}{0ex}{}{%
                  \bldd \in \InfSet(\delta)}{\bldw \in \SupSet(\omega)}}
\varlimsup_{N \rightarrow \infty}
\frac{1}{N}
\log_q M_\opt(N,d_N,w_N) ,
\]
where $\InfSet(\delta)$
(respectively, $\SupSet(\omega)$) is the set of
all sequences $(a_N)_{N=1}^\infty$ over $\Integers^+$
satisfying $\varliminf_{N \rightarrow \infty} a_N/N \ge \delta$
(respectively, $\varlimsup_{N \rightarrow \infty} a_N/N \le \omega$).

We have the following technical lemma, which is proved
in Appendix~\ref{sec:skippedproofs}.

\begin{lemma}
\label{lem:CW}
For every $\delta, \omega \in [0,(q{-}1)/q]$,
\[
\Rate_\opt(\delta,{\le}\omega) = \Rate_\opt(\delta,\omega) .
\]
\end{lemma}
\fi

\begin{proof}[Proof of \Theorem~\ref{thm:bound2}]
Let $(\varcode_i)_{i=1}^\infty$
be an all-disjoint linearly recoverable
$(\delta,n)$-LRC sequence over $F$.
Letting $\Ambient_i$ be an ambient space
of $\varcode_i$, the sequence $(\Ambient_i)_{i=1}^\infty$
satisfies the conditions of \Proposition~\ref{prop:R0}.

Fixing any
\ifPAGELIMIT
    $\omega$
\else
$\omega \in [\delta/2,(q{-}1)/q]$
\fi
and substituting $\varcode \leftarrow \varcode_i$,
$\Ambient \leftarrow \Ambient_i$, and
$\Subset \leftarrow \Ambient_i(\omega)$
in Lemma~\ref{lem:bassalygo} yield
that there exists $\bldy \in \Ambient_i$ for which
\begin{eqnarray*}
|\varcode_i| \cdot 
|\Ambient_i(\omega)|
& \le &
|\Ambient_i| \cdot 
\left| (\bldy + \varcode_i)
\cap \Ambient_i(\omega) \right| \\
& \le &
|\Ambient_i| \cdot M_\opt(N_i,d_i,\omega) ,
\end{eqnarray*}
where $N_i$ and $d_i$ are the length and
the minimum distance of $\varcode_i$, respectively.
Taking logarithms and dividing by $N_i$ yield
the following upper bound on the rate $R_i$ of $\varcode_i$:
\begin{equation}
\label{eq:Ri2}
R_i
\le
\frac{1}{N_i} \log_q \frac{|\Ambient_i|}{|\Ambient_i(\omega)|}
+ \frac{1}{N_i} \log_q M_\opt(N_i,d_i,w_i) ,
\end{equation}
where $w_i = \lfloor \omega N_i \rfloor$.
The result is obtained by taking
the limit as
\ifPAGELIMIT
    $i \rightarrow \infty$:
    the first term in the right-hand side of~(\ref{eq:Ri2})
    will then be at most $\Rate_0(\omega,n)$
    (by \Proposition~\ref{prop:R0}),
    and the second term will be at most $\Rate_\opt(\delta,\omega)$.
\else
$i \rightarrow \infty$.
Specifically, on the left-hand side of~(\ref{eq:Ri2}) we get
the rate $R = \varlimsup_{i \rightarrow \infty} R_i$
of $(\varcode_i)_{i=1}^\infty$,
and on the right-hand side the respective terms become
\[
\varlimsup_{i \rightarrow \infty}
\frac{1}{N_i} \log_q \frac{|\Ambient_i|}{|\Ambient_i(\omega)|}
\stackrel{\textrm{\Proposition~\ref{prop:R0}}}{\le}
\Rate_0(\omega,n)
\]
and
\[
\varlimsup_{i \rightarrow \infty}
\frac{1}{N_i}
\log_q M_\opt(N_i,d_i,w_i)
\stackrel{\textrm{Lemma}~\ref{lem:CW}}{\le}
\Rate_\opt(\delta,\omega) .
\]
\fi
\end{proof}

\ifPAGELIMIT
\else
\section{The case $q = 2$ and $n = 3$}
\label{sec:q=2,n=3}

We consider in this section
the case $n = 3$ over the binary field.
We assume that the LRCs are all-disjoint, but not
necessarily linearly recoverable.

As a warm-up, we start with the case $n = 2$ over any finite field.
It is straightforward to see that
if an all-disjoint $(N,M,d)$ LRC $\varcode$ over $F$
with locality $n-1 = 1$ has no fixed (in particular, trivial)
coordinates, then all the repair groups have to be of size
(exactly) $2$, and all the constituent codes have minimum distance
(exactly) $2$.\footnote{%
If $\varcode$ does contain trivial coordinates,
we can shorten it on these coordinates,
thereby only increasing the r.m.d.\ and the rate
and, thus, obtaining an upper bound on that larger rate.}
Thus, both $N$ and $d$ are even and,
without loss of generality,
each constituent code is the $[2,1,2]$ repetition code over $F$.
We can therefore view $\varcode$ as a concatenated code over $F$
comprising an outer $(N/2,M,d/2)$ code over $F$
and an inner $[2,1,2]$ repetition code over $F$.
We conclude that the rate of any
all-disjoint $(\delta,2)$-LRC sequence over $F$ is
bounded from above by $(1/2) \cdot \Rate_\opt (\delta)$;
moreover, we can get arbitrarily close to this bound
by a concatenated code construction.\footnote{%
This upper bound holds in fact also when the repair groups
are not necessarily disjoint, as overlapping repair groups force
all the entries that are indexed by their union to be equal.
This corresponds to having inner repetition codes
of rate smaller than $1/2$.}

We now turn to the case $n = 3$ when $F = \Finitefield_2$.
Let $\varcode$ be an all-disjoint $(N,M,d)$ LRC over $F$
with locality $n-1 = 2$ and without fixed coordinates,
and assume first that all the repair groups have size exactly $3$.
Then each constituent code $(\varcode)_{\Local_j}$,
being of length $3$ and minimum distance${} \ge 2$
and having no fixed coordinates,
has size $2$, $3$, or $4$. It is rather easy to see
that such codes are essentially unique (up to a replacement
$0 \leftrightarrow 1$ at any given coordinate across all codewords).
Moreover, any $(3,3,2)$ code over $F$ is necessarily equi-distant
(i.e., the distance between any two distinct codewords is exactly $2$)
and can be augmented by a fourth codeword while still remaining
equi-distant. Hence, we assume that each constituent code
is either the $[3,1,3]$ repetition code
or the $[3,2,2]$ parity code (which is equi-distant).
We can therefore view $\varcode$ as a concatenated code over $F$
comprising an outer code $\Code$ of length $\ell = N/3$
and size $M$ over a \emph{mixed alphabet}, namely:
\[
\Code \subseteq F^t \times (F^2)^{\ell-t} ,
\]
for some $t \le \ell$. The first $t$ (binary) coordinates
are mapped to an inner code which is the $[3,1,3]$ repetition code,
and the remaining $\ell-t$ (quaternary) coordinates 
are mapped to the $[3,2,2]$ parity code.
Writing each codeword of $\Code$ as $(\bldc \; \bldc')$,
where $\bldc \in F^t$ and $\bldc' \in (F^2)^{\ell-t}$,
the following inequality must hold for any two distinct codewords
$(\bldc_1 \; \bldc'_1), (\bldc_2 \; \bldc'_2)\in \Code$
to maintain the minimum distance $d$ of $\varcode$:
\begin{equation}
\label{eq:distance}
3 \distance_F(\bldc_1,\bldc_2)
+ 2 \distance_{F^2}(\bldc'_1,\bldc'_2) \ge d ,
\end{equation}
where $\distance_Q(\cdot,\cdot)$ denotes Hamming distance
over the alphabet $Q$.
Regarding now $\varcode$ as an element in
a $(\delta,3)$-LRC sequence
where $t/\ell$ converges to $\tau \in [0,1]$,
the condition~(\ref{eq:distance}) becomes
\begin{equation}
\label{eq:distance-asymptotic}
\tau \cdot \frac{\distance_F(\bldc_1,\bldc_2)}{t} 
+ \frac{2}{3}
(1-\tau) \cdot \frac{\distance_{F^2}(\bldc'_1,\bldc'_2)}{\ell-t} 
\ge \delta .
\end{equation}
We are interested in finding the value of $\tau$ for which
the rate of the LRC sequence is maximized,
subject to satisfying~(\ref{eq:distance-asymptotic}).
We do this using the following lemma;
the function $\delta \mapsto \Rate_{\LP;Q}(\delta)$
stands for the (second) linear-programming bound,
due to Aaltonen~\cite[p.~141]{Aaltonen},
on the rate of code sequences over a finite alphabet $Q$
(see also~\cite[\Theorem~1]{BHL1}).

\begin{lemma}
\label{lem:BHL}
Let $Q$ and $Q'$ be finite alphabets
and fix $\beta, \beta' \in \Realfield^+$
and $\tau \in [0,1]$.
Let $(\Code_i)_{i=1}^\infty$ be a code sequence such that
\[
\Code_i \subseteq Q^{t_i} \times (Q')^{\ell_i-t_i},
\]
where $\lim_{i \rightarrow \infty} t_i/\ell_i = \tau$.
Suppose also that
the (weighted) r.m.d.\ of the sequence is at least $\delta$,
where, for the purpose of computing distances,
each coordinate of $Q$ (respectively, $Q'$) contributes
$\beta$ (respectively, $\beta'$) to the distance.
Then
\begin{eqnarray}
\varlimsup_{i \rightarrow \infty}
\frac{\log |\Code_i|}{\ell_i}
\!\!\! & \le & \!\!\!
\min_{(\theta,\theta')}
\Bigl\{
\tau \cdot \Rate_{\LP;Q}(\theta) \log |Q| \nonumber \\
\label{eq:BHL}
&&
{}
+ (1-\tau) \cdot \Rate_{\LP;Q'}(\theta') \log |Q'| \Bigr\} ,
\end{eqnarray}
where the minimum is taken over
$(\theta,\theta') \in [0,1]^2$ such that
\[
\tau \cdot \beta \cdot \theta
+ (1-\tau) \cdot \beta' \cdot \theta' = \delta .
\]
\end{lemma}

A special case of this lemma, for the case $Q = Q'$
and $\beta = \beta'$, was proved by Ben-Haim and Litsyn
in~\cite{BHL1}.\footnote{%
This is not stated so explicitly in~\cite{BHL1},
but this is what \Theorem~7 in~\cite{BHL1} reduces to
when Eq.~(36) in that paper is plugged into Eq.~(39)
and the result is then plugged into Eq.~(40).}
For this case, Lemma~\ref{lem:BHL} states that
the bound $\delta \mapsto \Rate_{\LP;Q}(\delta)$
can be improved by taking its lower convex envelope
(it turns out that this bound is not convex
for general $Q$). With little effort,
\Theorem~5 in~\cite{BHL1} can be adapted
to the case of mixed alphabets and
weighted distance~\cite{BHL2}.\footnote{%
Specifically, one only needs to modify the definition of
the function $f$ in that theorem so that
the first multiplicand therein is
$\beta (a_1 + b_1 - x_1 - y_1) + \beta' (a_2 + b_2 - x_2 - y_2)$.
The proof, as is, holds also when
the prefixes and suffixes---of lengths $n_1$ and $n_2$---of
the codewords are over different alphabets.}

Applying Lemma~\ref{lem:BHL}
with $Q = \Finitefield_2$,
$Q' = \Finitefield_4$, $\beta = 1$, and $\beta' = 2/3$,
we have verified numerically
that the expression~(\ref{eq:BHL}) is maximized when $\tau = 0$.
Hence, we get the following upper
bound on the rate $R$ of any all-disjoint $(\delta,3)$-LRC sequence
over $F$ in which the repair groups are all of size $3$:
\begin{equation}
\label{eq:n=3}
R \le \frac{2}{3}
\cdot \Rate_{\LP;\Finitefield_4}\left( \frac{3 \delta}{2} \right) .
\end{equation}

It remains to consider the case where the $(N,M,d)$ LRC $\varcode$
has repair groups of size $2$, namely,
some $s$ out of the $t$ coordinates of the outer code
$\Code \; (\subseteq F^t \times (F^2)^{\ell-t})$
map to the $[2,1,2]$ repetition code.
By adding $s$ information bits to $\Code$
(thereby increasing its size by a factor of $2^s$)
we can obtain a new concatenated code $\varcode^*$
in which we replace all
the inner instances of the $[2,1,2]$ repetition code
by instances of the $[3,2,2]$ parity code. 
Doing so, the overall code length becomes $N + s$,
the minimum distance remains unchanged
(and, so, the r.m.d.\ reduces by a factor of $N/(N+s)$),
and the rate becomes
\begin{equation}
\label{eq:newrate}
\frac{(\log_2 M) + s} {N + s}
= \frac{(\log_2 M)/N + (s/N)} {1 + (s/N)} .
\end{equation}
Switching to a code sequence $(\varcode_i)_{i=1}^\infty$
where $\lim_{i \rightarrow \infty} s_i/N_i = \sigma$,
we get from~(\ref{eq:newrate}) the following relationship
between the rate $R$ of the sequence and the rate $R^*$
of $(\varcode^*_i)_{i=1}^\infty$:
\[
R^* = \frac{R + \sigma}{1 + \sigma} .
\]
On the other hand,
the bound~(\ref{eq:n=3}) applies to $(\varcode^*_i)_{i=1}^\infty$,
namely:
\[
R^* \le
\frac{2}{3} \cdot \Rate_{\LP;\Finitefield_4}
\left( \frac{3 \delta}{2 (1 + \sigma)} \right) .
\]
Combining the last two equations yields:
\[
R \le
\frac{2(1 + \sigma)}{3} \cdot \Rate_{\LP;\Finitefield_4}
\left( \frac{3 \delta}{2 (1 + \sigma)} \right) - \sigma.
\]
We have verified numerically
that this expression is maximized when $\sigma = 0$.
We therefore conclude that~(\ref{eq:n=3}) holds for
any all-disjoint $(\delta,3)$-LRC sequence over $F$.
The bound~(\ref{eq:n=3}) is depicted
in Figure~\ref{fig:q=2,n=3} (curve (g)).
\fi

\section{The general linearly recoverable case}
\label{sec:nondisjoint}

In this section, we present some (weaker) bounds for
$(\delta,n)$-LRC sequences that are not necessarily
all-disjoint (but they are still linearly recoverable).
Our results are based on deriving the asymptotic version
of the sphere-packing bound of~\cite{WZL}
(\Theorems~5 and~12 therein).

Let $\varcode$ be a linearly recoverable LRC of length $N$,
minimum distance $d$, and locality $r = n-1$ over $F$
and let $(\Local_j)_{j \in [N]}$ be
a list of repair groups of $\varcode$.
We assume hereafter in this section (without loss of generality)
that this list satisfies the following conditions.
\begin{list}{}{\settowidth{\labelwidth}{\textrm{P2)}}}
\item[R1)]
$|\Local_j| \le n$ for each index $j \in [N]$
(this condition follows directly from the locality).
\item[R2)]
For each $j \in [N]$, the repair group $\Local_j$ is
\ifPAGELIMIT
    minimal, i.e., no proper subset of it is a repair group for $j$.
\else
minimal in the sense that no proper subset of it
is a repair group for $j$; this, in turn, implies
that $\Local_j$ is a repair group for all $j' \in \Local_j$.
\fi
\item[R3)]
Each repair group $\Local_j$ contains at least one index $j'$
that is not contained in any repair group $\Local_i \ne \Local_j$
(otherwise, $\Local_j$ can be spared).
\end{list}

\ifPAGELIMIT
    For each repair group $\Local_j$ among
    the distinct repair groups of $\varcode$,
\else
Considering now only the set of distinct repair groups
$\{ \Local_j \}_{j \in [\ell]}$ of $\varcode$,
for each repair group $\Local_j$,
\fi
the constituent code $(\varcode)_{\Local_j}$ is contained in
a linear $[|\Local_j|,|\Local_j|{-}1]$ code over $F$,
which we denote by $\code_j$; moreover, by condition~(R2),
the code $\code_j$ has minimum distance~$2$.
An ambient space $\Ambient$ of $\varcode$ is
defined as in~(\ref{eq:ambient});
by condition~(R3),
the containment in~(\ref{eq:ambient}) holds in fact with equality
\ifPAGELIMIT
    for all $j$.
\else
for all $j \in [\ell]$.
\fi

The method of~\cite{WZL} for obtaining
a (non-asymptotic) sphere-packing-type
bound on the rate of $\varcode$ is based on shortening $\varcode$ on
the set of coordinates, $\AltSet$, on which repair groups intersect,
thereby reducing to the all-disjoint case.
Denoting $\nu = |\AltSet|/N$,
the resulting shortened code, $\widehat{\varcode}$,
is an all-disjoint linearly recoverable LRC
of length $(1-\nu)N$ and minimum distance${} \ge d$,
with the following $\ell$ distinct repair groups:
\[
\widehat{\Local}_j = \Local_j \setminus \AltSet ,
\quad j \in [\ell] .
\]
Now, on the one hand, we have:
\ifPAGELIMIT
    \begin{equation}
    \label{eq:AltSet1}
    \ell \cdot n \ge 2 |\AltSet| + (N - |\AltSet|) .
    \end{equation}
\else
\begin{eqnarray}
\ell \cdot n
& \ge &
\sum_{j \in [\ell]} |\Local_j|
= \sum_{t \in [N]}
\left| \bigl\{  j \in [\ell] \,:\, t \in \Local_j \bigr\} \right|
\nonumber \\
& = &
\sum_{t \in \AltSet}
\left| \bigl\{  j \,:\, t \in \Local_j \bigr\} \right|
+ \sum_{t \in [N] \setminus \AltSet}
\left| \bigl\{  j \,:\, t \in \Local_j \bigr\} \right|
\nonumber \\
\label{eq:AltSet1}
& \ge &
2 |\AltSet| + (N - |\AltSet|) .
\end{eqnarray}
\fi
On the other hand, by condition~(R3), we also have:
\begin{equation}
\label{eq:AltSet2}
N - |\AltSet| \ge \ell .
\end{equation}
   From~(\ref{eq:AltSet1}) and~(\ref{eq:AltSet2}) we get:
\begin{equation}
\label{eq:ell}
\frac{1+\nu}{n} \le \frac{\ell}{N} \le 1-\nu
\end{equation}
and, in particular,
\ifPAGELIMIT
    $0 \le \nu \le (n{-}1)/(n{+}1)$.
\else
\[
0 \le \nu \le \frac{n-1}{n+1} .
\]
\fi
Moreover, from~(\ref{eq:ell}) we get
that the average size of the distinct (and disjoint) repair groups of
$\widehat{\varcode}$ satisfies
\begin{equation}
\label{eq:average}
\frac{1}{\ell}
\sum_{j \in [\ell]} |\widehat{\Local}_j|
= \frac{(1-\nu) N}{\ell}
\le \frac{1{-}\nu}{1{+}\nu} \cdot n = \mu .
\end{equation}

The rate $R = (\log_q |\varcode|)/N$
of $\varcode$ is related to that of $\widehat{\varcode}$ by:
\begin{equation}
\label{eq:hatD}
R = \frac{1}{N} \cdot \log_q |\varcode|
\le
\frac{1}{N}
\left( |\AltSet| + \log_q |\widehat{\varcode}| \right)
=
\nu + (1-\nu) \cdot \widehat{R} .
\end{equation}
Thus, given $\nu$, any upper bound on $\widehat{R}$
implies an upper bound on $R$, and the dependence on $\nu$ can then be
eliminated by maximizing the latter bound over
$\nu \in [0,(n{-}1)/(n{+}1)]$.

Using the above strategy,
we next turn to adapting our previous bounds
to linearly recoverable $(\delta,n)$-LRC sequences
that are not necessarily all-disjoint.
\ifPAGELIMIT
\else
We will make use of the following notation.
\fi
For $n \in \Realfield^+$ and
$\omega \in \Realfield_{\ge 0}$, define:
\begin{equation}
\label{eq:WZL2}
\widehat{\Rate}_0(\omega,n) =
\max_\nu
\left\{
\nu + (1-\nu) \cdot
\overline{\Rate}_0 \left( \frac{\omega}{1{-}\nu}, n, 
\frac{1{-}\nu}{1{+}\nu} \cdot n \right) \right\} ,
\end{equation}
where $\overline{\Rate}_0(\cdot,\cdot,\cdot)$ is
as defined in~(\ref{eq:WZL1})
and the maximum is taken over all $\nu \in [0,(n{-}1)/(n{+}1)]$.

\ifPAGELIMIT
\else
\begin{remark}
\label{rem:computations2}
Substituting $\nu = 0$ and $\nu = (n{-}1)/(n{+}1)$
in the objective function in~(\ref{eq:WZL2}),
we get the following \emph{lower} bound on
$\widehat{\Rate}_0(\omega,n)$:
\[
\widehat{\Rate}_0(\omega,n)
\ge
\max \left\{
\Rate_0(\omega,n),  \frac{n{-}1}{n{+}1} \right\} .
\]
Based on our numerical evidence,
we conjecture that this lower bound is tight
(see also Remark~\ref{rem:WZL2} in Appendix~\ref{sec:R0properties}).\qed
\end{remark}
\fi

The following proposition is a (weaker) counterpart
of \Proposition~\ref{prop:R0} for the case where repair groups
can intersect.

\begin{proposition}
\label{prop:R0-nondisjoint}
Given $n \in \Integers^+$,
let $(\Ambient_i)_{i=1}^\infty$ be
an infinite sequence of codes over $F$
where each $\Ambient_i$ is
a linear code of length $N_i$ over $F$ defined by
\[
(\Ambient_i)_{\Local_{i,j}} = \code_{i,j} ,
\quad j \in [N_i] ,
\]
with the list $(\Local_{i,j})_{j \in [N_i]}$ satisfying
conditions (R1)--(R3)
and each constituent code $\code_{i,j}$ being
a linear $[|\Local_{i,j}|,|\Local_{i,j}|{-}1,2]$ code over $F$.
Then for any nonnegative real sequence $(\omega_i)_{i=1}^\infty$
such that $\varliminf_{i \rightarrow \infty} \omega_i = \omega$:
\[
\varlimsup_{i \rightarrow \infty}
\frac{1}{N_i}
\log_q
\frac{|\Ambient_i|}{|\Ambient_i(\omega_i)|} \le
\widehat{\Rate}_0(\omega,n) .
\]
\end{proposition}

\begin{proof}
For each $i \in \Integers^+$, let $\widehat{\Ambient}_i$
be obtained by shortening $\Ambient_i$
on the set of coordinates on which
repair groups (i.e., subsets) $\Local_{i,j}$
intersect. Denoting by $\nu_i \; (\in [0,(n{-}1)/(n{+}1)])$
the fraction of removed coordinates, by possibly restricting
to a subsequence of the codes,
we can assume that $(\nu_i)_{i=1}^\infty$ converges to a limit $\nu$
and that the length, $m_i = (1{-}\nu_i)N_i$,
of $\widehat{\Ambient}_i$ strictly increases with $i$.
We then get that the code sequence
$(\widehat{\Ambient}_i)_{i=1}^\infty$
satisfies the conditions of
\Proposition~\ref{prop:R0} (and, hence, of Lemma~\ref{lem:mu});
\ifPAGELIMIT
    moreover, the average length of
    the constituent codes of each $\widehat{\Ambient}_i$
\else
in particular, since each $\code_{i,j}$ has minimum distance $2$,
the respective (shortened) constituent codes of each
$\widehat{\Ambient}_i$ all (still) have redundancy $1$.
Moreover, the average length of those constituent codes
\fi
is bounded from above by
\[
\frac{1{-}\nu_i}{1{-}\nu_i} \cdot n
\]
(as in~(\ref{eq:average})).
In addition, 
\begin{equation}
\label{eq:rate}
\frac{1}{N_i}
\log_q |\Ambient_i|
\le
\nu_i + (1-\nu_i) \cdot
\frac{1}{m_i} \log_q |\widehat{\Ambient}_i|
\end{equation}
(as in~(\ref{eq:hatD})) and,
for any nonnegative real sequence $(\omega_i)_{i=1}^\infty$:
\begin{eqnarray}
\lefteqn{
\frac{1}{N_i}
\log_q |\Ambient_i(\omega_i)|
\ge
\frac{1}{N_i}
\log_q |\widehat{\Ambient}_i(\omega_i/(1{-}\nu_i))|
} \makebox[5ex]{} 
\nonumber \\
\label{eq:logvolume}
& = &
(1-\nu_i) \cdot
\frac{1}{m_i}
\log_q |\widehat{\Ambient}_i(\omega_i/(1{-}\nu_i))| .
\end{eqnarray}
Therefore,
\begin{eqnarray*}
\lefteqn{
\varlimsup_{i \rightarrow \infty}
\frac{1}{N_i}
\log_q \frac{|\Ambient_i|}{|\Ambient_i(\omega_i)|}
} \makebox[5ex]{} \\
& \stackrel{\textrm{(\ref{eq:rate})+(\ref{eq:logvolume})}}{\le} &
\nu +
(1-\nu) \cdot
\varlimsup_{i \rightarrow \infty}
\frac{1}{m_i}
\log_q
\frac{|\widehat{\Ambient}_i|}%
                          {|\widehat{\Ambient}_i(\omega_i/(1{-}\nu_i))|}
\\
& \stackrel{\textrm{Lemma~\ref{lem:mu}}}{\le} &
\nu + (1-\nu) \cdot
\overline{\Rate}_0 \left( \frac{\omega}{1{-}\nu},n,
\frac{1{-}\nu}{1{+}\nu} n \right) \\
& \stackrel{\textrm{(\ref{eq:WZL2})}}{\le} &
\widehat{\Rate}_0(\omega,n) ,
\end{eqnarray*}
where $\omega = \varliminf_{i \rightarrow \infty}\omega_i$.
\end{proof}

The following sphere-packing bound is proved similarly
to \Theorem~\ref{thm:spherepacking},
except that we use \Proposition~\ref{prop:R0-nondisjoint}
instead of \Proposition~\ref{prop:R0}.

\begin{theorem}
\label{thm:spherepacking-nondisjoint}
The rate of any linearly recoverable $(\delta,n)$-LRC sequence over $F$
is bounded from above by
\[
\widehat{\Rate}_\SP(\delta,n) =
\ifPAGELIMIT
    \widehat{\Rate}_0(\delta/2,n)
\else
\widehat{\Rate}_0 \left( \frac{\delta}{2} , n \right)
\fi
.
\]
\end{theorem}

Next are our (weaker) versions
of \Theorems~\ref{thm:bound1} and~\ref{thm:bound2}
for general linearly recoverable LRC sequences.

\begin{theorem}
\label{thm:bound1-nondisjoint}
The rate of any linearly recoverable $(\delta,n)$-LRC sequence over $F$
is bounded from above by
\ifPAGELIMIT
    \begin{eqnarray*}
    \widehat{\Rate}_1(\delta,n)
    & = &
    \inf_{\tau \in (0,1)} \min_{(\theta,\theta')}
    \Bigl\{ 
    \tau \cdot \widehat{\Rate}_0(\theta/2,n)
    \\
    &&
    \quad \quad \quad \quad \quad \quad
    {}
    + 
    (1-\tau) \cdot \Rate_\opt(\theta',n)
    \Bigr\} ,
    \end{eqnarray*}
    where $(\theta,\theta')$ ranges over all pairs
    that satisfy~(\ref{eq:theta}).
\else
\begin{eqnarray*}
\widehat{\Rate}_1(\delta,n)
& = &
\inf_{\tau \in (0,1)} \min_{(\theta,\theta')}
\Biggl\{ 
\tau \cdot \widehat{\Rate}_0 \left( \frac{\theta}{2}, n \right)
\\
&&
\quad \quad \quad \quad \quad \quad
{}
+ 
(1-\tau) \cdot \Rate_\opt(\theta',n)
\Biggr\} ,
\end{eqnarray*}
where the (inner) minimum is taken over
all pairs $(\theta,\theta')$ in $[0,(q{-}1)/q]^2$
that satisfy~(\ref{eq:theta}).
\fi
\end{theorem}

\ifPAGELIMIT
    The proof resembles that of \Theorem~\ref{thm:bound1}
    (details can be found in~\cite{RothFull}).
\else
\begin{proof}
The proof resembles that of \Theorem~\ref{thm:bound1},
yet requires several modifications which are described below.
Given a linearly recoverable $(\delta,n)$-LRC sequence
$(\varcode_i)_{i=1}^\infty$ over $F$,
we let $(\Local_{i,j})_{j \in [N_i]}$
be a list of repair groups of $\varcode_i$
that satisfies conditions (R1)--(R3).
Letting $\{ \Local_{i,j} \}_{j \in [\ell_i]}$
be the set of distinct repair groups in the list,
we define (as in the proof of \Theorem~\ref{thm:bound1})
\[
t_i = \left|
{\textstyle\bigcup_{j \in [\ell'_i]}} \Local_{i,j} \right| ,
\]
where $\ell'_i$ is the smallest in $[\ell_i]$ so that
$t_i \ge \tau N_i$.
And by possibly permuting the coordinates
of $\varcode_i$, we assume that
$\bigcup_{j \in [\ell'_i]} \Local_{i,j} = [t_i]$.

We now specify the changes to the proof of \Theorem~\ref{thm:bound1}.
First, we need to modify
the argument that leads to Eq.~(\ref{eq:bound1-1}),
since the decomposition~(\ref{eq:Subset}) may no longer hold.
Still, due to condition~(R3), the size of the following set is
the same for all $\bldy' \in \Ambient'_i$:
\[
\left\{
\bldy'' \in F^{N_i-t_i} \,:\,
(\bldy' \; \bldy'') \in \Ambient_i \right\} .
\]
Thus,
\[
\frac{|\Ambient_i|}{|\Subset|}
= \frac{|\Ambient'_i|}{|\Ambient'_i(w_i/t_i)|} ,
\]
which justifies~(\ref{eq:bound1-1}).

Secondly, we need to weaken Eq.~(\ref{eq:bound1-2}),
since the code $\varcode''_i$ is now not necessarily an LRC.
Specifically, (\ref{eq:bound1-2}) now becomes
\[
|\varcode^*_i| = |\varcode''_i| \le M_\opt(N_i - t_i, d_i - 2w_i) ,
\]
where $M_\opt(N,D)$ denotes the largest size
of any code over $F$ with length $N$ and minimum distance $D$.
Accordingly, from this point in the proof of \Theorem~\ref{thm:bound1},
we change the instances of
$\Rate_\LRC(\cdot,n)$ into $\Rate_\opt(\cdot)$.

Finally, using \Proposition~\ref{prop:R0-nondisjoint},
we change~(\ref{eq:ratio}) into
\[
\varlimsup_{i \rightarrow \infty}
\frac{1}{t_i}
\log_q \frac{|\Ambient'_i|}{|\Ambient'_i(w_i/t_i)|}
\le \widehat{\Rate}_0 \left( \frac{\theta}{2} , n \right) .
\]
\end{proof}

\fi
Below is our variant of \Theorem~\ref{thm:bound2}, which
is proved using \Proposition~\ref{prop:R0-nondisjoint}
instead of \Proposition~\ref{prop:R0}.

\begin{theorem}
\label{thm:bound2-nondisjoint}
The rate of any linearly recoverable $(\delta,n)$-LRC sequence over $F$
is bounded from above by
\[
\widehat{\Rate}_2(\delta,n) = \min_{\omega \in [\delta/2,(q{-}1)/q]}
\Bigl\{ \widehat{\Rate}_0(\omega,n) + \Rate_\opt(\delta,\omega)\Bigr\} .
\]
\end{theorem}

Notice that unlike the all-disjoint case,
we \emph{cannot} substitute $\widehat{\Rate}_2(\theta',n)$
for $\Rate_\opt(\theta',n)$ in \Theorem~\ref{thm:bound1-nondisjoint}.

\ifPAGELIMIT
\else
The various bounds are plotted
in Figure~\ref{fig:q=2,n=4-nondisjoint} for $q = 2$ and $n = 4$.
Curve~(a) is the sphere-packing bound
of \Theorem~\ref{thm:spherepacking-nondisjoint}.
Curve~(b) is identical to its counterpart in Figure~\ref{fig:q=2,n=4}
and, as it turns out, curve~(c) is the same as in that figure too
(namely, $\widehat{\Rate}_1(\delta,n) = \Rate_1(\delta,n)$
for the examined parameters); this is due to the fact
that the minimum in \Theorem~\ref{thm:bound2-nondisjoint}
is attained at values $\theta$
where $\widehat{\Rate}_0(\theta,n) = \Rate_0(\theta,n)$.
When plotting curve~(d), we have taken
the minimum of $\widehat{\Rate}_2(\delta,n)$ and $\Rate_\LP(\delta,n)$.
\fi
Some values of the bounds are listed in
Table~\ref{tab:q=2,n=4-nondisjoint}, where
entries that differ from those in Table~\ref{tab:q=2,n=4} 
are marked in
\ifPAGELIMIT
    italics (full plots can be found in~\cite{RothFull}).
\else
italics.
\fi
There is still a range where
\Theorem~\ref{thm:bound2-nondisjoint} yields the best upper bound,
yet this range is smaller compared to the all-disjoint case.

\ifPAGELIMIT
\else
\begin{remark}
\label{rem:bound2-nondisjoint}
A second variant of \Theorem~\ref{thm:bound2} can be obtained
by first shortening the codes in a given LRC sequence
on the intersections of repair groups, and then applying
\Theorem~\ref{thm:bound2} to the resulting (all-disjoint) LRC sequence.
This yields the upper bound
\begin{eqnarray*}
\lefteqn{
\widehat{\Rate}_3(\delta,n)
= \max_\nu \Biggl\{ \nu + (1 - \nu) \cdot
\min_\omega
\biggl\{
\overline{\Rate}_0
\Bigl(\omega,n, \frac{1{+}\nu}{1{-}\nu} \cdot n \Bigr)
} \makebox[27ex]{} \\
\\
&&
\quad {}
+ \Rate_\opt \Bigl( \frac{\delta}{1{-}\nu} , \omega \Bigr)
\biggr\}
\Biggr\} ,
\end{eqnarray*}
where the outer maximum is over $\nu \in [0,(n{-}1)/(n{+}1)]$
and the inner minimum is over
$\omega \in [\delta/(2{-}2\nu),(q{-}1)/q]$.
Yet at least for the parameters that we have tested,
we have observed no difference between the values of
$\widehat{\Rate}_2(\delta,n)$ and $\widehat{\Rate}_3(\delta,n)$.\qed
\end{remark}
\fi

\begin{table}[hbt]
\caption{Values of the bounds for $q = 2$ and $n = 4$
without the all-disjoint constraint.}
\label{tab:q=2,n=4-nondisjoint}
\ifPAGELIMIT
    \vspace{-3ex}
\fi
\[
\renewcommand{\arraystretch}{1.1}
\begin{array}{ccccc}
\hline\hline
\quad\quad \delta \quad\quad &
\ifPAGELIMIT
    \widehat{\Rate}_\SP & \Rate_\CM & \widehat{\Rate}_1 &
    \widehat{\Rate}_2\rule{0ex}{2.5ex} \\
\else
\mathrm{(a)}&\mathrm{(b)}&\mathrm{(c)}&\mathrm{(d)} \\
\fi
\hline
0.07 & 0.6133 & 0.6317 & 0.6131 & 0.6079 \\
0.10 & \mathit{0.6000} & 0.5809 & 0.5643 & \mathit{0.6000} \\
0.15 & \mathit{0.6000} & 0.4964 & 0.4830 & \mathit{0.6000} \\
0.30 & \mathit{0.6000} & 0.2427 & 0.2391 & \mathit{0.6000} \\
\hline\hline
\end{array}
\]
\end{table}

\section*{Acknowledgment}

I would like to thank Yael Ben-Haim and Simon Litsyn
for helpful
\ifPAGELIMIT
    discussions.
\else
discussions and for making me aware of
the results of~\cite{BCCST}.

\ifIEEE
   \appendices
\else
   \section*{$\,$\hfill Appendices\hfill$\,$}
   \appendix
\fi

\section{Skipped proofs}
\label{sec:skippedproofs}

\begin{proof}[Proof of Lemma~\ref{lem:gamma}]
By convexity, for all $z \in (0,1]$ we have
$\Expected \{ z^X \} \ge z^{\Expected \{ X \}}$,
with equality holding when $z = 1$; hence,
$\gamma(u) = 1$ when $u \ge \Expected \{ X \}$.
When $u \le x_{\min}$,
the infimum in~(\ref{eq:gamma}) is attained at $z \rightarrow 0$
and, so, $\gamma(u) = 0$ when $u < x_{\min}$
and $\gamma(u) = p(x_{\min})$ when $u = x_{\min}$.
For $x_{\min} < u < \Expected \{ X \}$,
we differentiate $g_u(z)$ with respect to $z$ to obtain
\[
g'_u(z) = z^{x_{\min} - u - 1} \cdot f_u(z) ,
\]
where
\[
f_u(z) = \sum_{x \in \Set}
(x - u) \cdot p(x) \cdot z^{x - x_{\min}} .
\]
Thus, $f_u(1) = \Expected \{ X \} - u > 0$
and $f_u(0) = (x_{\min} - u) \cdot p(x_{\min}) < 0$,
which implies that the infimum in~(\ref{eq:gamma}) is
a proper minimum attained at an (interior) point $z_u \in (0,1)$;
the point $z_u$ satisfies
\[
1 \ge g_u(z_u)
> z_u^{x_{\min} - u} \cdot p(x_{\min}) ,
\]
i.e.,
\[
z_u > p(x_{\min})^{1/(u - x_{\min})} .
\]
This means that $z_u$ is bounded away from zero whenever
$u$ is bounded away from $x_{\min}$.
Defining $z_u = 1$ for $u \ge \Expected \{ X \}$,
for any $v > u \ge x_{\min}$ we have
\[
\gamma(u) \le g_u(z_v)
\le z_v^{u-v} \cdot g_u(z_v) = g_v(z_v) = \gamma(v),
\]
with the second inequality being strict when $v < \Expected \{ X \}$.
Hence, $u \mapsto \gamma(u)$
is strictly increasing when $x_{\min} \le u < \Expected \{ X \}$.
On the other hand, we also have
\[
\gamma(u) = g_u(z_u) 
= z_u^{v-u} g_v(z_u)
\ge z_u^{v-u} \gamma(v),
\]
which, combined with $\gamma(u) \le \gamma(v)$,
means that $u \mapsto \gamma(u)$ is continuous
when $u > x_{\min}$. Moreover, it is right-continuous
at $u = x_{\min}$, since
\[
\lim_{\varepsilon \rightarrow 0^+}
\gamma(x_{\min}+\varepsilon)
\le \lim_{\varepsilon \rightarrow 0^+}
g_{x_{\min}+\varepsilon}(\varepsilon) = p(x_{\min})
= \gamma(x_{\min}) .
\]
Finally, the concavity of $u \mapsto \log \gamma(u)$
follows from~(\ref{eq:cramer}) and
\begin{eqnarray*}
\lefteqn{
\Prob \left\{ \frac{1}{\ell_1{+}\ell_2} \sum_{i=1}^{\ell_1 + \ell_2} X_i
\le
\frac{\ell_1}{\ell_1{+}\ell_2} u_1
+
\frac{\ell_2}{\ell_1{+}\ell_2} u_2 \right\}
}\\
& \ge &
\!\!
\Prob \left\{
\frac{1}{\ell_1} \sum_{i=1}^{\ell_1} X_i \le u_1 \right\}
\cdot
\Prob \left\{
\frac{1}{\ell_2} \sum_{i=\ell_1+1}^{\ell_1 + \ell_2} X_i
\le u_2 \right\} .
\end{eqnarray*}
\end{proof}

Turning to the proof of Lemma~\ref{lem:CW}, it makes use of
the following theorem, which was proved in~\cite[\Theorem~1]{BCCST}
for the special case of the binary alphabet.
For completeness, we provide a proof of the theorem for general $q$
right after the proof of Lemma~\ref{lem:CW}.

\begin{theorem}
\label{thm:BCCST}
For any $\delta \in [0,(q{-}1)/q]$,
the mapping $\omega \mapsto \Rate_\opt(\delta,\omega)$
is non-decreasing on $\omega \in [0,(q{-}1)/q]$.
\end{theorem}

\begin{proof}[Proof of Lemma~\ref{lem:CW}]
Obviously,
$\Rate_\opt(\delta,\omega) \le \Rate_\opt(\delta,{\le}\omega)$.
To prove the inequality in the other direction,
Let $(\varcode_i)_{i=1}^\infty$
be a code sequence with r.m.d.${} \ge \delta$,
with the (length-$N_i$) codewords of each $\varcode_i$
all having weight at most $w_i$, such that
$\varlimsup_{i \rightarrow \infty} w_i/N_i \le \omega$.
Letting $\varcode^*_i$ be a largest constant-weight subcode
of $\varcode_i$ (of codeword weight $w^*_i \le w_i$),
we have $|\varcode^*_i| \ge |\varcode_i|/(w_i + 1)$ and
$\varlimsup_{i \rightarrow \infty} w^*_i/N_i = \omega^* \le \omega$
(and by possibly restricting to a subsequence of
$(\varcode^*_i)_{i=1}^\infty$ we can assume
that $\omega^*$ is a proper limit of $(w_i/N_i)_{i=1}^\infty$).
Thus,
\[
\Rate_\opt(\delta,{\le}\omega) \le
\sup_{\omega^* \in [0,\omega]} \Rate_\opt(\delta,\omega^*) .
\]
By \Theorem~\ref{thm:BCCST} we then get that when
$\omega \in [0,(q{-}1)/q]$,
the supremum is attained at $\omega^* = \omega$.
\end{proof}

\begin{proof}[Proof of \Theorem~\ref{thm:BCCST}]
Given an Abelian group $F$ of size $q$,
fix $\omega, \theta \in [0,(q{-}1)/q]$
and let $\code$ be a constant-weight code of length $N$
and minimum distance $d$ over $F$ with codeword weight
$\lfloor \omega N \rfloor$. We show that there exists
$\bldy \in F^N$ such that the translation $\bldy + \code$ 
contains a constant-weight subcode $\code^*$ of size
\begin{equation}
\label{eq:sizeratio}
|\code^*| \ge \frac{|\code|}{(N{+}1)^3}
\end{equation}
and of codeword weight $\lfloor \omega^* N \rfloor$, where
\begin{equation}
\label{eq:omegastar}
\omega^* = \omega + \theta \left( 1 - \frac{q \, w}{q{-}1} \right) .
\end{equation}
The result will follow by observing
that $\omega^*$ ranges over $[\omega,(q{-}1)/q]$
as $\theta$ ranges over $[0,(q{-}1)/q]$.
Hereafter in the proof, we assume that $\omega$ and $\theta$
are rational numbers and $N$ is such that $\omega N$, $\theta N$,
and $\theta \omega N/(q{-}1)$ are integers
(it is easy to see that those assumptions are allowed
in order to obtain the asymptotic result that we seek).

For any codeword $\bldc \in \code$,
let $\Typical(\bldc)$ denote the set of all words $\bldy \in F^N$
that satisfy the following conditions.
\begin{list}{}{\settowidth{\labelwidth}{\textit{(Y3)}}}
\item[Y1)]
The subword $\bldy' \in F^{\omega N}$ of $\bldy$ that is indexed by
the support of $\bldc$ has weight $\theta \omega N$,
\item[Y2)]
a fraction $1/(q{-}1)$ of the nonzero entries of $-\bldy'$ agree 
with the respective entries in $\bldc$, and---
\item[Y3)]
the subword $\bldy'' \in F^{(1-\omega) N}$ of $\bldy$
that is indexed by the zero entries of $\bldc$
has weight $\theta (1{-}\omega) N$.
\end{list}

It follows from these conditions that
the weight of each $\bldy \in \Typical(\bldc)$ is $\theta N$
and that for each $\bldy \in \Typical(\bldc)$,
the weight of $\bldy + \bldc$ is
\[
\left(
\frac{q{-}2}{q{-}1} \cdot \theta \omega
+ (1{-}\theta) \omega + \theta(1{-}\omega) \right) N
\stackrel{\textrm{(\ref{eq:omegastar})}}{=} \omega^* N .
\]
We also have:
\begin{eqnarray*}
|\Typical(\bldc)|
& = &
\binom{\omega N}{\theta \omega N} \\
&&
\quad {}
\cdot
\binom{\theta \omega N}{\theta \omega N/(q{-}1)} 
\cdot \left( q{-}2 \right)^{\theta \omega N (q-2)/(q-1)} \\
&&
\quad {} \cdot \binom{(1{-}\omega)N}{\theta (1{-}\omega)N}
\cdot (q{-}1)^{\theta (1{-}\omega) N} ,
\end{eqnarray*}
where the first term is the number of possible supports 
of a subword $\bldy'$ that satisfies condition~(Y1)
and, for each such support, the second term counts
the number of subwords $\bldy'$ that satisfy condition~(Y2).
The third term counts the number of subwords $\bldy''$
that satisfy condition~(Y3).

We now apply the following well known approximation
of the binomial coefficients:
\[
\entropy \left( \frac{k}{m}\right) - \frac{\log \, (m{+}1)}{m}
\le
\frac{1}{m}
\log \binom{m}{k} \le \entropy \left( \frac{k}{m}\right) ,
\]
where $x \mapsto \entropy(x)$
is the (binary) entropy function
$- x \log x - (1{-}x) \log \, (1{-}x)$
(see~\cite[pp.~105--106]{Roth}).
We get:
\begin{eqnarray}
\frac{\log |\Typical(\bldc)|}{N}
& \ge &
\omega \cdot \entropy(\theta) \nonumber \\
&&
\quad {}
+ \theta \omega \left( \entropy \left( \frac{1}{q{-}1} \right)
+ \frac{q{-}2}{q{-}1} \log \, (q{-}2) \right) \nonumber \\
&&
\quad {}
+ (1{-}\omega)
\cdot \entropy(\theta) + \theta (1{-}\omega) \log \, (q{-}1)
\nonumber \\
&&
\quad {}
- \frac{3 \log \, (N{+}1)}{N} \nonumber \\
\label{eq:Typical}
& = &
\entropy(\theta) + \theta \log \, (q{-}1)
- \frac{3 \log \, (N{+}1)}{N} .
\end{eqnarray}
On the other hand, the set, $\YY$, of all words
$\bldy \in F^N$ of weight $\theta N$ has size
\[
|\YY| = \binom{N}{\theta N} (q{-}1)^{\theta N}
\]
and, so,
\begin{equation}
\label{eq:Y}
\frac{\log |\YY|}{N} 
\le \entropy(\theta) + \theta \log \, (q{-}1) .
\end{equation}
Combining~(\ref{eq:Typical}) and~(\ref{eq:Y})
we conclude that
\begin{equation}
\label{eq:typical}
\frac{\Typical(\bldc)}{|\YY|} \ge \frac{1}{(N{+}1)^3}.
\end{equation}
For $\bldy \in \YY$, let
\[
\code(\bldy)
= \Bigl\{ \bldc \in \code \,:\, \bldy \in \Typical(\bldc) \Bigr\} .
\]
The set $\code^*(\bldy) = \bldy + \code(\bldy)$
forms a constant-weight code of codeword weight $\omega^* N$
(and of the same minimum distance $d$ as $\code$).
Summing now on the size of $\code^*(\bldy)$ over
all $\bldy \in \YY$ yields:
\begin{eqnarray*}
\sum_{\bldy \in \YY} |\code^*(\bldy)|
& = &
\sum_{\bldy \in \YY} |\code(\bldy)| \\
& = & 
\left|
\bigl\{
(\bldy,\bldc) \in \YY \times \code \,:\, \bldy \in \Typical(\bldc)
\bigr\}
\right| \\
& = &
\sum_{\bldc \in \code} |\Typical(\bldc)|
\end{eqnarray*}
and, so,
\[
\frac{1}{|\YY|}
\sum_{\bldy \in \YY} |\code^*(\bldy)|
= 
\frac{1}{|\YY|}
\sum_{\bldc \in \code} |\Typical(\bldc)|
\stackrel{\textrm{(\ref{eq:typical})}}{\ge}
\frac{|\code|}{(N{+}1)^3} .
\]
Hence, there must be at least one word $\bldy \in \YY$
for which $\code^* = \code^*(\bldy)$ satisfies~(\ref{eq:sizeratio}).
\end{proof}

\section{Characterization of $\overline{\Rate}_0(\omega,n,\mu)$}
\label{sec:R0properties}

We show here how to compute the expression~(\ref{eq:WZL1}).
It is easy to see that this expression
equals $(n{-}1)/n$ when $\omega = 0$
(since $\bldvartheta$ is forced then to be all-zero
on the support of $\bldpi$) and it vanishes
when $\omega \ge (q{-}1)/q$ (by taking $\vartheta_s = \omega$
for all $s \in [n]$). Hence, we can assume from now on that
$\omega \in (0,(q{-}1)/q)$.

We introduce the following notation.
For $n \in \Integers^+$ and $\omega \in [0,(q{-}1)/q]$,
let $\zeta_n(\omega)$ denote a particular minimizing $z$ of
the expression for $\lambda(\omega,n)$ in~(\ref{eq:lambda}).
Also, for $n \in \Integers^+$, let the polynomials $P_n(z)$ 
and $Q_n(z)$ be defined by
\begin{eqnarray*}
P_n(z) & = &
- (q{-}1) \cdot
\bigl( (1 + (q{-}1)z)^{n-1} - (1{-}z)^{n-1} \bigr) \\
Q_n(z) & = & (1 + (q{-}1)z)^n + (q{-}1)(1{-}z)^n
\end{eqnarray*}
(note that $n \cdot P_n(z)$ is the derivative of $Q_n(z)$).
Then~(\ref{eq:lambda}) can be written as
\begin{eqnarray}
\label{eq:lambdaalt}
\lambda(\omega,n)
\!
& = & \!\!\! \inf_{z \in (0,1]} 
\Bigl\{ -\omega \log_q z - \frac{1}{n} + \frac{1}{n} \log_q Q_n(z)
\Bigr\} \\
\nonumber
& = &
-\omega \log_q \zeta_n(\omega)
- \frac{1}{n} + \frac{1}{n} \log_q Q_n \left( \zeta_n(\omega) \right) .
\end{eqnarray}
We have the following lemma.

\begin{lemma}
\label{lem:unique}
For $n > 1$ and $\omega \in [0,(q{-}1)/q]$,
the value $\zeta_n(\omega)$ is
the unique real root in $[0,1]$ of the polynomial
\[
U_{\omega,n}(z) = \omega \cdot Q_n(z) - z \cdot P_n(z) .
\]
Moreover,
the mapping $\omega \mapsto \zeta_n(\omega)$ is strictly increasing. 
\end{lemma}

\begin{proof}
It follows from the proof of Lemma~\ref{lem:gamma}
that $\zeta_n(0) = 0$ and $\zeta_n((q{-}1)/q) = 1$.
Assuming hereafter that $\omega \in (0,(q{-}1)/q)$,
it also follows from that proof that
$\zeta_n(\omega)$ is an (interior) point in $(0,1)$;
as such, it is a local minimum
of the objective function in~(\ref{eq:lambdaalt})
and, so, it equals a value $z$ at which
the derivative of that function (with respect to $z$) vanishes:
\begin{equation}
\label{eq:zetainverse}
\frac{\omega}{z} - \frac{P_n(z)}{Q_n(z)} = 0 .
\end{equation}
This equation, which is equivalent to
requiring that $U_{\omega,n}(z) = 0$,
can be rearranged into
\begin{equation}
\label{eq:implicit}
\frac{(1{-}\omega)(q{-}1) z - \omega}{(1{-}\omega)z + \omega}
= (q{-}1) \cdot \left(
\frac{1-z}{1 + (q{-}1)z}
\right)^{n-1} .
\end{equation}
In the range $z \in [0,1]$,
the left-hand side of~(\ref{eq:implicit}) is strictly increasing in $z$
(from the value $-1$ at $z = 0$
to $q{-}1 - q \omega$ at $z = 1$)
while the right-hand side of~(\ref{eq:implicit})
is strictly decreasing in $z$. Hence, for any
$\omega \in (0,(q{-}1)/q)$, there is (at most) one $z \in [0,1]$
that satisfies~(\ref{eq:implicit}),
and $\zeta_n(\omega)$ must then be that $z$.
Moreover, since the left-hand side of~(\ref{eq:implicit})
is a strictly decreasing expression in $\omega$,
the mapping $\omega \mapsto \zeta_n(\omega)$ is strictly increasing.
\end{proof}

\begin{remark}
\label{rem:unique}
By~(\ref{eq:zetainverse}), the inverse mapping
$z \mapsto \omega = \zeta_n^{-1}(z)$ is given by
\[
\zeta_n^{-1}(z) = \frac{z \cdot P_n(z)}{Q_n(z)} .
\]
Note also that Lemma~\ref{lem:unique} is false when $n = 1$.
In this case $\zeta_1(0)$ is arbitrary
(and $U_{0,1}(z)$ is identically zero) while
$\zeta_1(\omega) = 1$ when $\omega > 0$
(it is then a global---rather than local---minimum of
the objective function in~(\ref{eq:lambdaalt})).\qed
\end{remark}

We proceed to the characterization of the minimum in the inner
expression in~(\ref{eq:WZL1}).
We will use the notation $[n]^*$ for the set $[n] \setminus \{ 1 \}$.

\begin{lemma}
\label{lem:inner}
Given $n \in \Integers^+$, $\mu \in [n]$,
and $\omega \in (0,(q{-}1)/q)$,
let $\bldpi = (\pi_s)_{s \in [n]} \in \Realfield_{\ge 0}^n$
be a vector that satisfies conditions (P1)--(P2)
with support $\AltSet = \Support(\bldpi)$.
A minimizer of the inner expression in~(\ref{eq:WZL1})
under the constraint~(P3) is any vector
$\bldvartheta = (\vartheta_s)_{s \in [n]} \in \Realfield_{\ge 0}^n$
whose subvector $(\vartheta_s)_{s \in \AltSet}$
is uniquely determined as follows:
$\vartheta_s = \zeta_s^{-1}(z^*)$,
where $z^*$ is the unique real in $[0,1]$ that satisfies
\[
\sum_{s \in \AltSet} \pi_s \cdot \zeta_s^{-1}(z^*) = \omega
\]
(taking $\zeta_1^{-1}(z^*) \equiv 0$
unless $\AltSet = \{ 1 \}$, in which case
$\vartheta_1 = \omega$).
\end{lemma}

\begin{proof}
Since $\omega \mapsto \Rate_0(\omega,1)$ is identically zero,
the lemma holds when $\AltSet = \{ 1 \}$,
so we assume hereafter that $\AltSet \ne \{ 1 \}$
and denote $\AltSet^* = \AltSet \setminus \{ 1 \}$.
By condition~(P3), we can further assume that $\vartheta_1 = 0$,
since otherwise we can reduce $\vartheta_1$ and increase
$\vartheta_s$ for $s \in [n]^*$, thereby only decreasing
the inner expression in~(\ref{eq:WZL1}).

Define the function
$\bldvartheta = (\vartheta_s)_{s \in [n]^*} \mapsto f(\bldvartheta)$
for every $\bldvartheta \in \Realfield_{\ge 0}^{n-1}$
to be the minimand in~(\ref{eq:WZL1}):
\[
f(\bldvartheta) = f_\bldpi(\bldvartheta)
= \sum_{s \in \AltSet^*} \pi_s \cdot \Rate_0(\vartheta_s,s) .
\]
Since $\vartheta \mapsto \Rate_0(\vartheta,s)$ is convex, so is 
$\bldvartheta \mapsto f(\bldvartheta)$.

We minimize $f$ subject to condition~(P3)
using the method of Lagrange multipliers~\cite[\S 10.3]{Luenberger}:
we introduce a variable $\xi$ and require that
the partial derivatives of the Lagrangian
\[
L(\bldvartheta,\xi) = f(\bldvartheta) + \xi
\cdot \Bigl( \omega - \sum_{s \in [n]^*} \pi_s \cdot \vartheta_s \Bigr)
\]
be zero with respect to $\xi$ and
the entries of $(\vartheta_s)_{s \in \AltSet^*}$, namely,
\begin{equation}
\label{eq:lagrange1}
\sum_{s \in \AltSet^*} \pi_s \cdot \vartheta_s = \omega
\end{equation}
and
\begin{equation}
\label{eq:lagrange2}
\frac{\partial}{\partial \vartheta_s} L(\bldvartheta,\xi)
= 0 , \quad s \in \AltSet^* .
\end{equation}
Denoting
$\bldzeta(\bldvartheta)
=\left(\zeta_s(\vartheta_s)\right)_{s \in \AltSet^*}$,
we can write $f(\bldvartheta)$ as:
\[
f(\bldvartheta) = f(\bldvartheta,\bldz)
\bigm|_{\bldz = \bldzeta(\bldvartheta)} ,
\]
where $\bldz = (z_s)_{s \in \AltSet^*}$
is a real vector of variables and
\[
f(\bldvartheta,\bldz)
= 1 + \sum_{s \in \AltSet^*} \pi_s 
\Bigl(
\vartheta_s \cdot \log_q z_s
- \frac{1}{s} \log_q Q_s(z_s) \Bigr) .
\]
Recalling that
\[
\frac{\partial}{\partial z_i} f(\bldvartheta,\bldz)
\Bigm|_{\bldz = \bldzeta(\bldvartheta)} = 0 ,
\quad
\textrm{for every $i \in \AltSet^*$} ,
\]
we get for every $s \in \AltSet^*$:
\begin{eqnarray*}
\frac{\partial}{\partial \vartheta_s} f(\bldvartheta)
& = &
\frac{\partial}{\partial \vartheta_s} f(\bldvartheta,\bldz)
\Bigm|_{\bldz = \bldzeta(\bldvartheta)} \\
&&
\quad {}
+ 
\sum_{i \in \AltSet^*}
{\underbrace{
\frac{\partial}{\partial z_i} f(\bldvartheta,\bldz)
\Bigm|_{\bldz = \bldzeta(\bldvartheta)} 
}_0}
\cdot
\frac{\partial}{\partial \vartheta_s}
\zeta_i(\vartheta_i) \\
& = &
\pi_s \cdot \log_q \zeta_s(\vartheta_s) .
\end{eqnarray*}
Hence, by~(\ref{eq:lagrange2}),
\[
\pi_s \cdot \left( \log_q \zeta_s(\vartheta_s) - \xi \right) = 0 ,
\]
namely, the values $\zeta_s(\vartheta_s)$
are equal to (the same value) $z^* = q^\xi$, for all $s \in \AltSet^*$.
Finally, by~(\ref{eq:lagrange1}), the value $z^*$ must be such that
\[
\sum_{s \in \AltSet^*} \pi_s \cdot \zeta_s^{-1}(z^*) = \omega .
\]
This equality determines $z^*$ uniquely,
since $z \mapsto \zeta_s^{-1}(z)$ is strictly increasing
for any $s \in [n]^*$.
\end{proof}

We next turn to the characterization of the outer maximum
in~(\ref{eq:WZL1}).

\begin{lemma}
\label{lem:outer}
Given $n \in \Integers^+$, $\mu \in [n]$,
and $\omega \in (0,(q{-}1)/q)$,
let $k = \lfloor \mu \rfloor$ and
\begin{equation}
\label{eq:pik}
\pi =
\frac{k(k{+}1)}{\mu} - k .
\end{equation}
The entries of the maximizing vector
$\bldpi = (\pi_s)_{s \in [n]} \in \Realfield_{\ge 0}^n$
in~(\ref{eq:WZL1}) under the constraints (P1)--(P2)
are all zero, except for
the entries that are indexed by $k$
and (possibly) $k+1$, where
\[
\pi_k = \pi
\quad \textrm{and} \quad
\pi_{k+1} = 1 - \pi .
\]
\end{lemma}

\begin{proof}
When $\mu = 1$, conditions (P1)--(P2) force $\bldpi$
to be $(1 \, 0 \, 0 \, \ldots \, 0)$
(i.e., $\Support(\bldpi) = \{ 1 \}$),
in which case $\omega \mapsto \overline{\Rate}_0(\omega,n,1)$
is identically zero.
Hence, we assume hereafter that $\mu > 1$, in which case
a maximizing $\bldpi$ must have support${} \ne \{ 1 \}$
to achieve $\overline{\Rate}_0(\omega,n,\mu) > 0$.

Define the function $\bldpi = (\pi_s)_{s \in [n]} \mapsto g(\bldpi)$
for every $\bldpi \in \Realfield^n_{\ge 0}$ to be the maximand
in~(\ref{eq:WZL1}), namely,
\begin{eqnarray*}
g(\bldpi)
& = &
\sum_{s \in [n]} \pi_s \cdot \Rate_0(\vartheta^*_s,s) \\
& = &
1 - \omega \cdot \log_q z^*
+ \sum_{s \in [n]}
\frac{\pi_s}{s} \log_q Q_s(z^*) ,
\end{eqnarray*}
where $\vartheta^*_s = \zeta_s^{-1}(z^*)$ for each $s \in [n]$
and $z^*$ is as in Lemma~\ref{lem:inner};
namely, $z^*$ is determined uniquely by $\bldpi$ (and $\omega$)
and, therefore, so is each $\vartheta^*_s$
(in particular, $\vartheta^*_1 = 0$).
We do the maximization subject to the constraints (P1)--(P2)
using the Kuhn--Tucker conditions~\cite[\S 10.8]{Luenberger}:
we introduce a real variable $\xi$
and require that the partial derivatives of
\begin{eqnarray}
K(\bldpi,\xi)
\!\! & = & \!\!
g(\bldpi) 
+ \xi \cdot \Bigl( 1 - \sum_{s \in [n]} \pi_s \Bigr) \nonumber \\
\label{eq:KKT0}
&& {}
- \beta
\cdot \Bigl( \frac{1}{\mu} - \sum_{s \in [n]} \frac{\pi_s}{s} \Bigr)
+ \sum_{s \in [n]} \eta_s \cdot \pi_s
\end{eqnarray}
be zero with respect to $\xi$
and the entries of $\bldpi$, for some nonnegative
$\beta$ and $\bldeta = (\eta_s)_{s \in [n]}$ that satisfy
\begin{equation}
\label{eq:KKT1}
\beta
\cdot \Bigl( \frac{1}{\mu} - \sum_{s \in [n]} \frac{\pi_s}{s} \Bigr)
= 0
\end{equation}
and
\begin{equation}
\label{eq:KKT2}
\eta_s \cdot \pi_s = 0 , \quad \textrm{for each $s \in [n]$} .
\end{equation}
The second term in the right-hand side of~(\ref{eq:KKT0})
corresponds to condition~(P1);
the third term and~(\ref{eq:KKT1}) correspond to condition~(P2);
and the last term and~(\ref{eq:KKT2})
correspond to requiring that $\bldpi$ be nonnegative.

Similarly to what we have done in the proof
of Lemma~\ref{lem:inner}, we can write $g(\bldpi)$ as
\[
g(\bldpi) = g(\bldpi,\bldz)
\bigm|_{\bldz = z^* \cdot \bldone} ,
\]
where $\bldz = (z_s)_{s \in[n]^*}$
is a real vector of variables
and $\bldone$ stands for the all-one vector
in $\Realfield_{\ge 0}^{n-1}$.
Recalling that
\[
\frac{\partial}{\partial z_i} g(\bldpi,\bldz)
\Bigm|_{z_i = z^*} = 0 ,
\quad \textrm{for every $i \in [n]^*$} ,
\]
we get:
\begin{eqnarray*}
\frac{\partial}{\partial \pi_s} g(\bldpi)
& = &
\frac{\partial}{\partial \pi_s} g(\bldpi,\bldz)
\Bigm|_{\bldz = z^* \cdot \bldone} \\
&&
\;\;\;
{}
+
\sum_{i \in [n]^*}
{\underbrace{
\frac{\partial}{\partial z_i} g(\bldpi,\bldz) \Bigm|_{z_i = z^*}
}_0}
\cdot
\frac{\partial}{\partial \pi_s}
\zeta_i(\vartheta^*_i(\bldpi)) \\
& = &
\frac{1}{s} \log_q Q_s(z^*) .
\end{eqnarray*}
Returning to the expression for $K(\bldpi,\xi)$ in~(\ref{eq:KKT0}),
we conclude that $(\partial K(\bldpi,\xi))/\partial \pi_s = 0$
translates into
\begin{equation}
\label{eq:KKT3}
\eta_s = \psi(s) ,
\end{equation}
where
\[
\psi(s) = \psi(s,\beta)
= \frac{1}{s} \Bigl(
\log_q Q_s(z^*) - \beta \Bigr) + \xi .
\]

Next, we analyze the function $s \mapsto \psi(s)$
in the \emph{real} variable $s \in [1,\infty)$.
Taking the derivative of this function, we get:
\begin{eqnarray}
\psi'(s)
& = &
\frac{1}{s} \cdot \frac{1}{Q_s(z^*)}
\biggl(
(1+(q{-}1)z^*)^s \log_q (1+(q{-}1)z^*) \nonumber \\
&& \quad\quad\quad {}
+ (q{-}1)(1{-}z^*)^s \log_q (1{-}z^*) \biggr) \nonumber \\
&& \quad {}
- \frac{1}{s^2}
\left( \log_q Q_s(z^*) - \beta \right) \nonumber \\
\label{eq:psider}
& = &
\frac{1}{s^2} \bigl( \beta - \entropy_q(p_s) \bigr) ,
\end{eqnarray}
where
\[
p_s = \frac{(q{-}1) (1{-}z^*)^s}{Q_s(z^*)}
\]
and $x \mapsto \entropy_q(x)$ is the $q$-ary entropy function
$x \cdot \log_q (q{-}1) - x \log_q x - (1{-}x) \log_q (1{-}x)$.
This function is strictly increasing on $[0,(q{-}1)/q]$
and $s \mapsto p_s$ is strictly decreasing on $[1,\infty)$,
with the following maximum value attained at $s = 1$:
\[
p_1 = \frac{(q{-}1)(1{-}z^*)}{Q_1(z^*)} < \frac{q{-}1}{q} .
\]
We conclude from~(\ref{eq:psider}) that, on the interval $[1,n]$,
the function $s \mapsto \psi(s)$ is either
(i)~strictly increasing,
or
(ii)~strictly decreasing,
or
(iii)~unimodal with one local minimum.
Note from~(\ref{eq:psider}) that $\beta = 0$ implies case~(ii).

As our next step, we turn back to~(\ref{eq:KKT3}).
Since $\eta_s \ge 0$ for all $s \in [n]$, in case~(i)
we can have $\eta_s = 0$ only when $s = 1$. Similarly, in case~(ii)
we can have $\eta_s = 0$ only when $s = n$. And in case~(iii)
we can have $\eta_s = 0$ only for the (one or)
two consecutive indexes $s$ that are adjacent to
the minimum of $\psi(\cdot)$.
We conclude that $\eta_s > 0$ for all $s$ except for those two indexes;
this, together with~(\ref{eq:KKT2}), implies that the support
of a maximizing $\bldpi$ contains (only) those two indexes.

Observe that by~(\ref{eq:KKT1}),
a maximizing $\bldpi$ must satisfy condition~(P2) with equality,
unless $\beta = 0$, in which case
$\bldpi = (0 \, 0 \, 0 \, \ldots \, 1)$ (case~(ii) above).
Condition~(P2) would then read
$1/n = \pi_n/n \ge 1/\mu$, i.e., $\mu \ge n$.
Yet we assume that $\mu \in [n]$, thereby forcing $\mu = n$,
so condition~(P2) holds with equality in this case too.

Finally, let $k$ and (possibly) $k+1$ be the elements
in the support of a maximizing $\bldpi$.
We then get from conditions (P1) (and the equality version of) (P2)
that the entries $\pi_k$ and $\pi_{k+1}$
satisfy the following equations:
\[
\begin{array}{ccc}
\displaystyle
\pi_k + \pi_{k+1} & = & 1 \phantom{.}\rule[-2ex]{0ex}{1ex} \\
\displaystyle
\frac{\pi_k}{k} + \frac{\pi_{k+1}}{k{+}1}
& = &
\displaystyle
\frac{1}{\mu} .
\end{array}
\]
Solving these equations yields $\pi_k = \pi$ and
$\pi_{k+1} = 1-\pi$, where $\pi$ is given by~(\ref{eq:pik}),
and the value of $k$ is determined by
\[
\renewcommand{\arraystretch}{1.1}
\begin{array}{rcl}
\pi \ge 0 & \Rightarrow & \mu \le k+1 \\
\pi \le 1 & \Rightarrow & \mu \ge k .
\end{array}
\]
\end{proof}

Combining Lemmas~\ref{lem:inner} and~\ref{lem:outer}
leads to the following recipe for computing
the value of $\overline{\Rate}_0(\omega,n,\mu)$.

\begin{proposition}
\label{prop:WZL}
Given $n \in \Integers^+$, $\mu \in [n]$,
and $\omega \in [0,(q{-}1)/q]$,
let $k = \lfloor \mu \rfloor$
and let $\pi$ be as in (\ref{eq:pik}).
Also, let $z^*$ be the unique real in $[0,1]$ that satisfies
\begin{equation}
\label{eq:zstar}
\pi \cdot \zeta_k^{-1}(z^*)
+ (1-\pi) \cdot \zeta_{k+1}^{-1}(z^*) = \omega
\end{equation}
(taking $\zeta_1^{-1}(z^*) \equiv 0$ when $\mu \in (1,2)$
and $\zeta_1^{-1}(z^*) = \omega$ when $\mu = 1$).
Then
\[
\overline{\Rate}_0(\omega,n,\mu) =
\pi \cdot \Rate_0(\vartheta, k)
+ (1-\pi) \cdot \Rate_0(\vartheta',k{+}1) ,
\]
where
\[
\vartheta = \zeta_k^{-1}(z^*)
\quad \textrm{and} \quad
\vartheta' = \zeta_{k+1}^{-1}(z^*) .
\]
In particular, when $\mu$ is an integer then
\[
\overline{\Rate}_0(\omega,n,\mu) = \Rate_0(\omega,\mu) .
\]
\end{proposition}

\begin{remark}
\label{rem:WZL}
By Remark~\ref{rem:unique},
solving~(\ref{eq:zstar}) for $z^*$ amounts
to finding the (unique) real root in $[0,1]$ of the polynomial
\begin{eqnarray*}
\lefteqn{
\omega \cdot Q_k(z) \cdot Q_{k+1}(z) 
} \makebox[0ex]{} \\
&& 
\!\! {}
- z \cdot \bigl( \pi \cdot P_k(z) \cdot Q_{k+1}(z)
+ (1-\pi) \cdot P_{k+1}(z) \cdot Q_k(z) \bigr) .
\end{eqnarray*}\qed
\end{remark}

\begin{remark}
\label{rem:WZL2}
Using \Proposition~\ref{prop:WZL}, it is fairly easy
to compute the maximization in~(\ref{eq:WZL2}) numerically.
Given $n \in \Integers^+$,
we next argue that for sufficiently small positive $\omega$
(within an interval whose length depends on $n$),
the maximum in~(\ref{eq:WZL2}) is attained at $\nu = 0$.
For any $\nu \in [0,(n{-}1)/(n{+}1)]$,
let $\mu(\nu)$ denote the value $((1{-}\nu)/(1{+}\nu)) n \; (\in [n])$.
Given $\omega \in [0,2(q{-}1)/(q (n{+}1))]$,
the objective function in~(\ref{eq:WZL2})
can be verified to be (continuous and)
piecewise differentiable in $\nu$.
Specifically, at any $\nu \in [0,(n{-}1)/(n{+}1)]$
for which $\mu(\nu)$ is not an integer, the derivative is given by
\begin{equation}
\label{eq:derivative}
\left( 1 + \frac{k}{n} \right) \log_q Q_{k+1}(z^*)
- \left( 1 + \frac{k{+}1}{n} \right) \log_q Q_k(z^*) ,
\end{equation}
where $k$ and $z^*$ are as in \Proposition~\ref{prop:WZL},
with $\mu$ and $\omega$ therein taken 
as $\mu = \mu(\nu)$
and $\omega/(1{-}\nu) \; (\in [0,(q{-}1)/q])$, respectively.
Now, by~(\ref{eq:zetainverse})
and~(\ref{eq:zstar}) it follows that $z^* \rightarrow 0$
as $\omega \rightarrow 0$ (uniformly in $\nu$).
And since the expression~(\ref{eq:derivative}) is negative
at $z^* = 0$ for any $k \in [n{-}1]$, we conclude that
for sufficiently small $\omega > 0$, the objective function
in~(\ref{eq:WZL2}) is decreasing in $\nu$.
Thus, the maximum therein is attained at $\nu = 0$, in which case
\[
\widehat{\Rate}_0(\omega,n) = \Rate_0(\omega,n) .
\]\qed
\end{remark}

\begin{remark}
\label{rem:averagelength2}
As we pointed out in Remark~\ref{rem:averagelength1},
the value $\overline{\Rate}_0(\delta/2,n,\mu)$
bounds from above the rate of any all-disjoint
linearly recoverable $(\delta,n)$-LRC sequence
$(\varcode_i)_{i=1}^\infty$ under the additional
condition that the average size of the repair groups
of each $\varcode_i$ is at most $\mu$
(in the limit when $i \rightarrow \infty$).
A related relevant problem is obtaining such a bound
when the average is computed per coordinate, i.e.,
the average is taken over the whole list of the $N_i$ repair groups
of $\varcode_i$; thus, each of the distinct repair groups
is counted a number of times which equals
the number of coordinates that it covers
(namely, its size). 
A counterpart of the bound~(\ref{eq:Singleton})
for this setting was presented in~\cite{SKA}.

A sphere-packing upper bound for this setting
is obtained by substituting $\omega = \delta/2$
in~(\ref{eq:WZL1}), except that condition~(P2) is replaced by
\begin{list}{}{\settowidth{\labelwidth}{\textrm{P2)}}}
\item[P2)]
$\displaystyle \sum_{s \in [n]} s \cdot \pi_s \le \mu$.
\end{list}
\Proposition~\ref{prop:WZL} still holds, except that
the value of $\pi$ in~(\ref{eq:pik}) is changed into
\[
\pi = k + 1 - \mu .
\]
Specifically, Lemma~\ref{lem:inner} holds as is;
as for Lemma~\ref{lem:outer},
the third term in~(\ref{eq:KKT0})
(and, accordingly, the left-hand side of~(\ref{eq:KKT1})) is replaced by
\[
\beta \cdot \Bigl( \mu - \sum_{s \in [n]} s \cdot \pi_s \Bigr)
\]
(with a plus sign). Consequently, (\ref{eq:psider}) becomes
\[
\beta - \frac{1}{s^2} \cdot \entropy_q(p_s)
\]
which, in turn, leads to the same conclusions about the function
$s \mapsto \psi(s)$ (cases (i)--(iii)).
Thus, the support of a maximizing $\bldpi$ contains up to
two indexes, $k$ and $k+1$, and, by (P1)--(P2), the values of
$\pi_k$ and $\pi_{k+1}$ are the solutions of
\[
\begin{array}{c@{\;}c@{\;}ccc}
\displaystyle
\pi_k & + & \pi_{k+1} & = & 1 \phantom{.}\rule[-2ex]{0ex}{1ex} \\
\displaystyle
k \cdot \pi_k & + & (k{+}1) \cdot \pi_{k+1}
& = &
\mu .
\end{array}
\]\qed
\end{remark}
\fi

%
%   Gray  -- SP
%   Green -- CM
%   Blue  -- Shortening of LP
%   Red   -- New bound
%   Cyan  -- Shortening of new bound
%
\newcommand{\Bfull}{\circle*{0.80}}
\newcommand{\Bbox}{{\put(-.05,-.05){\framebox(0.1,0.1){$\cdot$}}}}
\newcommand{\Bdiam}{{\multiput(-.07,-.33)(-.13,.13){3}{\line(1,1){.4}}}}
\newcommand{\Bplus}{{\thinlines%
            \put(-.38,0){\line(1,0){.76}}\put(0,-.38){\line(0,1){.76}}}}
\newcommand{\Bprod}{{\thinlines%
            \put(-.4,-.4){\line(1,1){.8}}\put(-.4,.4){\line(1,-1){.8}}}}
\newcommand{\Bempty}{{\color{gray}\circle{0.6}}}
\newcommand{\Plot}[1]{%
    \begin{center}
    \small
    \thicklines
    \ifIEEE
        \setlength{\unitlength}{1.3mm}
    \else
        \setlength{\unitlength}{1.2mm}
    \fi
    \begin{picture}(135,135)(-10,-10)
    \put(-10,000){\vector(1,0){120}}
    \put(114,000){\makebox(0,0)[l]{$\delta$}}
    \put(000,-10){\vector(0,1){120}}
    \put(-04,114){\makebox(0,0)[r]{$R$}}
    \put(-05,-04){\makebox(0,0)[t]{$0$}}
    \multiput(020,000)(020,000){5}{\line(0,-1){2}}
    \put(020,-04){\makebox(0,0)[t]{$0.1$}}
    \put(040,-04){\makebox(0,0)[t]{$0.2$}}
    \put(060,-04){\makebox(0,0)[t]{$0.3$}}
    \put(080,-04){\makebox(0,0)[t]{$0.4$}}
    \put(100,-04){\makebox(0,0)[t]{$0.5$}}
    \multiput(000,010)(000,010){10}{\line(-1,0){2}}
    \put(-04,050){\makebox(0,0)[r]{$0.5$}}
    \put(-04,100){\makebox(0,0)[r]{$1.0$}}
    \put(000,000){#1}
    \end{picture}
    \thinlines
    \setlength{\unitlength}{1pt}
    \end{center}
}
\newcommand{\Legend}{%
        %%\put(000,000){\color{brown}\line(1,0){10}}
        \put(000,000){\multiput(.2,0)(1.92,000){6}{\color{brown}\Bbox}}
        \put(013,000){\makebox(0,0){(a)}}
        \put(016,000){\makebox(0,0)[l]{%
                    $\delta \mapsto \Rate_\SP(\delta,n)$}}
        \put(000,-04){\color{green}\line(1,0){10}}
        \put(000,-04){\multiput(.28,0)(1.89,00){6}{\color{green}\Bdiam}}
        \put(013,-04){\makebox(0,0){(b)}}
        \put(016,-04){\makebox(0,0)[l]{%
                    $\delta \mapsto \Rate_\CM(\delta,n)$
                    with $\Rate_\opt(\delta) = \Rate_\LP(\delta)$}}
        \put(000,-08){\color{blue}\line(1,0){10}}
        \put(000,-08){\multiput(.4,0)(1.84,000){6}{\color{blue}\Bprod}}
        \put(013,-08){\makebox(0,0){(c)}}
        \put(016,-08){\makebox(0,0)[l]{%
                    $\delta \mapsto \Rate_1(\delta,n)$
                    with $\Rate_\LRC(\delta,n) = \Rate_\LP(\delta)$}}
        \put(000,-12){\color{red}\line(1,0){10}}
        \put(013,-12){\makebox(0,0){(d)}}
        \put(016,-12){\makebox(0,0)[l]{%
                    $\delta \mapsto \Rate_2(\delta,n)$}}
        \put(000,-16){\color{cyan}\line(1,0){10}}
        \put(013,-16){\makebox(0,0){(e)}}
        \put(016,-16){\makebox(0,0)[l]{%
                    $\delta \mapsto \Rate_1(\delta,n)$
                    with $\Rate_\LRC(\delta,n) = \Rate_2(\delta,n)$}}
        \put(000,-16){\multiput(.2,0)(1.92,000){6}{\color{cyan}\Bfull}}
        \put(013,-20){\makebox(0,0){(f)}}
        \put(016,-20){\makebox(0,0)[l]{%
                    $\delta \mapsto \Rate_0(\delta,n)$ (lower bound)}}
        \put(000,-20){\multiput(.2,0)(1.92,000){6}{\color{black}\Bfull}}
}
\newcommand{\LegendAlt}{%
        %%\put(000,000){\color{brown}\line(1,0){10}}
        \put(000,000){\multiput(.2,0)(1.92,000){6}{\color{brown}\Bbox}}
        \put(013,000){\makebox(0,0){(a)}}
        \put(016,000){\makebox(0,0)[l]{%
                    $\delta \mapsto \widehat{\Rate}_0(\delta/2,n)$
                    (sphere-packing bound)}}
        \put(000,-04){\color{green}\line(1,0){10}}
        \put(000,-04){\multiput(.3,0)(1.88,000){6}{\color{green}\Bdiam}}
        \put(013,-04){\makebox(0,0){(b)}}
        \put(016,-04){\makebox(0,0)[l]{%
                    $\delta \mapsto \Rate_\CM(\delta,n)$
                    with $\Rate_\opt(\delta) = \Rate_\LP(\delta)$}}
        \put(000,-08){\color{blue}\line(1,0){10}}
        \put(000,-08){\multiput(.4,0)(1.84,000){6}{\color{blue}\Bprod}}
        \put(013,-08){\makebox(0,0){(c)}}
        \put(016,-08){\makebox(0,0)[l]{%
                    $\delta \mapsto \widehat{\Rate}_1(\delta,n)$
                    with $\Rate_\opt(\delta,n) = \Rate_\LP(\delta)$}}
        \put(000,-12){\color{red}\line(1,0){10}}
        \put(013,-12){\makebox(0,0){(d)}}
        \put(016,-12){\makebox(0,0)[l]{%
                    $\delta \mapsto \min \bigl\{
                    \widehat{\Rate}_2(\delta,n),
                    \Rate_\LP(\delta,n) \bigr\}$}}
        \put(013,-16){\makebox(0,0){(e)}}
        \put(016,-16){\makebox(0,0)[l]{%
                    $\delta \mapsto \Rate_0(\delta,n)$ (lower bound)}}
        \put(000,-16){\multiput(.2,0)(1.92,000){6}{\color{black}\Bfull}}
}
\ifPAGELIMIT
\else
\begin{figure*}[hbt]
%% q = 2, n =  3:
\Plot{%
    \put(035,095){
	\put(000,000){\Legend}
        \put(000,-24){\multiput(0.3,000)(1.88,000){6}{\Bempty}}
        \put(013,-24){\makebox(0,0){(g)}}
        \put(016,-24){\makebox(0,0)[l]{Eq.~(\ref{eq:n=3})}}
    }
%%  Aaltonen:
    \put(0.000,66.667){\Bempty}\put(2.000,64.129){\Bempty}
    \put(4.000,62.065){\Bempty}\put(6.000,60.169){\Bempty}
    \put(8.000,58.375){\Bempty}\put(10.000,56.654){\Bempty}
    \put(12.000,54.988){\Bempty}\put(14.000,53.366){\Bempty}
    \put(16.000,51.779){\Bempty}\put(18.000,50.221){\Bempty}
    \put(20.000,48.688){\Bempty}\put(22.000,47.175){\Bempty}
    \put(24.000,45.680){\Bempty}\put(26.000,44.200){\Bempty}
    \put(28.000,42.734){\Bempty}\put(30.000,41.278){\Bempty}
    \put(32.000,39.832){\Bempty}\put(34.000,38.394){\Bempty}
    \put(36.000,36.962){\Bempty}\put(38.000,35.537){\Bempty}
    \put(40.000,34.116){\Bempty}\put(42.000,32.700){\Bempty}
    \put(44.000,31.286){\Bempty}\put(46.000,29.874){\Bempty}
    \put(48.000,28.465){\Bempty}\put(50.000,27.056){\Bempty}
    \put(52.000,25.647){\Bempty}\put(54.000,24.236){\Bempty}
    \put(56.000,22.825){\Bempty}\put(58.000,21.408){\Bempty}
    \put(60.000,19.981){\Bempty}\put(62.000,18.568){\Bempty}
    \put(64.000,17.176){\Bempty}\put(66.000,15.809){\Bempty}
    \put(68.000,14.467){\Bempty}\put(70.000,13.154){\Bempty}
    \put(72.000,11.871){\Bempty}\put(74.000,10.622){\Bempty}
    \put(76.000,9.410){\Bempty}\put(78.000,8.238){\Bempty}
    \put(80.000,7.111){\Bempty}\put(82.000,6.033){\Bempty}
    \put(84.000,5.010){\Bempty}\put(86.000,4.048){\Bempty}
    \put(88.000,3.155){\Bempty}\put(90.000,2.340){\Bempty}
    \put(92.000,1.614){\Bempty}\put(94.000,0.992){\Bempty}
    \put(96.000,0.493){\Bempty}\put(98.000,0.145){\Bempty}
    \put(100.000,0.000){\Bempty}
 % n =  3:
    \put(000,000){\color{brown}
        \put(0.000,66.667){\Bbox}\put(2.000,64.146){\Bbox}
        \put(4.000,62.129){\Bbox}\put(6.000,60.303){\Bbox}
        \put(8.000,58.602){\Bbox}\put(10.000,56.995){\Bbox}
        \put(12.000,55.464){\Bbox}\put(14.000,53.996){\Bbox}
        \put(16.000,52.582){\Bbox}\put(18.000,51.216){\Bbox}
        \put(20.000,49.894){\Bbox}\put(22.000,48.610){\Bbox}
        \put(24.000,47.363){\Bbox}\put(26.000,46.148){\Bbox}
        \put(28.000,44.964){\Bbox}\put(30.000,43.809){\Bbox}
        \put(32.000,42.681){\Bbox}\put(34.000,41.579){\Bbox}
        \put(36.000,40.501){\Bbox}\put(38.000,39.446){\Bbox}
        \put(40.000,38.414){\Bbox}\put(42.000,37.402){\Bbox}
        \put(44.000,36.411){\Bbox}\put(46.000,35.440){\Bbox}
        \put(48.000,34.488){\Bbox}\put(50.000,33.554){\Bbox}
        \put(52.000,32.637){\Bbox}\put(54.000,31.738){\Bbox}
        \put(56.000,30.856){\Bbox}\put(58.000,29.990){\Bbox}
        \put(60.000,29.140){\Bbox}\put(62.000,28.305){\Bbox}
        \put(64.000,27.486){\Bbox}\put(66.000,26.681){\Bbox}
        \put(68.000,25.891){\Bbox}\put(70.000,25.115){\Bbox}
        \put(72.000,24.353){\Bbox}\put(74.000,23.605){\Bbox}
        \put(76.000,22.870){\Bbox}\put(78.000,22.149){\Bbox}
        \put(80.000,21.441){\Bbox}\put(82.000,20.745){\Bbox}
        \put(84.000,20.063){\Bbox}\put(86.000,19.392){\Bbox}
        \put(88.000,18.735){\Bbox}\put(90.000,18.089){\Bbox}
        \put(92.000,17.455){\Bbox}\put(94.000,16.834){\Bbox}
        \put(96.000,16.224){\Bbox}\put(98.000,15.626){\Bbox}
        \put(100.000,15.040){\Bbox}
    \color{black}
        \put(0.000,66.667){\Bfull}\put(2.000,62.129){\Bfull}
        \put(4.000,58.602){\Bfull}\put(6.000,55.464){\Bfull}
        \put(8.000,52.582){\Bfull}\put(10.000,49.894){\Bfull}
        \put(12.000,47.363){\Bfull}\put(14.000,44.964){\Bfull}
        \put(16.000,42.681){\Bfull}\put(18.000,40.501){\Bfull}
        \put(20.000,38.414){\Bfull}\put(22.000,36.411){\Bfull}
        \put(24.000,34.488){\Bfull}\put(26.000,32.637){\Bfull}
        \put(28.000,30.856){\Bfull}\put(30.000,29.140){\Bfull}
        \put(32.000,27.486){\Bfull}\put(34.000,25.891){\Bfull}
        \put(36.000,24.353){\Bfull}\put(38.000,22.870){\Bfull}
        \put(40.000,21.441){\Bfull}\put(42.000,20.063){\Bfull}
        \put(44.000,18.735){\Bfull}\put(46.000,17.455){\Bfull}
        \put(48.000,16.224){\Bfull}\put(50.000,15.040){\Bfull}
        \put(52.000,13.902){\Bfull}\put(54.000,12.810){\Bfull}
        \put(56.000,11.762){\Bfull}\put(58.000,10.759){\Bfull}
        \put(60.000,9.800){\Bfull}\put(62.000,8.884){\Bfull}
        \put(64.000,8.012){\Bfull}\put(66.000,7.184){\Bfull}
        \put(68.000,6.399){\Bfull}\put(70.000,5.657){\Bfull}
        \put(72.000,4.958){\Bfull}\put(74.000,4.303){\Bfull}
        \put(76.000,3.692){\Bfull}\put(78.000,3.125){\Bfull}
        \put(80.000,2.602){\Bfull}\put(82.000,2.125){\Bfull}
        \put(84.000,1.693){\Bfull}\put(86.000,1.308){\Bfull}
        \put(88.000,0.970){\Bfull}\put(90.000,0.680){\Bfull}
        \put(92.000,0.440){\Bfull}\put(94.000,0.250){\Bfull}
        \put(96.000,0.113){\Bfull}\put(98.000,0.028){\Bfull}
        \put(100.000,0.000){\Bfull}
    \color{green}
        \qbezier(0.000,66.667)(1.000,65.934)(2.000,65.202)
        \qbezier(2.000,65.202)(2.966,64.494)(4.000,63.736)
        \qbezier(4.000,63.736)(4.316,63.505)(6.000,62.271)
        \qbezier(6.000,62.271)(7.921,60.864)(8.000,60.806)
        \qbezier(8.000,60.806)(9.000,60.074)(10.000,59.341)
        \qbezier(10.000,59.341)(11.146,58.502)(12.000,57.876)
        \qbezier(12.000,57.876)(13.000,57.143)(14.000,56.411)
        \qbezier(14.000,56.411)(14.809,55.819)(16.000,54.946)
        \qbezier(16.000,54.946)(17.001,54.212)(18.000,53.481)
        \qbezier(18.000,53.481)(19.000,52.748)(20.000,52.016)
        \qbezier(20.000,52.016)(20.423,51.705)(22.000,50.551)
        \qbezier(22.000,50.551)(22.004,50.547)(24.000,49.085)
        \qbezier(24.000,49.085)(25.000,48.353)(26.000,47.620)
        \qbezier(26.000,47.620)(27.917,46.216)(28.000,46.155)
        \qbezier(28.000,46.155)(29.000,45.423)(30.000,44.690)
        \qbezier(30.000,44.690)(30.883,44.043)(32.000,43.225)
        \qbezier(32.000,43.225)(32.721,42.697)(34.000,41.760)
        \qbezier(34.000,41.760)(35.000,41.027)(36.000,40.295)
        \qbezier(36.000,40.295)(37.000,39.562)(38.000,38.830)
        \qbezier(38.000,38.830)(39.000,38.097)(40.000,37.365)
        \qbezier(40.000,37.365)(41.000,36.632)(42.000,35.899)
        \qbezier(42.000,35.899)(43.179,35.036)(44.000,34.434)
        \qbezier(44.000,34.434)(44.474,34.087)(46.000,32.969)
        \qbezier(46.000,32.969)(47.093,32.169)(48.000,31.504)
        \qbezier(48.000,31.504)(48.568,31.088)(50.000,30.039)
        \qbezier(50.000,30.039)(50.088,29.974)(52.000,28.574)
        \qbezier(52.000,28.574)(53.965,27.134)(54.000,27.109)
        \qbezier(54.000,27.109)(55.000,26.376)(56.000,25.644)
        \qbezier(56.000,25.644)(57.208,24.759)(58.000,24.179)
        \qbezier(58.000,24.179)(58.582,23.752)(60.000,22.713)
        \qbezier(60.000,22.713)(61.000,21.981)(62.000,21.248)
        \qbezier(62.000,21.248)(63.000,20.516)(64.000,19.783)
        \qbezier(64.000,19.783)(64.399,19.491)(66.000,18.318)
        \qbezier(66.000,18.318)(67.798,17.001)(68.000,16.853)
        \qbezier(68.000,16.853)(68.054,16.814)(70.000,15.388)
        \qbezier(70.000,15.388)(71.000,14.655)(72.000,13.923)
        \qbezier(72.000,13.923)(73.000,13.190)(74.000,12.458)
        \qbezier(74.000,12.458)(75.048,11.690)(76.000,10.993)
        \qbezier(76.000,10.993)(77.510,9.887)(78.000,9.537)
        \qbezier(78.000,9.537)(79.013,8.813)(80.000,8.147)
        \qbezier(80.000,8.147)(81.022,7.457)(82.000,6.838)
        \qbezier(82.000,6.838)(83.014,6.195)(84.000,5.615)
        \qbezier(84.000,5.615)(85.006,5.022)(86.000,4.484)
        \qbezier(86.000,4.484)(87.007,3.938)(88.000,3.451)
        \qbezier(88.000,3.451)(89.027,2.948)(90.000,2.527)
        \qbezier(90.000,2.527)(91.028,2.081)(92.000,1.719)
        \qbezier(92.000,1.719)(93.013,1.342)(94.000,1.041)
        \qbezier(94.000,1.041)(95.018,0.732)(96.000,0.509)
        \qbezier(96.000,0.509)(97.029,0.277)(98.000,0.147)
        \qbezier(98.000,0.147)(99.094,0.002)(100.000,0.000)
        \put(0.000,66.667){\Bbox}\put(2.000,65.202){\Bbox}
        \put(4.000,63.736){\Bbox}\put(6.000,62.271){\Bbox}
        \put(8.000,60.806){\Bbox}\put(10.000,59.341){\Bbox}
        \put(12.000,57.876){\Bbox}\put(14.000,56.411){\Bbox}
        \put(16.000,54.946){\Bbox}\put(18.000,53.481){\Bbox}
        \put(20.000,52.016){\Bbox}\put(22.000,50.550){\Bbox}
        \put(24.000,49.085){\Bbox}\put(26.000,47.620){\Bbox}
        \put(28.000,46.155){\Bbox}\put(30.000,44.690){\Bbox}
        \put(32.000,43.225){\Bbox}\put(34.000,41.760){\Bbox}
        \put(36.000,40.295){\Bbox}\put(38.000,38.830){\Bbox}
        \put(40.000,37.365){\Bbox}\put(42.000,35.899){\Bbox}
        \put(44.000,34.434){\Bbox}\put(46.000,32.969){\Bbox}
        \put(48.000,31.504){\Bbox}\put(50.000,30.039){\Bbox}
        \put(52.000,28.574){\Bbox}\put(54.000,27.109){\Bbox}
        \put(56.000,25.644){\Bbox}\put(58.000,24.179){\Bbox}
        \put(60.000,22.713){\Bbox}\put(62.000,21.248){\Bbox}
        \put(64.000,19.783){\Bbox}\put(66.000,18.318){\Bbox}
        \put(68.000,16.853){\Bbox}\put(70.000,15.388){\Bbox}
        \put(72.000,13.923){\Bbox}\put(74.000,12.458){\Bbox}
        \put(76.000,10.993){\Bbox}\put(78.000,9.537){\Bbox}
        \put(80.000,8.147){\Bbox}\put(82.000,6.838){\Bbox}
        \put(84.000,5.615){\Bbox}\put(86.000,4.484){\Bbox}
        \put(88.000,3.451){\Bbox}\put(90.000,2.527){\Bbox}
        \put(92.000,1.719){\Bbox}\put(94.000,1.041){\Bbox}
        \put(96.000,0.509){\Bbox}\put(98.000,0.147){\Bbox}
        \put(100.000,0.000){\Bbox}
    \color{red}
        \qbezier(0.000,66.667)(0.351,65.933)(2.000,64.133)
        \qbezier(2.000,64.133)(2.879,63.173)(4.000,62.080)
        \qbezier(4.000,62.080)(4.927,61.176)(6.000,60.200)
        \qbezier(6.000,60.200)(6.947,59.339)(8.000,58.429)
        \qbezier(8.000,58.429)(8.958,57.601)(10.000,56.736)
        \qbezier(10.000,56.736)(10.965,55.935)(12.000,55.103)
        \qbezier(12.000,55.103)(12.969,54.325)(14.000,53.519)
        \qbezier(14.000,53.519)(14.972,52.759)(16.000,51.976)
        \qbezier(16.000,51.976)(16.975,51.232)(18.000,50.466)
        \qbezier(18.000,50.466)(18.977,49.737)(20.000,48.987)
        \qbezier(20.000,48.987)(20.978,48.269)(22.000,47.532)
        \qbezier(22.000,47.532)(22.979,46.826)(24.000,46.101)
        \qbezier(24.000,46.101)(24.980,45.404)(26.000,44.688)
        \qbezier(26.000,44.688)(26.980,44.001)(28.000,43.294)
        \qbezier(28.000,43.294)(28.980,42.614)(30.000,41.915)
        \qbezier(30.000,41.915)(30.980,41.243)(32.000,40.550)
        \qbezier(32.000,40.550)(32.980,39.884)(34.000,39.197)
        \qbezier(34.000,39.197)(34.979,38.538)(36.000,37.856)
        \qbezier(36.000,37.856)(36.978,37.202)(38.000,36.524)
        \qbezier(38.000,36.524)(38.976,35.876)(40.000,35.200)
        \qbezier(40.000,35.200)(40.973,34.558)(42.000,33.883)
        \qbezier(42.000,33.883)(42.969,33.247)(44.000,32.573)
        \qbezier(44.000,32.573)(44.963,31.944)(46.000,31.268)
        \qbezier(46.000,31.268)(46.951,30.649)(48.000,29.967)
        \qbezier(48.000,29.967)(48.929,29.364)(50.000,28.670)
        \qbezier(50.000,28.670)(50.868,28.107)(52.000,27.374)
        \qbezier(52.000,27.374)(52.225,27.228)(54.000,26.079)
        \qbezier(54.000,26.079)(55.220,25.288)(56.000,24.783)
        \qbezier(56.000,24.783)(57.104,24.068)(58.000,23.486)
        \qbezier(58.000,23.486)(59.073,22.790)(60.000,22.186)
        \qbezier(60.000,22.186)(61.060,21.496)(62.000,20.881)
        \qbezier(62.000,20.881)(63.055,20.190)(64.000,19.568)
        \qbezier(64.000,19.568)(65.055,18.873)(66.000,18.245)
        \qbezier(66.000,18.245)(67.061,17.541)(68.000,16.909)
        \qbezier(68.000,16.909)(69.077,16.184)(70.000,15.552)
        \qbezier(70.000,15.552)(71.138,14.771)(72.000,14.161)
        \qbezier(72.000,14.161)(73.000,13.350)(74.000,12.539)
        \qbezier(74.000,12.539)(75.009,11.746)(76.000,11.002)
        \qbezier(76.000,11.002)(77.010,10.244)(78.000,9.537)
        \qbezier(78.000,9.537)(79.010,8.815)(80.000,8.147)
        \qbezier(80.000,8.147)(81.011,7.465)(82.000,6.838)
        \qbezier(82.000,6.838)(83.011,6.197)(84.000,5.615)
        \qbezier(84.000,5.615)(85.012,5.018)(86.000,4.484)
        \qbezier(86.000,4.484)(87.013,3.935)(88.000,3.451)
        \qbezier(88.000,3.451)(89.015,2.954)(90.000,2.527)
        \qbezier(90.000,2.527)(91.017,2.085)(92.000,1.719)
        \qbezier(92.000,1.719)(93.020,1.339)(94.000,1.041)
        \qbezier(94.000,1.041)(95.024,0.730)(96.000,0.509)
        \qbezier(96.000,0.509)(97.035,0.276)(98.000,0.147)
        \qbezier(98.000,0.147)(99.099,0.001)(100.000,0.000)
    \color{blue}
        \qbezier(0.000,66.667)(0.362,65.915)(2.000,64.147)
        \qbezier(2.000,64.147)(2.887,63.189)(4.000,62.129)
        \qbezier(4.000,62.129)(4.934,61.239)(6.000,60.303)
        \qbezier(6.000,60.303)(6.951,59.467)(8.000,58.602)
        \qbezier(8.000,58.602)(8.964,57.807)(10.000,56.995)
        \qbezier(10.000,56.995)(10.971,56.235)(12.000,55.464)
        \qbezier(12.000,55.464)(12.976,54.733)(14.000,53.996)
        \qbezier(14.000,53.996)(14.971,53.297)(16.000,52.582)
        \qbezier(16.000,52.582)(16.043,52.552)(18.000,51.193)
        \qbezier(18.000,51.193)(19.752,49.977)(20.000,49.804)
        \qbezier(20.000,49.804)(21.000,49.110)(22.000,48.415)
        \qbezier(22.000,48.415)(23.000,47.721)(24.000,47.026)
        \qbezier(24.000,47.026)(25.790,45.784)(26.000,45.637)
        \qbezier(26.000,45.637)(26.082,45.581)(28.000,44.249)
        \qbezier(28.000,44.249)(29.644,43.107)(30.000,42.860)
        \qbezier(30.000,42.860)(31.000,42.165)(32.000,41.471)
        \qbezier(32.000,41.471)(33.000,40.776)(34.000,40.082)
        \qbezier(34.000,40.082)(35.000,39.387)(36.000,38.693)
        \qbezier(36.000,38.693)(37.023,37.982)(38.000,37.304)
        \qbezier(38.000,37.304)(38.481,36.970)(40.000,35.915)
        \qbezier(40.000,35.915)(41.000,35.221)(42.000,34.526)
        \qbezier(42.000,34.526)(43.530,33.464)(44.000,33.137)
        \qbezier(44.000,33.137)(45.000,32.443)(46.000,31.749)
        \qbezier(46.000,31.749)(46.241,31.581)(48.000,30.360)
        \qbezier(48.000,30.360)(49.000,29.665)(50.000,28.971)
        \qbezier(50.000,28.971)(50.060,28.929)(52.000,27.582)
        \qbezier(52.000,27.582)(53.000,26.887)(54.000,26.193)
        \qbezier(54.000,26.193)(55.000,25.499)(56.000,24.804)
        \qbezier(56.000,24.804)(57.000,24.110)(58.000,23.415)
        \qbezier(58.000,23.415)(59.408,22.438)(60.000,22.026)
        \qbezier(60.000,22.026)(61.000,21.332)(62.000,20.637)
        \qbezier(62.000,20.637)(63.000,19.943)(64.000,19.249)
        \qbezier(64.000,19.249)(65.050,18.520)(66.000,17.860)
        \qbezier(66.000,17.860)(66.346,17.619)(68.000,16.471)
        \qbezier(68.000,16.471)(69.847,15.188)(70.000,15.082)
        \qbezier(70.000,15.082)(71.000,14.388)(72.000,13.693)
        \qbezier(72.000,13.693)(73.000,12.999)(74.000,12.304)
        \qbezier(74.000,12.304)(75.000,11.610)(76.000,10.915)
        \qbezier(76.000,10.915)(77.000,10.221)(78.000,9.526)
        \qbezier(78.000,9.526)(79.501,8.484)(80.000,8.147)
        \qbezier(80.000,8.147)(81.005,7.468)(82.000,6.838)
        \qbezier(82.000,6.838)(83.019,6.192)(84.000,5.615)
        \qbezier(84.000,5.615)(85.025,5.011)(86.000,4.484)
        \qbezier(86.000,4.484)(87.014,3.935)(88.000,3.451)
        \qbezier(88.000,3.451)(89.008,2.957)(90.000,2.527)
        \qbezier(90.000,2.527)(91.010,2.088)(92.000,1.719)
        \qbezier(92.000,1.719)(93.037,1.333)(94.000,1.041)
        \qbezier(94.000,1.041)(95.041,0.726)(96.000,0.509)
        \qbezier(96.000,0.509)(97.029,0.277)(98.000,0.147)
        \qbezier(98.000,0.147)(99.078,0.004)(100.000,0.000)
        \put(0.000,66.667){\Bplus}\put(2.000,64.147){\Bplus}
        \put(4.000,62.129){\Bplus}\put(6.000,60.303){\Bplus}
        \put(8.000,58.602){\Bplus}\put(10.000,56.995){\Bplus}
        \put(12.000,55.464){\Bplus}\put(14.000,53.996){\Bplus}
        \put(16.000,52.582){\Bplus}\put(18.000,51.193){\Bplus}
        \put(20.000,49.804){\Bplus}\put(22.000,48.415){\Bplus}
        \put(24.000,47.026){\Bplus}\put(26.000,45.637){\Bplus}
        \put(28.000,44.249){\Bplus}\put(30.000,42.860){\Bplus}
        \put(32.000,41.471){\Bplus}\put(34.000,40.082){\Bplus}
        \put(36.000,38.693){\Bplus}\put(38.000,37.304){\Bplus}
        \put(40.000,35.915){\Bplus}\put(42.000,34.526){\Bplus}
        \put(44.000,33.137){\Bplus}\put(46.000,31.749){\Bplus}
        \put(48.000,30.360){\Bplus}\put(50.000,28.971){\Bplus}
        \put(52.000,27.582){\Bplus}\put(54.000,26.193){\Bplus}
        \put(56.000,24.804){\Bplus}\put(58.000,23.415){\Bplus}
        \put(60.000,22.026){\Bplus}\put(62.000,20.637){\Bplus}
        \put(64.000,19.249){\Bplus}\put(66.000,17.860){\Bplus}
        \put(68.000,16.471){\Bplus}\put(70.000,15.082){\Bplus}
        \put(72.000,13.693){\Bplus}\put(74.000,12.304){\Bplus}
        \put(76.000,10.915){\Bplus}\put(78.000,9.526){\Bplus}
        \put(80.000,8.147){\Bplus}\put(82.000,6.838){\Bplus}
        \put(84.000,5.615){\Bplus}\put(86.000,4.484){\Bplus}
        \put(88.000,3.451){\Bplus}\put(90.000,2.527){\Bplus}
        \put(92.000,1.719){\Bplus}\put(94.000,1.041){\Bplus}
        \put(96.000,0.509){\Bplus}\put(98.000,0.147){\Bplus}
        \put(100.000,0.000){\Bplus}
    \color{cyan}
        \qbezier(0.000,66.667)(0.351,65.933)(2.000,64.133)
        \qbezier(2.000,64.133)(2.879,63.173)(4.000,62.080)
        \qbezier(4.000,62.080)(4.927,61.176)(6.000,60.200)
        \qbezier(6.000,60.200)(6.947,59.339)(8.000,58.429)
        \qbezier(8.000,58.429)(8.958,57.601)(10.000,56.736)
        \qbezier(10.000,56.736)(10.965,55.935)(12.000,55.103)
        \qbezier(12.000,55.103)(12.969,54.325)(14.000,53.519)
        \qbezier(14.000,53.519)(14.972,52.759)(16.000,51.976)
        \qbezier(16.000,51.976)(16.975,51.232)(18.000,50.466)
        \qbezier(18.000,50.466)(18.977,49.737)(20.000,48.987)
        \qbezier(20.000,48.987)(20.978,48.269)(22.000,47.532)
        \qbezier(22.000,47.532)(22.979,46.826)(24.000,46.101)
        \qbezier(24.000,46.101)(24.980,45.404)(26.000,44.689)
        \qbezier(26.000,44.689)(26.980,44.001)(28.000,43.294)
        \qbezier(28.000,43.294)(28.980,42.614)(30.000,41.915)
        \qbezier(30.000,41.915)(30.980,41.243)(32.000,40.550)
        \qbezier(32.000,40.550)(32.980,39.885)(34.000,39.197)
        \qbezier(34.000,39.197)(34.979,38.538)(36.000,37.856)
        \qbezier(36.000,37.856)(36.978,37.202)(38.000,36.524)
        \qbezier(38.000,36.524)(38.976,35.876)(40.000,35.200)
        \qbezier(40.000,35.200)(40.973,34.558)(42.000,33.883)
        \qbezier(42.000,33.883)(42.969,33.247)(44.000,32.573)
        \qbezier(44.000,32.573)(44.963,31.944)(46.000,31.268)
        \qbezier(46.000,31.268)(46.951,30.649)(48.000,29.967)
        \qbezier(48.000,29.967)(48.929,29.364)(50.000,28.670)
        \qbezier(50.000,28.670)(50.868,28.107)(52.000,27.374)
        \qbezier(52.000,27.374)(52.225,27.228)(54.000,26.079)
        \qbezier(54.000,26.079)(55.220,25.288)(56.000,24.783)
        \qbezier(56.000,24.783)(56.449,24.492)(58.000,23.415)
        \qbezier(58.000,23.415)(59.000,22.721)(60.000,22.026)
        \qbezier(60.000,22.026)(61.000,21.332)(62.000,20.637)
        \qbezier(62.000,20.637)(63.000,19.943)(64.000,19.249)
        \qbezier(64.000,19.249)(65.000,18.554)(66.000,17.860)
        \qbezier(66.000,17.860)(67.000,17.165)(68.000,16.471)
        \qbezier(68.000,16.471)(69.000,15.776)(70.000,15.082)
        \qbezier(70.000,15.082)(71.000,14.387)(72.000,13.693)
        \qbezier(72.000,13.693)(73.000,12.999)(74.000,12.304)
        \qbezier(74.000,12.304)(75.000,11.610)(76.000,10.915)
        \qbezier(76.000,10.915)(77.000,10.221)(78.000,9.526)
        \qbezier(78.000,9.526)(79.514,8.475)(80.000,8.147)
        \qbezier(80.000,8.147)(81.011,7.465)(82.000,6.838)
        \qbezier(82.000,6.838)(83.011,6.197)(84.000,5.615)
        \qbezier(84.000,5.615)(85.012,5.018)(86.000,4.484)
        \qbezier(86.000,4.484)(87.013,3.935)(88.000,3.451)
        \qbezier(88.000,3.451)(89.015,2.954)(90.000,2.527)
        \qbezier(90.000,2.527)(91.017,2.085)(92.000,1.719)
        \qbezier(92.000,1.719)(93.020,1.339)(94.000,1.041)
        \qbezier(94.000,1.041)(95.024,0.730)(96.000,0.509)
        \qbezier(96.000,0.509)(97.035,0.276)(98.000,0.147)
        \qbezier(98.000,0.147)(99.097,0.002)(100.000,0.000)
        \put(0.000,66.667){\Bfull}\put(2.000,64.133){\Bfull}
        \put(4.000,62.080){\Bfull}\put(6.000,60.200){\Bfull}
        \put(8.000,58.429){\Bfull}\put(10.000,56.736){\Bfull}
        \put(12.000,55.103){\Bfull}\put(14.000,53.519){\Bfull}
        \put(16.000,51.976){\Bfull}\put(18.000,50.466){\Bfull}
        \put(20.000,48.987){\Bfull}\put(22.000,47.532){\Bfull}
        \put(24.000,46.101){\Bfull}\put(26.000,44.689){\Bfull}
        \put(28.000,43.294){\Bfull}\put(30.000,41.915){\Bfull}
        \put(32.000,40.550){\Bfull}\put(34.000,39.197){\Bfull}
        \put(36.000,37.856){\Bfull}\put(38.000,36.524){\Bfull}
        \put(40.000,35.200){\Bfull}\put(42.000,33.883){\Bfull}
        \put(44.000,32.573){\Bfull}\put(46.000,31.268){\Bfull}
        \put(48.000,29.967){\Bfull}\put(50.000,28.670){\Bfull}
        \put(52.000,27.374){\Bfull}\put(54.000,26.079){\Bfull}
        \put(56.000,24.783){\Bfull}\put(58.000,23.415){\Bfull}
        \put(60.000,22.026){\Bfull}\put(62.000,20.637){\Bfull}
        \put(64.000,19.249){\Bfull}\put(66.000,17.860){\Bfull}
        \put(68.000,16.471){\Bfull}\put(70.000,15.082){\Bfull}
        \put(72.000,13.693){\Bfull}\put(74.000,12.304){\Bfull}
        \put(76.000,10.915){\Bfull}\put(78.000,9.526){\Bfull}
        \put(80.000,8.147){\Bfull}\put(82.000,6.838){\Bfull}
        \put(84.000,5.615){\Bfull}\put(86.000,4.484){\Bfull}
        \put(88.000,3.451){\Bfull}\put(90.000,2.527){\Bfull}
        \put(92.000,1.719){\Bfull}\put(94.000,1.041){\Bfull}
        \put(96.000,0.509){\Bfull}\put(98.000,0.147){\Bfull}
        \put(100.000,0.000){\Bfull}
    }
}
\caption{Bounds for $q = 2$ and $n = 3$.}
\label{fig:q=2,n=3}
\end{figure*}

\begin{figure*}[hbt]
%% q = 2, n =  4:
\Plot{%
    \put(035,095){
	\put(000,000){\Legend}
    }
    \put(000,000){\color{brown}
        \put(0.000,75.000){\Bbox}\put(2.000,72.334){\Bbox}
        \put(4.000,70.171){\Bbox}\put(6.000,68.201){\Bbox}
        \put(8.000,66.356){\Bbox}\put(10.000,64.606){\Bbox}
        \put(12.000,62.933){\Bbox}\put(14.000,61.323){\Bbox}
        \put(16.000,59.769){\Bbox}\put(18.000,58.263){\Bbox}
        \put(20.000,56.802){\Bbox}\put(22.000,55.380){\Bbox}
        \put(24.000,53.995){\Bbox}\put(26.000,52.644){\Bbox}
        \put(28.000,51.324){\Bbox}\put(30.000,50.034){\Bbox}
        \put(32.000,48.772){\Bbox}\put(34.000,47.537){\Bbox}
        \put(36.000,46.327){\Bbox}\put(38.000,45.141){\Bbox}
        \put(40.000,43.978){\Bbox}\put(42.000,42.837){\Bbox}
        \put(44.000,41.717){\Bbox}\put(46.000,40.619){\Bbox}
        \put(48.000,39.540){\Bbox}\put(50.000,38.480){\Bbox}
        \put(52.000,37.439){\Bbox}\put(54.000,36.417){\Bbox}
        \put(56.000,35.412){\Bbox}\put(58.000,34.425){\Bbox}
        \put(60.000,33.455){\Bbox}\put(62.000,32.501){\Bbox}
        \put(64.000,31.564){\Bbox}\put(66.000,30.643){\Bbox}
        \put(68.000,29.737){\Bbox}\put(70.000,28.847){\Bbox}
        \put(72.000,27.972){\Bbox}\put(74.000,27.112){\Bbox}
        \put(76.000,26.267){\Bbox}\put(78.000,25.436){\Bbox}
        \put(80.000,24.620){\Bbox}\put(82.000,23.818){\Bbox}
        \put(84.000,23.030){\Bbox}\put(86.000,22.256){\Bbox}
        \put(88.000,21.496){\Bbox}\put(90.000,20.749){\Bbox}
        \put(92.000,20.016){\Bbox}\put(94.000,19.297){\Bbox}
        \put(96.000,18.591){\Bbox}\put(98.000,17.898){\Bbox}
        \put(100.000,17.218){\Bbox}
    \color{black}
        \put(0.000,75.000){\Bfull}\put(2.000,70.171){\Bfull}
        \put(4.000,66.356){\Bfull}\put(6.000,62.933){\Bfull}
        \put(8.000,59.769){\Bfull}\put(10.000,56.802){\Bfull}
        \put(12.000,53.995){\Bfull}\put(14.000,51.324){\Bfull}
        \put(16.000,48.772){\Bfull}\put(18.000,46.327){\Bfull}
        \put(20.000,43.978){\Bfull}\put(22.000,41.717){\Bfull}
        \put(24.000,39.540){\Bfull}\put(26.000,37.439){\Bfull}
        \put(28.000,35.412){\Bfull}\put(30.000,33.455){\Bfull}
        \put(32.000,31.564){\Bfull}\put(34.000,29.737){\Bfull}
        \put(36.000,27.972){\Bfull}\put(38.000,26.267){\Bfull}
        \put(40.000,24.620){\Bfull}\put(42.000,23.030){\Bfull}
        \put(44.000,21.496){\Bfull}\put(46.000,20.016){\Bfull}
        \put(48.000,18.591){\Bfull}\put(50.000,17.218){\Bfull}
        \put(52.000,15.898){\Bfull}\put(54.000,14.630){\Bfull}
        \put(56.000,13.414){\Bfull}\put(58.000,12.249){\Bfull}
        \put(60.000,11.136){\Bfull}\put(62.000,10.074){\Bfull}
        \put(64.000,9.063){\Bfull}\put(66.000,8.104){\Bfull}
        \put(68.000,7.197){\Bfull}\put(70.000,6.341){\Bfull}
        \put(72.000,5.537){\Bfull}\put(74.000,4.786){\Bfull}
        \put(76.000,4.088){\Bfull}\put(78.000,3.443){\Bfull}
        \put(80.000,2.851){\Bfull}\put(82.000,2.314){\Bfull}
        \put(84.000,1.832){\Bfull}\put(86.000,1.405){\Bfull}
        \put(88.000,1.034){\Bfull}\put(90.000,0.719){\Bfull}
        \put(92.000,0.461){\Bfull}\put(94.000,0.259){\Bfull}
        \put(96.000,0.115){\Bfull}\put(98.000,0.029){\Bfull}
        \put(100.000,0.000){\Bfull}
    \color{green}
        \qbezier(0.000,75.000)(0.889,74.248)(2.000,73.309)
        \qbezier(2.000,73.309)(3.480,72.057)(4.000,71.618)
        \qbezier(4.000,71.618)(4.317,71.350)(6.000,69.927)
        \qbezier(6.000,69.927)(7.539,68.626)(8.000,68.236)
        \qbezier(8.000,68.236)(9.000,67.390)(10.000,66.545)
        \qbezier(10.000,66.545)(11.000,65.699)(12.000,64.854)
        \qbezier(12.000,64.854)(13.000,64.008)(14.000,63.163)
        \qbezier(14.000,63.163)(15.000,62.317)(16.000,61.471)
        \qbezier(16.000,61.471)(17.000,60.626)(18.000,59.780)
        \qbezier(18.000,59.780)(19.000,58.935)(20.000,58.089)
        \qbezier(20.000,58.089)(21.058,57.195)(22.000,56.398)
        \qbezier(22.000,56.398)(23.386,55.226)(24.000,54.707)
        \qbezier(24.000,54.707)(24.858,53.982)(26.000,53.016)
        \qbezier(26.000,53.016)(27.000,52.171)(28.000,51.325)
        \qbezier(28.000,51.325)(29.000,50.480)(30.000,49.634)
        \qbezier(30.000,49.634)(31.000,48.788)(32.000,47.943)
        \qbezier(32.000,47.943)(33.000,47.097)(34.000,46.252)
        \qbezier(34.000,46.252)(34.785,45.588)(36.000,44.561)
        \qbezier(36.000,44.561)(37.000,43.715)(38.000,42.870)
        \qbezier(38.000,42.870)(39.094,41.945)(40.000,41.179)
        \qbezier(40.000,41.179)(41.000,40.333)(42.000,39.488)
        \qbezier(42.000,39.488)(43.259,38.423)(44.000,37.796)
        \qbezier(44.000,37.796)(45.000,36.951)(46.000,36.105)
        \qbezier(46.000,36.105)(47.000,35.260)(48.000,34.414)
        \qbezier(48.000,34.414)(49.000,33.569)(50.000,32.723)
        \qbezier(50.000,32.723)(51.531,31.429)(52.000,31.032)
        \qbezier(52.000,31.032)(53.000,30.187)(54.000,29.341)
        \qbezier(54.000,29.341)(55.000,28.496)(56.000,27.650)
        \qbezier(56.000,27.650)(57.000,26.805)(58.000,25.959)
        \qbezier(58.000,25.959)(59.000,25.113)(60.000,24.268)
        \qbezier(60.000,24.268)(61.000,23.422)(62.000,22.577)
        \qbezier(62.000,22.577)(62.911,21.807)(64.000,20.886)
        \qbezier(64.000,20.886)(65.543,19.581)(66.000,19.195)
        \qbezier(66.000,19.195)(67.000,18.349)(68.000,17.504)
        \qbezier(68.000,17.504)(69.445,16.282)(70.000,15.813)
        \qbezier(70.000,15.813)(71.153,14.838)(72.000,14.144)
        \qbezier(72.000,14.144)(73.005,13.321)(74.000,12.539)
        \qbezier(74.000,12.539)(75.005,11.749)(76.000,11.002)
        \qbezier(76.000,11.002)(77.005,10.247)(78.000,9.537)
        \qbezier(78.000,9.537)(79.021,8.808)(80.000,8.147)
        \qbezier(80.000,8.147)(81.022,7.457)(82.000,6.838)
        \qbezier(82.000,6.838)(83.007,6.200)(84.000,5.615)
        \qbezier(84.000,5.615)(85.006,5.022)(86.000,4.484)
        \qbezier(86.000,4.484)(87.007,3.938)(88.000,3.451)
        \qbezier(88.000,3.451)(89.031,2.946)(90.000,2.527)
        \qbezier(90.000,2.527)(91.033,2.080)(92.000,1.719)
        \qbezier(92.000,1.719)(93.013,1.342)(94.000,1.041)
        \qbezier(94.000,1.041)(95.018,0.732)(96.000,0.509)
        \qbezier(96.000,0.509)(97.029,0.277)(98.000,0.147)
        \qbezier(98.000,0.147)(99.094,0.002)(100.000,0.000)
        \put(0.000,75.000){\Bbox}\put(2.000,73.309){\Bbox}
        \put(4.000,71.618){\Bbox}\put(6.000,69.927){\Bbox}
        \put(8.000,68.236){\Bbox}\put(10.000,66.545){\Bbox}
        \put(12.000,64.854){\Bbox}\put(14.000,63.163){\Bbox}
        \put(16.000,61.471){\Bbox}\put(18.000,59.780){\Bbox}
        \put(20.000,58.089){\Bbox}\put(22.000,56.398){\Bbox}
        \put(24.000,54.707){\Bbox}\put(26.000,53.016){\Bbox}
        \put(28.000,51.325){\Bbox}\put(30.000,49.634){\Bbox}
        \put(32.000,47.943){\Bbox}\put(34.000,46.252){\Bbox}
        \put(36.000,44.561){\Bbox}\put(38.000,42.870){\Bbox}
        \put(40.000,41.179){\Bbox}\put(42.000,39.488){\Bbox}
        \put(44.000,37.796){\Bbox}\put(46.000,36.105){\Bbox}
        \put(48.000,34.414){\Bbox}\put(50.000,32.723){\Bbox}
        \put(52.000,31.032){\Bbox}\put(54.000,29.341){\Bbox}
        \put(56.000,27.650){\Bbox}\put(58.000,25.959){\Bbox}
        \put(60.000,24.268){\Bbox}\put(62.000,22.577){\Bbox}
        \put(64.000,20.886){\Bbox}\put(66.000,19.195){\Bbox}
        \put(68.000,17.504){\Bbox}\put(70.000,15.813){\Bbox}
        \put(72.000,14.144){\Bbox}\put(74.000,12.539){\Bbox}
        \put(76.000,11.002){\Bbox}\put(78.000,9.537){\Bbox}
        \put(80.000,8.147){\Bbox}\put(82.000,6.838){\Bbox}
        \put(84.000,5.615){\Bbox}\put(86.000,4.484){\Bbox}
        \put(88.000,3.451){\Bbox}\put(90.000,2.527){\Bbox}
        \put(92.000,1.719){\Bbox}\put(94.000,1.041){\Bbox}
        \put(96.000,0.509){\Bbox}\put(98.000,0.147){\Bbox}
        \put(100.000,0.000){\Bbox}
    \color{red}
        \qbezier(0.000,75.000)(0.350,74.242)(2.000,72.319)
        \qbezier(2.000,72.319)(2.878,71.295)(4.000,70.117)
        \qbezier(4.000,70.117)(4.926,69.144)(6.000,68.086)
        \qbezier(6.000,68.086)(6.945,67.154)(8.000,66.162)
        \qbezier(8.000,66.162)(8.956,65.262)(10.000,64.314)
        \qbezier(10.000,64.314)(10.963,63.439)(12.000,62.524)
        \qbezier(12.000,62.524)(12.967,61.670)(14.000,60.780)
        \qbezier(14.000,60.780)(14.970,59.945)(16.000,59.075)
        \qbezier(16.000,59.075)(16.972,58.254)(18.000,57.402)
        \qbezier(18.000,57.402)(18.974,56.594)(20.000,55.755)
        \qbezier(20.000,55.755)(20.974,54.959)(22.000,54.131)
        \qbezier(22.000,54.131)(22.975,53.345)(24.000,52.527)
        \qbezier(24.000,52.527)(24.974,51.750)(26.000,50.940)
        \qbezier(26.000,50.940)(26.973,50.171)(28.000,49.367)
        \qbezier(28.000,49.367)(28.972,48.606)(30.000,47.806)
        \qbezier(30.000,47.806)(30.969,47.053)(32.000,46.256)
        \qbezier(32.000,46.256)(32.965,45.511)(34.000,44.715)
        \qbezier(34.000,44.715)(34.958,43.979)(36.000,43.180)
        \qbezier(36.000,43.180)(36.945,42.457)(38.000,41.651)
        \qbezier(38.000,41.651)(38.919,40.950)(40.000,40.126)
        \qbezier(40.000,40.126)(40.839,39.487)(42.000,38.603)
        \qbezier(42.000,38.603)(43.000,37.842)(44.000,37.081)
        \qbezier(44.000,37.081)(45.175,36.187)(46.000,35.558)
        \qbezier(46.000,35.558)(47.092,34.726)(48.000,34.033)
        \qbezier(48.000,34.033)(49.066,33.219)(50.000,32.503)
        \qbezier(50.000,32.503)(51.055,31.695)(52.000,30.967)
        \qbezier(52.000,30.967)(53.050,30.158)(54.000,29.421)
        \qbezier(54.000,29.421)(55.050,28.607)(56.000,27.864)
        \qbezier(56.000,27.864)(57.054,27.039)(58.000,26.289)
        \qbezier(58.000,26.289)(59.067,25.444)(60.000,24.691)
        \qbezier(60.000,24.691)(61.111,23.796)(62.000,23.058)
        \qbezier(62.000,23.058)(63.000,22.117)(64.000,21.177)
        \qbezier(64.000,21.177)(65.008,20.234)(66.000,19.332)
        \qbezier(66.000,19.332)(67.008,18.417)(68.000,17.544)
        \qbezier(68.000,17.544)(69.008,16.657)(70.000,15.813)
        \qbezier(70.000,15.813)(71.008,14.956)(72.000,14.144)
        \qbezier(72.000,14.144)(73.009,13.318)(74.000,12.539)
        \qbezier(74.000,12.539)(75.009,11.746)(76.000,11.002)
        \qbezier(76.000,11.002)(77.010,10.244)(78.000,9.537)
        \qbezier(78.000,9.537)(79.010,8.815)(80.000,8.147)
        \qbezier(80.000,8.147)(81.011,7.465)(82.000,6.838)
        \qbezier(82.000,6.838)(83.011,6.197)(84.000,5.615)
        \qbezier(84.000,5.615)(85.012,5.018)(86.000,4.484)
        \qbezier(86.000,4.484)(87.013,3.935)(88.000,3.451)
        \qbezier(88.000,3.451)(89.015,2.954)(90.000,2.527)
        \qbezier(90.000,2.527)(91.017,2.085)(92.000,1.719)
        \qbezier(92.000,1.719)(93.020,1.339)(94.000,1.041)
        \qbezier(94.000,1.041)(95.024,0.730)(96.000,0.509)
        \qbezier(96.000,0.509)(97.035,0.276)(98.000,0.147)
        \qbezier(98.000,0.147)(99.099,0.001)(100.000,0.000)
    \color{blue}
        \qbezier(0.000,75.000)(0.363,74.219)(2.000,72.334)
        \qbezier(2.000,72.334)(2.887,71.312)(4.000,70.171)
        \qbezier(4.000,70.171)(4.934,69.214)(6.000,68.201)
        \qbezier(6.000,68.201)(6.954,67.294)(8.000,66.356)
        \qbezier(8.000,66.356)(8.963,65.493)(10.000,64.606)
        \qbezier(10.000,64.606)(10.971,63.776)(12.000,62.933)
        \qbezier(12.000,62.933)(12.217,62.754)(14.000,61.305)
        \qbezier(14.000,61.305)(15.000,60.492)(16.000,59.679)
        \qbezier(16.000,59.679)(16.419,59.338)(18.000,58.052)
        \qbezier(18.000,58.052)(19.000,57.239)(20.000,56.426)
        \qbezier(20.000,56.426)(21.000,55.613)(22.000,54.800)
        \qbezier(22.000,54.800)(23.701,53.417)(24.000,53.174)
        \qbezier(24.000,53.174)(24.129,53.069)(26.000,51.547)
        \qbezier(26.000,51.547)(27.000,50.734)(28.000,49.921)
        \qbezier(28.000,49.921)(29.000,49.108)(30.000,48.295)
        \qbezier(30.000,48.295)(31.000,47.482)(32.000,46.668)
        \qbezier(32.000,46.668)(33.000,45.855)(34.000,45.042)
        \qbezier(34.000,45.042)(35.000,44.229)(36.000,43.416)
        \qbezier(36.000,43.416)(37.000,42.603)(38.000,41.790)
        \qbezier(38.000,41.790)(39.000,40.976)(40.000,40.163)
        \qbezier(40.000,40.163)(41.000,39.350)(42.000,38.537)
        \qbezier(42.000,38.537)(42.696,37.971)(44.000,36.911)
        \qbezier(44.000,36.911)(45.000,36.098)(46.000,35.285)
        \qbezier(46.000,35.285)(47.000,34.471)(48.000,33.658)
        \qbezier(48.000,33.658)(49.000,32.845)(50.000,32.032)
        \qbezier(50.000,32.032)(51.471,30.836)(52.000,30.406)
        \qbezier(52.000,30.406)(53.137,29.482)(54.000,28.779)
        \qbezier(54.000,28.779)(55.000,27.966)(56.000,27.153)
        \qbezier(56.000,27.153)(57.000,26.340)(58.000,25.527)
        \qbezier(58.000,25.527)(58.610,25.031)(60.000,23.901)
        \qbezier(60.000,23.901)(61.000,23.088)(62.000,22.274)
        \qbezier(62.000,22.274)(63.000,21.461)(64.000,20.648)
        \qbezier(64.000,20.648)(65.000,19.835)(66.000,19.022)
        \qbezier(66.000,19.022)(67.566,17.749)(68.000,17.396)
        \qbezier(68.000,17.396)(68.705,16.822)(70.000,15.769)
        \qbezier(70.000,15.769)(71.569,14.494)(72.000,14.143)
        \qbezier(72.000,14.143)(73.176,13.186)(74.000,12.539)
        \qbezier(74.000,12.539)(75.005,11.750)(76.000,11.002)
        \qbezier(76.000,11.002)(77.005,10.247)(78.000,9.537)
        \qbezier(78.000,9.537)(79.005,8.819)(80.000,8.147)
        \qbezier(80.000,8.147)(81.005,7.468)(82.000,6.838)
        \qbezier(82.000,6.838)(83.026,6.188)(84.000,5.615)
        \qbezier(84.000,5.615)(85.025,5.012)(86.000,4.484)
        \qbezier(86.000,4.484)(87.007,3.938)(88.000,3.451)
        \qbezier(88.000,3.451)(89.008,2.957)(90.000,2.527)
        \qbezier(90.000,2.527)(91.010,2.088)(92.000,1.719)
        \qbezier(92.000,1.719)(93.037,1.333)(94.000,1.041)
        \qbezier(94.000,1.041)(95.041,0.726)(96.000,0.509)
        \qbezier(96.000,0.509)(97.029,0.277)(98.000,0.147)
        \qbezier(98.000,0.147)(99.078,0.004)(100.000,0.000)
        \put(0.000,75.000){\Bplus}\put(2.000,72.334){\Bplus}
        \put(4.000,70.171){\Bplus}\put(6.000,68.201){\Bplus}
        \put(8.000,66.356){\Bplus}\put(10.000,64.606){\Bplus}
        \put(12.000,62.933){\Bplus}\put(14.000,61.305){\Bplus}
        \put(16.000,59.679){\Bplus}\put(18.000,58.052){\Bplus}
        \put(20.000,56.426){\Bplus}\put(22.000,54.800){\Bplus}
        \put(24.000,53.174){\Bplus}\put(26.000,51.547){\Bplus}
        \put(28.000,49.921){\Bplus}\put(30.000,48.295){\Bplus}
        \put(32.000,46.668){\Bplus}\put(34.000,45.042){\Bplus}
        \put(36.000,43.416){\Bplus}\put(38.000,41.790){\Bplus}
        \put(40.000,40.163){\Bplus}\put(42.000,38.537){\Bplus}
        \put(44.000,36.911){\Bplus}\put(46.000,35.285){\Bplus}
        \put(48.000,33.658){\Bplus}\put(50.000,32.032){\Bplus}
        \put(52.000,30.406){\Bplus}\put(54.000,28.779){\Bplus}
        \put(56.000,27.153){\Bplus}\put(58.000,25.527){\Bplus}
        \put(60.000,23.901){\Bplus}\put(62.000,22.274){\Bplus}
        \put(64.000,20.648){\Bplus}\put(66.000,19.022){\Bplus}
        \put(68.000,17.396){\Bplus}\put(70.000,15.769){\Bplus}
        \put(72.000,14.143){\Bplus}\put(74.000,12.539){\Bplus}
        \put(76.000,11.002){\Bplus}\put(78.000,9.537){\Bplus}
        \put(80.000,8.147){\Bplus}\put(82.000,6.838){\Bplus}
        \put(84.000,5.615){\Bplus}\put(86.000,4.484){\Bplus}
        \put(88.000,3.451){\Bplus}\put(90.000,2.527){\Bplus}
        \put(92.000,1.719){\Bplus}\put(94.000,1.041){\Bplus}
        \put(96.000,0.509){\Bplus}\put(98.000,0.147){\Bplus}
        \put(100.000,0.000){\Bplus}
    \color{cyan}
        \qbezier(0.000,75.000)(0.350,74.242)(2.000,72.319)
        \qbezier(2.000,72.319)(2.878,71.295)(4.000,70.117)
        \qbezier(4.000,70.117)(4.926,69.144)(6.000,68.086)
        \qbezier(6.000,68.086)(6.945,67.154)(8.000,66.162)
        \qbezier(8.000,66.162)(8.956,65.262)(10.000,64.314)
        \qbezier(10.000,64.314)(10.963,63.439)(12.000,62.524)
        \qbezier(12.000,62.524)(12.967,61.670)(14.000,60.780)
        \qbezier(14.000,60.780)(14.970,59.945)(16.000,59.075)
        \qbezier(16.000,59.075)(16.972,58.254)(18.000,57.402)
        \qbezier(18.000,57.402)(18.974,56.594)(20.000,55.755)
        \qbezier(20.000,55.755)(20.974,54.959)(22.000,54.131)
        \qbezier(22.000,54.131)(22.975,53.345)(24.000,52.527)
        \qbezier(24.000,52.527)(24.974,51.750)(26.000,50.940)
        \qbezier(26.000,50.940)(26.973,50.171)(28.000,49.367)
        \qbezier(28.000,49.367)(28.972,48.606)(30.000,47.806)
        \qbezier(30.000,47.806)(30.969,47.053)(32.000,46.256)
        \qbezier(32.000,46.256)(32.965,45.511)(34.000,44.715)
        \qbezier(34.000,44.715)(34.958,43.979)(36.000,43.180)
        \qbezier(36.000,43.180)(36.945,42.457)(38.000,41.651)
        \qbezier(38.000,41.651)(38.919,40.950)(40.000,40.126)
        \qbezier(40.000,40.126)(40.728,39.571)(42.000,38.537)
        \qbezier(42.000,38.537)(43.000,37.724)(44.000,36.911)
        \qbezier(44.000,36.911)(45.000,36.098)(46.000,35.285)
        \qbezier(46.000,35.285)(47.000,34.471)(48.000,33.658)
        \qbezier(48.000,33.658)(49.000,32.845)(50.000,32.032)
        \qbezier(50.000,32.032)(51.000,31.219)(52.000,30.406)
        \qbezier(52.000,30.406)(53.000,29.593)(54.000,28.779)
        \qbezier(54.000,28.779)(55.000,27.966)(56.000,27.153)
        \qbezier(56.000,27.153)(57.000,26.340)(58.000,25.527)
        \qbezier(58.000,25.527)(59.000,24.714)(60.000,23.901)
        \qbezier(60.000,23.901)(61.000,23.088)(62.000,22.274)
        \qbezier(62.000,22.274)(63.000,21.461)(64.000,20.648)
        \qbezier(64.000,20.648)(65.000,19.835)(66.000,19.022)
        \qbezier(66.000,19.022)(67.000,18.209)(68.000,17.396)
        \qbezier(68.000,17.396)(69.000,16.582)(70.000,15.769)
        \qbezier(70.000,15.769)(71.285,14.725)(72.000,14.143)
        \qbezier(72.000,14.143)(73.182,13.182)(74.000,12.539)
        \qbezier(74.000,12.539)(75.009,11.746)(76.000,11.002)
        \qbezier(76.000,11.002)(77.010,10.244)(78.000,9.537)
        \qbezier(78.000,9.537)(79.010,8.815)(80.000,8.147)
        \qbezier(80.000,8.147)(81.011,7.465)(82.000,6.838)
        \qbezier(82.000,6.838)(83.011,6.197)(84.000,5.615)
        \qbezier(84.000,5.615)(85.012,5.018)(86.000,4.484)
        \qbezier(86.000,4.484)(87.013,3.935)(88.000,3.451)
        \qbezier(88.000,3.451)(89.015,2.954)(90.000,2.527)
        \qbezier(90.000,2.527)(91.017,2.085)(92.000,1.719)
        \qbezier(92.000,1.719)(93.020,1.339)(94.000,1.041)
        \qbezier(94.000,1.041)(95.024,0.730)(96.000,0.509)
        \qbezier(96.000,0.509)(97.035,0.276)(98.000,0.147)
        \qbezier(98.000,0.147)(99.097,0.002)(100.000,0.000)
    \color{cyan}
        \put(0.000,75.000){\Bfull}\put(2.000,72.319){\Bfull}
        \put(4.000,70.117){\Bfull}\put(6.000,68.086){\Bfull}
        \put(8.000,66.162){\Bfull}\put(10.000,64.314){\Bfull}
        \put(12.000,62.524){\Bfull}\put(14.000,60.780){\Bfull}
        \put(16.000,59.075){\Bfull}\put(18.000,57.402){\Bfull}
        \put(20.000,55.755){\Bfull}\put(22.000,54.131){\Bfull}
        \put(24.000,52.527){\Bfull}\put(26.000,50.940){\Bfull}
        \put(28.000,49.367){\Bfull}\put(30.000,47.806){\Bfull}
        \put(32.000,46.256){\Bfull}\put(34.000,44.715){\Bfull}
        \put(36.000,43.180){\Bfull}\put(38.000,41.651){\Bfull}
        \put(40.000,40.126){\Bfull}\put(42.000,38.537){\Bfull}
        \put(44.000,36.911){\Bfull}\put(46.000,35.285){\Bfull}
        \put(48.000,33.658){\Bfull}\put(50.000,32.032){\Bfull}
        \put(52.000,30.406){\Bfull}\put(54.000,28.779){\Bfull}
        \put(56.000,27.153){\Bfull}\put(58.000,25.527){\Bfull}
        \put(60.000,23.901){\Bfull}\put(62.000,22.274){\Bfull}
        \put(64.000,20.648){\Bfull}\put(66.000,19.022){\Bfull}
        \put(68.000,17.396){\Bfull}\put(70.000,15.769){\Bfull}
        \put(72.000,14.143){\Bfull}\put(74.000,12.539){\Bfull}
        \put(76.000,11.002){\Bfull}\put(78.000,9.537){\Bfull}
        \put(80.000,8.147){\Bfull}\put(82.000,6.838){\Bfull}
        \put(84.000,5.615){\Bfull}\put(86.000,4.484){\Bfull}
        \put(88.000,3.451){\Bfull}\put(90.000,2.527){\Bfull}
        \put(92.000,1.719){\Bfull}\put(94.000,1.041){\Bfull}
        \put(96.000,0.509){\Bfull}\put(98.000,0.147){\Bfull}
        \put(100.000,0.000){\Bfull}
    }
}
\caption{Bounds for $q = 2$ and $n = 4$.}
\label{fig:q=2,n=4}
\end{figure*}

\begin{figure*}[hbt]
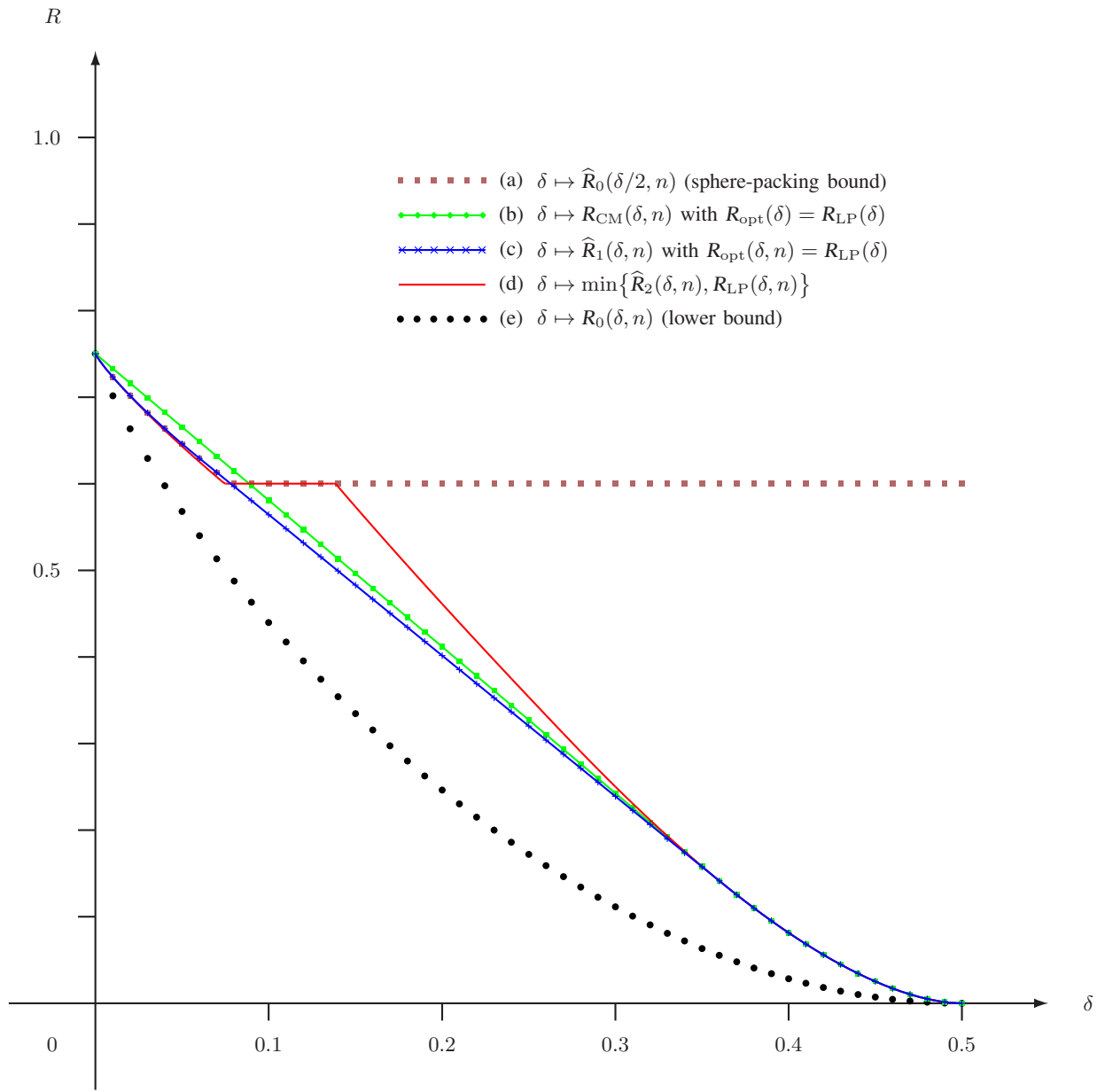

%% q = 2, n =  4:
\Plot{%
    \put(035,095){
	\put(000,000){\LegendAlt}
    }
    \put(000,000){\color{brown}
        \put(0.000,75.000){\Bbox}\put(2.000,72.334){\Bbox}
        \put(4.000,70.171){\Bbox}\put(6.000,68.201){\Bbox}
        \put(8.000,66.356){\Bbox}\put(10.000,64.606){\Bbox}
        \put(12.000,62.933){\Bbox}\put(14.000,61.323){\Bbox}
        \put(16.000,60.000){\Bbox}\put(18.000,60.000){\Bbox}
        \put(20.000,60.000){\Bbox}\put(22.000,60.000){\Bbox}
        \put(24.000,60.000){\Bbox}\put(26.000,60.000){\Bbox}
        \put(28.000,60.000){\Bbox}\put(30.000,60.000){\Bbox}
        \put(32.000,60.000){\Bbox}\put(34.000,60.000){\Bbox}
        \put(36.000,60.000){\Bbox}\put(38.000,60.000){\Bbox}
        \put(40.000,60.000){\Bbox}\put(42.000,60.000){\Bbox}
        \put(44.000,60.000){\Bbox}\put(46.000,60.000){\Bbox}
        \put(48.000,60.000){\Bbox}\put(50.000,60.000){\Bbox}
        \put(52.000,60.000){\Bbox}\put(54.000,60.000){\Bbox}
        \put(56.000,60.000){\Bbox}\put(58.000,60.000){\Bbox}
        \put(60.000,60.000){\Bbox}\put(62.000,60.000){\Bbox}
        \put(64.000,60.000){\Bbox}\put(66.000,60.000){\Bbox}
        \put(68.000,60.000){\Bbox}\put(70.000,60.000){\Bbox}
        \put(72.000,60.000){\Bbox}\put(74.000,60.000){\Bbox}
        \put(76.000,60.000){\Bbox}\put(78.000,60.000){\Bbox}
        \put(80.000,60.000){\Bbox}\put(82.000,60.000){\Bbox}
        \put(84.000,60.000){\Bbox}\put(86.000,60.000){\Bbox}
        \put(88.000,60.000){\Bbox}\put(90.000,60.000){\Bbox}
        \put(92.000,60.000){\Bbox}\put(94.000,60.000){\Bbox}
        \put(96.000,60.000){\Bbox}\put(98.000,60.000){\Bbox}
        \put(100.000,60.000){\Bbox}
    \color{black}
        \put(0.000,75.000){\Bfull}\put(2.000,70.171){\Bfull}
        \put(4.000,66.356){\Bfull}\put(6.000,62.933){\Bfull}
        \put(8.000,59.769){\Bfull}\put(10.000,56.802){\Bfull}
        \put(12.000,53.995){\Bfull}\put(14.000,51.324){\Bfull}
        \put(16.000,48.772){\Bfull}\put(18.000,46.327){\Bfull}
        \put(20.000,43.978){\Bfull}\put(22.000,41.717){\Bfull}
        \put(24.000,39.540){\Bfull}\put(26.000,37.439){\Bfull}
        \put(28.000,35.412){\Bfull}\put(30.000,33.455){\Bfull}
        \put(32.000,31.564){\Bfull}\put(34.000,29.737){\Bfull}
        \put(36.000,27.972){\Bfull}\put(38.000,26.267){\Bfull}
        \put(40.000,24.620){\Bfull}\put(42.000,23.030){\Bfull}
        \put(44.000,21.496){\Bfull}\put(46.000,20.016){\Bfull}
        \put(48.000,18.591){\Bfull}\put(50.000,17.218){\Bfull}
        \put(52.000,15.898){\Bfull}\put(54.000,14.630){\Bfull}
        \put(56.000,13.414){\Bfull}\put(58.000,12.249){\Bfull}
        \put(60.000,11.136){\Bfull}\put(62.000,10.074){\Bfull}
        \put(64.000,9.063){\Bfull}\put(66.000,8.104){\Bfull}
        \put(68.000,7.197){\Bfull}\put(70.000,6.341){\Bfull}
        \put(72.000,5.537){\Bfull}\put(74.000,4.786){\Bfull}
        \put(76.000,4.088){\Bfull}\put(78.000,3.443){\Bfull}
        \put(80.000,2.851){\Bfull}\put(82.000,2.314){\Bfull}
        \put(84.000,1.832){\Bfull}\put(86.000,1.405){\Bfull}
        \put(88.000,1.034){\Bfull}\put(90.000,0.719){\Bfull}
        \put(92.000,0.461){\Bfull}\put(94.000,0.259){\Bfull}
        \put(96.000,0.115){\Bfull}\put(98.000,0.029){\Bfull}
        \put(100.000,0.000){\Bfull}
    \color{green}
        \qbezier(0.000,75.000)(0.889,74.248)(2.000,73.309)
        \qbezier(2.000,73.309)(3.480,72.057)(4.000,71.618)
        \qbezier(4.000,71.618)(4.317,71.350)(6.000,69.927)
        \qbezier(6.000,69.927)(7.539,68.626)(8.000,68.236)
        \qbezier(8.000,68.236)(9.000,67.390)(10.000,66.545)
        \qbezier(10.000,66.545)(11.000,65.699)(12.000,64.854)
        \qbezier(12.000,64.854)(13.000,64.008)(14.000,63.163)
        \qbezier(14.000,63.163)(15.000,62.317)(16.000,61.471)
        \qbezier(16.000,61.471)(17.000,60.626)(18.000,59.780)
        \qbezier(18.000,59.780)(19.000,58.935)(20.000,58.089)
        \qbezier(20.000,58.089)(21.058,57.195)(22.000,56.398)
        \qbezier(22.000,56.398)(23.386,55.226)(24.000,54.707)
        \qbezier(24.000,54.707)(24.858,53.982)(26.000,53.016)
        \qbezier(26.000,53.016)(27.000,52.171)(28.000,51.325)
        \qbezier(28.000,51.325)(29.000,50.480)(30.000,49.634)
        \qbezier(30.000,49.634)(31.000,48.788)(32.000,47.943)
        \qbezier(32.000,47.943)(33.000,47.097)(34.000,46.252)
        \qbezier(34.000,46.252)(34.785,45.588)(36.000,44.561)
        \qbezier(36.000,44.561)(37.000,43.715)(38.000,42.870)
        \qbezier(38.000,42.870)(39.094,41.945)(40.000,41.179)
        \qbezier(40.000,41.179)(41.000,40.333)(42.000,39.488)
        \qbezier(42.000,39.488)(43.259,38.423)(44.000,37.796)
        \qbezier(44.000,37.796)(45.000,36.951)(46.000,36.105)
        \qbezier(46.000,36.105)(47.000,35.260)(48.000,34.414)
        \qbezier(48.000,34.414)(49.000,33.569)(50.000,32.723)
        \qbezier(50.000,32.723)(51.531,31.429)(52.000,31.032)
        \qbezier(52.000,31.032)(53.000,30.187)(54.000,29.341)
        \qbezier(54.000,29.341)(55.000,28.496)(56.000,27.650)
        \qbezier(56.000,27.650)(57.000,26.805)(58.000,25.959)
        \qbezier(58.000,25.959)(59.000,25.113)(60.000,24.268)
        \qbezier(60.000,24.268)(61.000,23.422)(62.000,22.577)
        \qbezier(62.000,22.577)(62.911,21.807)(64.000,20.886)
        \qbezier(64.000,20.886)(65.543,19.581)(66.000,19.195)
        \qbezier(66.000,19.195)(67.000,18.349)(68.000,17.504)
        \qbezier(68.000,17.504)(69.445,16.282)(70.000,15.813)
        \qbezier(70.000,15.813)(71.153,14.838)(72.000,14.144)
        \qbezier(72.000,14.144)(73.005,13.321)(74.000,12.539)
        \qbezier(74.000,12.539)(75.005,11.749)(76.000,11.002)
        \qbezier(76.000,11.002)(77.005,10.247)(78.000,9.537)
        \qbezier(78.000,9.537)(79.021,8.808)(80.000,8.147)
        \qbezier(80.000,8.147)(81.022,7.457)(82.000,6.838)
        \qbezier(82.000,6.838)(83.007,6.200)(84.000,5.615)
        \qbezier(84.000,5.615)(85.006,5.022)(86.000,4.484)
        \qbezier(86.000,4.484)(87.007,3.938)(88.000,3.451)
        \qbezier(88.000,3.451)(89.031,2.946)(90.000,2.527)
        \qbezier(90.000,2.527)(91.033,2.080)(92.000,1.719)
        \qbezier(92.000,1.719)(93.013,1.342)(94.000,1.041)
        \qbezier(94.000,1.041)(95.018,0.732)(96.000,0.509)
        \qbezier(96.000,0.509)(97.029,0.277)(98.000,0.147)
        \qbezier(98.000,0.147)(99.094,0.002)(100.000,0.000)
        \put(0.000,75.000){\Bbox}\put(2.000,73.309){\Bbox}
        \put(4.000,71.618){\Bbox}\put(6.000,69.927){\Bbox}
        \put(8.000,68.236){\Bbox}\put(10.000,66.545){\Bbox}
        \put(12.000,64.854){\Bbox}\put(14.000,63.163){\Bbox}
        \put(16.000,61.471){\Bbox}\put(18.000,59.780){\Bbox}
        \put(20.000,58.089){\Bbox}\put(22.000,56.398){\Bbox}
        \put(24.000,54.707){\Bbox}\put(26.000,53.016){\Bbox}
        \put(28.000,51.325){\Bbox}\put(30.000,49.634){\Bbox}
        \put(32.000,47.943){\Bbox}\put(34.000,46.252){\Bbox}
        \put(36.000,44.561){\Bbox}\put(38.000,42.870){\Bbox}
        \put(40.000,41.179){\Bbox}\put(42.000,39.488){\Bbox}
        \put(44.000,37.796){\Bbox}\put(46.000,36.105){\Bbox}
        \put(48.000,34.414){\Bbox}\put(50.000,32.723){\Bbox}
        \put(52.000,31.032){\Bbox}\put(54.000,29.341){\Bbox}
        \put(56.000,27.650){\Bbox}\put(58.000,25.959){\Bbox}
        \put(60.000,24.268){\Bbox}\put(62.000,22.577){\Bbox}
        \put(64.000,20.886){\Bbox}\put(66.000,19.195){\Bbox}
        \put(68.000,17.504){\Bbox}\put(70.000,15.813){\Bbox}
        \put(72.000,14.144){\Bbox}\put(74.000,12.539){\Bbox}
        \put(76.000,11.002){\Bbox}\put(78.000,9.537){\Bbox}
        \put(80.000,8.147){\Bbox}\put(82.000,6.838){\Bbox}
        \put(84.000,5.615){\Bbox}\put(86.000,4.484){\Bbox}
        \put(88.000,3.451){\Bbox}\put(90.000,2.527){\Bbox}
        \put(92.000,1.719){\Bbox}\put(94.000,1.041){\Bbox}
        \put(96.000,0.509){\Bbox}\put(98.000,0.147){\Bbox}
        \put(100.000,0.000){\Bbox}
    \color{red}
        \qbezier(0.000,75.000)(0.350,74.242)(2.000,72.319)
        \qbezier(2.000,72.319)(2.878,71.295)(4.000,70.117)
        \qbezier(4.000,70.117)(4.926,69.144)(6.000,68.086)
        \qbezier(6.000,68.086)(6.945,67.154)(8.000,66.162)
        \qbezier(8.000,66.162)(8.956,65.262)(10.000,64.314)
        \qbezier(10.000,64.314)(10.963,63.439)(12.000,62.524)
        \qbezier(12.000,62.524)(12.967,61.670)(14.000,60.780)
        %%\qbezier(14.000,60.780)(14.906,60.000)(16.000,60.000)
        %> Higher resolution at a non-differentiable point:
        \qbezier(14.000,60.780)(14.264,60.553)(14.500,60.351)
        \qbezier(14.500,60.351)(14.842,60.057)(15.000,60.013)
        \qbezier(15.000,60.013)(15.046,60.000)(15.250,60.000)
        \qbezier(15.250,60.000)(15.375,60.000)(15.500,60.000)
        \qbezier(15.500,60.000)(15.750,60.000)(16.000,60.000)
        %<
        \qbezier(16.000,60.000)(17.000,60.000)(18.000,60.000)
        \qbezier(18.000,60.000)(19.000,60.000)(20.000,60.000)
        \qbezier(20.000,60.000)(21.000,60.000)(22.000,60.000)
        \qbezier(22.000,60.000)(23.000,60.000)(24.000,60.000)
        \qbezier(24.000,60.000)(25.000,60.000)(26.000,60.000)
        %%\qbezier(26.000,60.000)(27.709,60.000)(28.000,59.661)
        %> Higher resolution at a non-differentiable point:
        \qbezier(26.000,60.000)(26.250,60.000)(26.500,60.000)
        \qbezier(26.500,60.000)(26.750,60.000)(27.000,60.000)
        \qbezier(27.000,60.000)(27.250,60.000)(27.500,60.000)
        \qbezier(27.500,60.000)(27.710,60.000)(27.750,59.953)
        \qbezier(27.750,59.953)(27.875,59.807)(28.000,59.661)
        %<
        \qbezier(28.000,59.661)(28.983,58.516)(30.000,57.345)
        \qbezier(30.000,57.345)(30.983,56.213)(32.000,55.057)
        \qbezier(32.000,55.057)(32.984,53.938)(34.000,52.794)
        \qbezier(34.000,52.794)(34.984,51.686)(36.000,50.554)
        \qbezier(36.000,50.554)(36.985,49.457)(38.000,48.336)
        \qbezier(38.000,48.336)(38.985,47.248)(40.000,46.136)
        \qbezier(40.000,46.136)(40.985,45.057)(42.000,43.954)
        \qbezier(42.000,43.954)(42.985,42.883)(44.000,41.787)
        \qbezier(44.000,41.787)(44.985,40.724)(46.000,39.636)
        \qbezier(46.000,39.636)(46.985,38.580)(48.000,37.497)
        \qbezier(48.000,37.497)(48.985,36.448)(50.000,35.371)
        \qbezier(50.000,35.371)(50.984,34.327)(52.000,33.256)
        \qbezier(52.000,33.256)(52.984,32.218)(54.000,31.151)
        \qbezier(54.000,31.151)(55.222,29.868)(56.000,29.063)
        \qbezier(56.000,29.063)(57.007,28.023)(58.000,27.020)
        \qbezier(58.000,27.020)(59.007,26.002)(60.000,25.022)
        \qbezier(60.000,25.022)(61.007,24.029)(62.000,23.074)
        \qbezier(62.000,23.074)(63.007,22.106)(64.000,21.177)
        \qbezier(64.000,21.177)(65.008,20.234)(66.000,19.332)
        \qbezier(66.000,19.332)(67.008,18.417)(68.000,17.544)
        \qbezier(68.000,17.544)(69.008,16.657)(70.000,15.813)
        \qbezier(70.000,15.813)(71.008,14.956)(72.000,14.144)
        \qbezier(72.000,14.144)(73.009,13.318)(74.000,12.539)
        \qbezier(74.000,12.539)(75.009,11.746)(76.000,11.002)
        \qbezier(76.000,11.002)(77.010,10.244)(78.000,9.537)
        \qbezier(78.000,9.537)(79.010,8.815)(80.000,8.147)
        \qbezier(80.000,8.147)(81.011,7.465)(82.000,6.838)
        \qbezier(82.000,6.838)(83.011,6.197)(84.000,5.615)
        \qbezier(84.000,5.615)(85.012,5.018)(86.000,4.484)
        \qbezier(86.000,4.484)(87.013,3.935)(88.000,3.451)
        \qbezier(88.000,3.451)(89.015,2.954)(90.000,2.527)
        \qbezier(90.000,2.527)(91.017,2.085)(92.000,1.719)
        \qbezier(92.000,1.719)(93.020,1.339)(94.000,1.041)
        \qbezier(94.000,1.041)(95.024,0.730)(96.000,0.509)
        \qbezier(96.000,0.509)(97.035,0.276)(98.000,0.147)
        \qbezier(98.000,0.147)(99.099,0.001)(100.000,0.000)
    \color{blue}
        \qbezier(0.000,75.000)(0.363,74.219)(2.000,72.334)
        \qbezier(2.000,72.334)(2.887,71.312)(4.000,70.171)
        \qbezier(4.000,70.171)(4.934,69.214)(6.000,68.201)
        \qbezier(6.000,68.201)(6.954,67.294)(8.000,66.356)
        \qbezier(8.000,66.356)(8.963,65.493)(10.000,64.606)
        \qbezier(10.000,64.606)(10.971,63.776)(12.000,62.933)
        \qbezier(12.000,62.933)(12.219,62.753)(14.000,61.305)
        \qbezier(14.000,61.305)(15.233,60.303)(16.000,59.679)
        \qbezier(16.000,59.679)(17.000,58.865)(18.000,58.052)
        \qbezier(18.000,58.052)(19.000,57.239)(20.000,56.426)
        \qbezier(20.000,56.426)(21.000,55.613)(22.000,54.800)
        \qbezier(22.000,54.800)(23.000,53.987)(24.000,53.174)
        \qbezier(24.000,53.174)(25.000,52.360)(26.000,51.547)
        \qbezier(26.000,51.547)(27.000,50.734)(28.000,49.921)
        \qbezier(28.000,49.921)(29.000,49.108)(30.000,48.295)
        \qbezier(30.000,48.295)(31.000,47.482)(32.000,46.668)
        \qbezier(32.000,46.668)(33.000,45.855)(34.000,45.042)
        \qbezier(34.000,45.042)(35.000,44.229)(36.000,43.416)
        \qbezier(36.000,43.416)(37.000,42.603)(38.000,41.790)
        \qbezier(38.000,41.790)(39.000,40.976)(40.000,40.163)
        \qbezier(40.000,40.163)(41.000,39.350)(42.000,38.537)
        \qbezier(42.000,38.537)(43.000,37.724)(44.000,36.911)
        \qbezier(44.000,36.911)(45.000,36.098)(46.000,35.285)
        \qbezier(46.000,35.285)(47.000,34.471)(48.000,33.658)
        \qbezier(48.000,33.658)(49.000,32.845)(50.000,32.032)
        \qbezier(50.000,32.032)(51.000,31.219)(52.000,30.406)
        \qbezier(52.000,30.406)(53.000,29.593)(54.000,28.779)
        \qbezier(54.000,28.779)(55.000,27.966)(56.000,27.153)
        \qbezier(56.000,27.153)(57.000,26.340)(58.000,25.527)
        \qbezier(58.000,25.527)(59.000,24.714)(60.000,23.901)
        \qbezier(60.000,23.901)(61.000,23.088)(62.000,22.274)
        \qbezier(62.000,22.274)(63.000,21.461)(64.000,20.648)
        \qbezier(64.000,20.648)(65.000,19.835)(66.000,19.022)
        \qbezier(66.000,19.022)(67.000,18.209)(68.000,17.396)
        \qbezier(68.000,17.396)(69.000,16.582)(70.000,15.769)
        \qbezier(70.000,15.769)(71.285,14.725)(72.000,14.143)
        \qbezier(72.000,14.143)(73.182,13.182)(74.000,12.539)
        \qbezier(74.000,12.539)(75.009,11.746)(76.000,11.002)
        \qbezier(76.000,11.002)(77.010,10.244)(78.000,9.537)
        \qbezier(78.000,9.537)(79.010,8.815)(80.000,8.147)
        \qbezier(80.000,8.147)(81.011,7.465)(82.000,6.838)
        \qbezier(82.000,6.838)(83.011,6.197)(84.000,5.615)
        \qbezier(84.000,5.615)(85.012,5.018)(86.000,4.484)
        \qbezier(86.000,4.484)(87.013,3.935)(88.000,3.451)
        \qbezier(88.000,3.451)(89.015,2.954)(90.000,2.527)
        \qbezier(90.000,2.527)(91.017,2.085)(92.000,1.719)
        \qbezier(92.000,1.719)(93.020,1.339)(94.000,1.041)
        \qbezier(94.000,1.041)(95.024,0.730)(96.000,0.509)
        \qbezier(96.000,0.509)(97.035,0.276)(98.000,0.147)
        \qbezier(98.000,0.147)(99.073,0.005)(100.000,0.000)
        \put(0.000,75.000){\Bplus}\put(2.000,72.334){\Bplus}
        \put(4.000,70.171){\Bplus}\put(6.000,68.201){\Bplus}
        \put(8.000,66.356){\Bplus}\put(10.000,64.606){\Bplus}
        \put(12.000,62.933){\Bplus}\put(14.000,61.305){\Bplus}
        \put(16.000,59.679){\Bplus}\put(18.000,58.052){\Bplus}
        \put(20.000,56.426){\Bplus}\put(22.000,54.800){\Bplus}
        \put(24.000,53.174){\Bplus}\put(26.000,51.547){\Bplus}
        \put(28.000,49.921){\Bplus}\put(30.000,48.295){\Bplus}
        \put(32.000,46.668){\Bplus}\put(34.000,45.042){\Bplus}
        \put(36.000,43.416){\Bplus}\put(38.000,41.790){\Bplus}
        \put(40.000,40.163){\Bplus}\put(42.000,38.537){\Bplus}
        \put(44.000,36.911){\Bplus}\put(46.000,35.285){\Bplus}
        \put(48.000,33.658){\Bplus}\put(50.000,32.032){\Bplus}
        \put(52.000,30.406){\Bplus}\put(54.000,28.779){\Bplus}
        \put(56.000,27.153){\Bplus}\put(58.000,25.527){\Bplus}
        \put(60.000,23.901){\Bplus}\put(62.000,22.274){\Bplus}
        \put(64.000,20.648){\Bplus}\put(66.000,19.022){\Bplus}
        \put(68.000,17.396){\Bplus}\put(70.000,15.769){\Bplus}
        \put(72.000,14.143){\Bplus}\put(74.000,12.539){\Bplus}
        \put(76.000,11.002){\Bplus}\put(78.000,9.537){\Bplus}
        \put(80.000,8.147){\Bplus}\put(82.000,6.838){\Bplus}
        \put(84.000,5.615){\Bplus}\put(86.000,4.484){\Bplus}
        \put(88.000,3.451){\Bplus}\put(90.000,2.527){\Bplus}
        \put(92.000,1.719){\Bplus}\put(94.000,1.041){\Bplus}
        \put(96.000,0.509){\Bplus}\put(98.000,0.147){\Bplus}
        \put(100.000,0.000){\Bplus}
    }
}
\caption{Bounds for $q = 2$ and $n = 4$ (not necessarily all-disjoint).}
\label{fig:q=2,n=4-nondisjoint}
\end{figure*}
\fi
\end{document}